\documentclass[11pt]{article}

\usepackage{natbib}

\usepackage{algorithm}
\usepackage{algpseudocode}
\usepackage{xcolor} 

\definecolor{highlightcolor}{RGB}{200, 0, 0}

\usepackage{amsthm}
\usepackage{amssymb}
\usepackage{graphicx}
\usepackage{thmtools,thm-restate}
\graphicspath{ {./images/} }
\usepackage{amsmath}
\usepackage{hyperref}
\usepackage[capitalise]{cleveref}
\usepackage{caption}
\newif\ifrestatingtheorem
\newcommand{\fnonlyonce}[1]{\ifrestatingtheorem
  \else
    \footnote{#1}\fi
}
\usepackage{booktabs}
\usepackage{setspace}

\usepackage[
  left=1in,
  right=1in,
  top=1in,
  bottom=1in
]{geometry}

\newtheorem{definition}{Definition}[section]
\declaretheorem[name=Theorem]{theorem}
\newtheorem{corollary}{Corollary}[theorem]
\declaretheorem[sibling=theorem, name=Lemma]{lemma}
\newtheorem{observation}[theorem]{Observation}

\declaretheorem[sibling=theorem, name=Example]{example}

\usepackage{etoolbox}
\AtBeginEnvironment{definition}{\setstretch{1.0}}
\AtBeginEnvironment{theorem}{\setstretch{1.0}}
\AtBeginEnvironment{corollary}{\setstretch{1.0}}
\AtBeginEnvironment{lemma}{\setstretch{1.0}}
\AtBeginEnvironment{observation}{\setstretch{1.0}}
\AtBeginEnvironment{proposition}{\setstretch{1.0}}
\AtBeginEnvironment{example}{\setstretch{1.0}}
\AtBeginEnvironment{restatable}{\setstretch{1.0}}
\AtBeginEnvironment{algorithmic}{\setstretch{1.0}}

\newcommand{\vutility}[3]{v_{#1,#2}^{Z^{#3}}}
\newcommand{\uutility}[3]{u_{#1,#2}^{Z^{#3}}}

\newcommand{\app}{\mathcal{A}}
\newcommand{\pos}{\mathcal{P}}

\theoremstyle{remark}
\newtheorem*{remark}{Remark}
\AtBeginEnvironment{remark}{\setstretch{1.0}}

\makeatletter
\renewcommand{\fps@algorithm}{h!}
\makeatother

\begin{document}

\title{Efficient Interview Scheduling  for Stable Matching}

\author{Moshe Babaioff\thanks{The Hebrew University of Jerusalem. Email: {\tt moshe.babaioff@mail.huji.ac.il}. Supported  by a Golda Meir Fellowship and the Israel Science Foundation (grant No. 301/24). }
\and Rotem Gil\thanks{The Hebrew University of Jerusalem. Email: {\tt rotem.gil@mail.huji.ac.il}. Supported  by the Israel Science Foundation (grant No. 301/24).}
\and Assaf Romm\thanks{The Hebrew University of Jerusalem, Israel. Email: {\tt Assaf.Romm@mail.huji.ac.il}. Supported  by United States-Israel Binational Science Foundation (Grant 2022417), the Israel Science Foundation (Grant 1796/22).\\ We thank Alon Eden, Noam Nisan and Ran I. Shorrer for their helpful comments.}}

\maketitle

\begin{spacing}{1}
\begin{abstract}\label{sec:abstract}
The study of stable matchings usually relies on the assumption that agents’ preferences over the opposite side are complete and known.
In many real markets, however, preferences might be uncertain and revealed only through costly interactions such as interviews. We show how to reach interim-stable matchings, under which all matched pairs must have interviewed and agents use expected utilities whenever true values remain unknown, while minimizing both the expected 
\emph{number of interviews} and the expected number of \emph{interview rounds}. We introduce two adaptive algorithms that produce interim-stable matchings: one operates sequentially, and another is a hybrid algorithm that begins by scheduling some interviews in parallel and continues sequentially. Focusing on cases where agents are ex-ante indifferent between agents on the other side, we show that the sequential algorithm performs $2$ interviews per agent in expectation.
We complement this by showing that any algorithm that performs less than $2$ interviews per agent, does not always guarantee interim-stability.
We also demonstrate that the hybrid algorithm requires only polylogarithmic expected number of rounds, while still performing only about $2$ interviews per agent in expectation.
Additionally, the interviews scheduled by our algorithms guarantee an interim-stable matching when Deferred-Acceptance is run after all interviews are completed.
\end{abstract} \end{spacing}
\section{Introduction}\label{sec:introduction}

Over the last few decades, matching market design theory has repeatedly proven effective in shaping and reshaping markets \citep{whathavewelearnedfromMD}. Recognizing the importance of stability, clearinghouses based on the Deferred Acceptance algorithm \citep{gale1962college} became the de facto standard for medical residency matching \citep[see, e.g.,][]{roth1999redesign}, controlled school choice \citep[see, e.g.,][]{abdulkadiroglu2005nyc,abdulkadiroglu2005boston}, and in other matching domains \citep[see, e.g.,][]{balinski1999tale,biro2008student,gonczarowski2020matchingisraelimechinotgapyear,romero1998implementation}. However, one issue that still poses a challenge to market designers is the management of pre-market interactions, with interview scheduling arguably being the most prominent example. 

The importance of pre-market interactions, and interviews in particular, arises from their effect on preference formation and reporting behavior in matching systems \citep{top_of_the_batch2022,srh2025}. Otherwise put, and somewhat in contrast to the standard assumptions of most of the matching theory literature, in some markets participants do not fully know their own preferences. Instead, participants often experience ex-ante uncertainty regarding the utility they will derive from each alternative. Some of this uncertainty can be alleviated through interactions that reveal the value of the match to both sides, but consume time and resources.

Considering both the need for stable outcomes and the costs involved in pre-market interactions, we present interview-scheduling algorithms for matching markets that always output an \textit{interim-stable} matching, and do this \textit{adaptively} in order to be \textit{efficient}. Our algorithms are efficient in the sense of minimizing the expected (total) number of interviews, and minimizing the expected number of interviewing rounds.
They are adaptive in the sense that after each round, the next set of interviews is determined based on the information revealed so far. Such a centralized and adaptive approach is natural in some real-life environments in which a centrally coordinated process can substantially reduce inefficiencies compared to decentralized interviewing.\footnote{When candidates accept as many interviews as possible and programs invite broadly, the result is an excessive number of low-yield and often unnecessary
interviews \citep{InterviewHoarding}. In the context of the American medical residency matches, both medical practitioners and academic market designers suggested that implementing a centralized interview-scheduling process could dramatically reduce the total number of interviews \citep{AreWeInterviewingTooMany,Wapnir2021}.}

Finally, our algorithms are interim stable as they reach an outcome in which (i) all matched pairs have conducted an interview, and (ii) given their interim beliefs---defined as revealed values for interviewed pairs and prior expectations for non-interviewed pairs---no pair of agents that are currently unmatched can benefit from deviating from the proposed matching and matching among themselves \citep{ashlagi2025stablematchinginterviews}. Importantly, this relaxation is what enables the mechanism to reach ``stability'' without requiring every match value to be revealed, thus avoiding unnecessary costly interviews.

Most of our results pertain to a model where agents on one side of the market are ex-ante indifferent between agents on the other side, which we call the ex-ante equivalent setting. When both sides satisfy this property, the market is \emph{bilaterally ex-ante equivalent}. This assumption mirrors the random and uniformly-distributed preferences model employed by \citet{doi:10.1137/0402048}, \citet{AshlagiKanoriaLeshno2017}, and many others.

\subsection{Overview of the Results}
How many interviews are necessary and sufficient to reach an interim-stable matching under uncertainty, and how many interview rounds are required to achieve such stability? To answer this, we study two algorithms that extend the classic Deferred Acceptance algorithm (DA) to conduct interviews and construct stable matchings under uncertainty: the first is completely sequential, while the other takes a hybrid approach that begins with parallel interviews and continues sequentially.\footnote{The sequential algorithm can also be viewed as a parallel one with a single interview per batch.}\footnote{Unlike DA, where the outcome is independent of the order in which proposers make their offers \citep{BLUM2002429}, the two approaches may lead to different interim-stable matchings.}

In the sequential approach, interviews are conducted one by one, allowing the algorithm to fully exploit the information gained in each round before deciding on the next interview. Our sequential approach minimizes redundant interviews, since no new interview is conducted unless it may affect the outcome.

In contrast, a parallel (or hybrid) approach allows multiple interviews to be conducted simultaneously, better utilizing limited time resources when interviews must occur within a restricted time window.  
The drawback of parallelism is potential redundancy. To illustrate, consider a setting with applicants and positions, and an instance in which positions are ex-ante equivalent, and agree that applicant $a_1$ is the most preferred applicant, both in terms of ex-ante value and across all possible interview outcomes. If $a_1$ is eventually matched to a position that interviewed other applicants while $a_1$ was interviewing with other positions, then every prior interview that the position conducted was redundant. Intuitively, the sequential approach aims to reduce the expected number of interviews (at the expense of potentially increasing the number of interview rounds), while the parallel one aims to reduce the number of interview rounds (at the expense of potentially increasing the expected number of interviews).

We analyze both approaches under the assumption of bilaterally ex-ante equivalence. Furthermore, for simplicity we assume throughout that each value distribution's median equals its expectation, as is also assumed by \citet{ashlagi2025stablematchinginterviews}. That is, when applicant $a$ interviews with position $p$, there is exactly a $50\%$ chance that $a$'s value for matching with $p$ increases relative to her ex-ante belief, and similarly for the position.\footnote{This assumption determines the constants used in the theorems below. It can be relaxed to the weaker assumption that for each applicant and position, the probabilities of value increase following an interview are uniformly bounded away from both $0$ and $1$.}

Our first main result establishes that only 2 interviews per applicant in expectation is enough
to reach interim stability in bilaterally ex-ante equivalent settings.
\begin{restatable}{theorem}{TwoInterviewsMainResult}\label{thm:2n}
    There exists an algorithm (\cref{algorithm:serial-adaptive}) that for any instance terminates with a matching that is interim stable, and in a bilaterally ex-ante equivalent setting with $n$ applicants and $m\ge n$ positions, the expected number of interviews is $2\cdot n+O(\log^3n)$. Thus, its expected number of interviews per applicant is $2+o(1)$. {In addition, for any sufficiently large $n$, with high probability every agent participates in at most $O(\log n)$ interviews.
} 
\end{restatable}

 We further establish a lower bound (\cref{thm:lower-bound-iid}): in this setting, any algorithm that for every $m$ and $n$
 always guarantees interim stability, must sometimes perform more than $2\cdot {\min\{m,n\}}-2$ interviews. 

We note that while for $n\leq m$ the theorem establishes the upper bound of 
about $2$ interviews per applicant only when  $n$ is large enough, our simulations indicate that the expected number of interviews per agent converges to approximately 2, even for relatively small market sizes (see \cref{sec:simulations}).

This result of two interviews per applicant aligns with the intuition that, in the bilaterally ex-ante equivalent setting, to ensure interim stability, there cannot be both an applicant and a position whose match values fall below their baseline expected values and that did not interview with each other. Furthermore, 
under our ``median equals expectation'' simplifying assumption, an agent encounters a partner whose realized value is above the baseline expected value, after two interviews in expectation. An efficient algorithm chooses one of the sides and tries to bring all agents on that side to this state. When the ex-ante equivalent condition holds for only one side of the market, we still obtain a constant bound (of $4$) on the expected number of interviews per agent. We further prove that there exists an algorithm that uses parallelism and achieves interim stability in an expected $O(\log^3 n)$ interview rounds, while the expected number of interviews is still only $2\cdot n + O(\log^3 n)$. 

\begin{restatable}{theorem}{ParallelMainResult}\label{thm:hybrid}
     There exists an algorithm that conducts interviews in parallel 
     (\cref{algorithm:parallel-bilateral-instantiation}) and always terminates with an interim-stable matching, and in a bilaterally ex-ante equivalent setting with $n$ applicants and $m\ge n$ positions, has the following guarantee:
     its expected number of interview rounds is $O(\log^3 n)$  and its expected number of interviews is $2n +O (\log^3 n)$. {Moreover, when $n = m$, with high probability every agent participates in at most $O(\log n)$ interviews.}
     Additionally, when $m\ge n +  \lceil 10 \log n \rceil $ its  expected number of interview rounds is at most $4+\log n$.\fnonlyonce{Throughout this paper, all logarithms are taken to base 2.}
\end{restatable}

Finally, in some cases, it may be natural to conduct all the required interviews prior to running an independent matching procedure.\footnote{We note that, under a different model, \citet{vohra2024matching} point out inefficiencies that can arise when interviews are conducted in a separate stage prior to matching.} Our results show that, in bilaterally ex-ante equivalent instances, using our algorithms to schedule interviews and running the applicant-proposing DA on the realized interview values yields an interim-stable matching with high probability.

\subsection{Discussion and Open Problems}
Our results demonstrate that, under ex-ante equivalence assumptions, interim stability can be achieved with only a constant number of interviews per agent. Our hybrid algorithm also achieves interim stability with a polylogarithmic number of interview rounds.
Several natural questions remain open.
Most notably, it is unclear whether a linear bound on the expected number of interviews can be achieved in general instances, without ex-ante equivalence assumptions.\footnote{Our simulations indicate that our sequential algorithm has a linear bound on the expected number of interviews in the setting in which applicants ex-ante agree on the order of positions, and positions ex-ante agree on the order of applicants  (\cref{appendix:simulations-exp-sequential}).} Similarly, while we obtain polylogarithmic interview-round complexity in bilaterally ex-ante equivalent markets, it remains open whether such guarantees
can be extended to more general environments.

\subsection{Additional Related Work}
Our work contributes to the literature on two-sided matching markets under preference
uncertainty, where agents do not fully know their preferences ex-ante, and they learn their preferences by costly interactions with the agents on the other side of the market (interviews). Prior research has sought to mitigate the resulting frictions through various mechanisms, including signaling, constraints on application numbers, and the centralized coordination of interviews.

Focusing specifically on such centralized interview coordination, \citet{ashlagi2025stablematchinginterviews} establish that interim-stable matchings can be reached under preference uncertainty using only a polylogarithmic number of interviews per agent, when agents agree on the order of the opposite side. Specifically, they propose an adaptive algorithm that extends DA, which with high probability schedules $O(\log^2 n)$ interviews for each applicant and each position, and produces an interim-stable matching. They also propose a non-adaptive algorithm for tiered random markets, which with high probability achieves interim stability with  $O(\log^3 n)$ interviews per agent. 
To the best of our knowledge, there is no existing literature that studies the problem of minimization of the number of interview rounds. However, in the related context of distributed settings with known preferences, \citet{10.1145/2767386.2767424} propose algorithms that achieve almost stable matchings within a polylogarithmic number of communication rounds.

As mentioned above, signaling is another method used in environments with costly interviews, where agents cannot feasibly interact with all potential partners. By allowing agents to send a limited number of signals indicating interest, these mechanisms facilitate screening and alleviate congestion, and have been shown to improve welfare even in settings with full preference certainty \citep{10.1257/mic.5.2.99,10.1007/978-3-030-04612-5_16}. \citet{allman2025signalinginterviewsrandommatching} study signaling in random matching markets, demonstrating that simple signaling protocols can achieve interim stability with a limited number of interviews. To mitigate the congestion arising from communication, \citet{ashlagi2020clearing} show that when preferences are fully known, stable matchings can be found with low communication complexity.

Another area of research investigates the effects of constraining the total number of interactions. \citet{agarwal2023stablematchingchoosingproposals} show that agents 
can restrict the number of applications without affecting their match.
\citet{cole2025distributedinterviewselectionstable} study interview selection mechanisms that mitigate interview congestion while maintaining a low probability of matching failure, and \citet{beyhaghi_et_al:LIPIcs.ITCS.2021.74,skancke2021welfare} study how limiting interviews makes matching more efficient.
\citet{lee2017interviewing} demonstrate that in a two-stage matching game, consisting of an interview stage followed by DA, the total number of matches is maximized when firms coordinate to interview overlapping sets of candidates.

Reduction in interaction costs, such as the shift to virtual interviews, can also introduce new inefficiencies. \citet{kadam2021interviewing} and \citet{ InterviewHoarding} demonstrate that reduced costs lead to excessive interviews per candidate, which can have a negative effect on the final matches.

Finally, our work also relates to the literature on preference discovery under incomplete information, where agents learn their rankings via limited refinements and stability is evaluated with respect to the underlying actual preferences \citep{10.1145/2492002.2482607, 10.5555/2540128.2540145,Drummond_Boutilier_2014}. Moreover, \citet{bampis_et_al:LIPIcs.APPROX/RANDOM.2024.17} minimize the number of queries required to verify or find stable matchings under one-sided uncertainty, measuring performance via competitive ratios against an offline oracle. Our approach targets two-sided uncertainty where values are drawn from an underlying distribution, aiming to reach interim stability with respect to the beliefs of the agents.

{\paragraph{Organization}
\cref{sec:model} introduces the model and the notations. \cref{sec:serial} presents our sequential adaptive algorithm and establishes bounds on its  expected number of interviews. \cref{sec:parallel} presents our hybrid adaptive algorithm and establishes bounds on the number of interview rounds. \cref{sec:interview-system} identifies conditions for decoupling the interview scheduling from the matching mechanism, and \cref{sec:conclusion} concludes.} Some proofs and simulations are relegated to the appendix. \section{Preliminaries}\label{sec:model}
\subsection{Model}\label{sec:model-inside-preliminaries}
There is a set $\app$ of $n\geq 2$ \emph{applicants}, and a set $\pos$ of $m\ge2$ \emph{positions}.
We use $\app\cup \pos$ to denote the set of \emph{agents} (an agent is either an applicant or a position). 
Each agent has a prior distribution on the non-negative \emph{cardinal value} of a match with each of the agents on the other side of the market.
When an applicant and position meet for an interview, the applicant realizes her value for the match with that position, and the position independently realizes its value for the match with that applicant. 
We assume that the value of not being matched is $0$ (thus, we assume that every agent prefers to be matched over being left unmatched).

Specifically, each applicant $a_i\in \app$ has a continuous prior distribution $F_{i,j}$ over her cardinal value for position $p_j\in \pos$. We assume that $F_{i,j}$ is supported on a bounded set of non-negative numbers. We denote the expected value of $F_{i,j}$ by $V_{i,j}$, and use  
$\rho_{i,j}=\Pr_{v\sim F_{i,j}}[v> V_{i,j}]$ to denote the probability that a value sampled from $F_{i,j}$ is higher than its expected value $V_{i,j}$. 

Similarly, each position $p_j\in \pos$
has a continuous prior distribution $G_{j,i}$ of its cardinal value for applicant $a_i\in \app$. We assume that $G_{j,i}$ is supported on a bounded set of non-negative numbers.
We denote the expected value of $G_{j,i}$ by $U_{j,i}$, and we use $\eta_{j,i}=\Pr_{v\sim G_{j,i}}[v> U_{j,i}]$ to denote the probability that a value sampled from $G_{j,i}$ is higher than its expected value $U_{j,i}$.
{\footnote{We can relax the assumption that $F_{i,j}$ and $G_{j,i}$ are continuous, and instead only require that $\Pr_{v\sim F_{i,j}}[v= V_{i,j}]=0$ and $\Pr_{v\sim G_{j,i}}[v= U_{j,i}]=0$, without affecting any of our results. }}
A \emph{matching-with-interviews instance} (or simply an \emph{instance}) {is a tuple
$I = (\app, \pos, \mathbf F, \mathbf G)$, where $\mathbf F =\{F_{i,j}\}_{i\in [n],j\in [m]}$ and where ${\bf G}=\{G_{j,i}\}_{j\in [m], i\in [n]}$}. 

Throughout the paper, we assume that for all $i\in [n]$ and $j \in[m]$ we have $\rho_{i,j}=\eta_{j,i}=\frac{1}{2}$. 
That is, we assume that for each distribution, the median is the same as the expected value: the probability of the value being above the expected value is the same as the probability of it being below it.\footnote{Note that this is a weaker assumption than assuming the probability is symmetric around the expected value (as in \cite{ashlagi2025stablematchinginterviews}).} 

We say that \emph{applicants view positions as ex-ante equivalent} 
if there exists a constant $V$ such that $V_{i,j} =V$  for every $i\in[n],j\in[m]$. 
Similarly, \emph{positions view applicants as ex-ante equivalent} if there exists a constant $U$ such that $U_{j,i}=U$  for every $i\in[n],j\in[m]$.
We say that a setting is \emph{bilaterally ex-ante equivalent} if applicants view positions as ex-ante equivalent and positions view applicants as ex-ante equivalent. {A special case covered by the definition of \emph{applicants view positions as ex-ante equivalent} 
is the i.i.d.\ environment, where there exists a single distribution $F$ with $F_{i,j}=F$ for all 
$i\in[n],j\in[m]$. In this case, all applicant values for all positions are drawn i.i.d.  
Analogously, \emph{positions view applicants as ex-ante equivalent} whenever 
there exists a constant $U$ such that $U_{j,i} =U$  for every $i\in[n],j\in[m]$. }
We also consider a relaxation where expected values are equal, while tie-breaking rules may differ.
In \cref{appendix:extensions-almost} we show that our main results for ex-ante equivalent setting extend to that case as well.

We often consider more general settings where agents are not necessarily ex-ante equivalent. We say that \emph{positions (ex-ante) agree on the order of applicants} if there is an order over applicants such that
for every position $p_j\in \pos$ and for every $i,i'$ such that $n\geq i'>i\geq 1$, it holds that $U_{j,i'}\leq U_{j,i}$ (so
lower index applicants are ex-ante preferred for every position).\footnote{Without loss of generality, this order is arbitrarily set to be by order of indices.}
Similarly, we say that \emph{applicants  (ex-ante) agree on the order of  positions} if for every applicant $a_i\in \app$ and for every $j,j'$ such that  
$m\geq j'>j\geq 1$, it holds that $V_{i,j'}\leq V_{i,j}$.

{\subsection{Matchings and Interim Stability}}
Given {an instance}, we look for a stable matching in the market. 
A \emph{matching} matches at most one applicant to each position, and at most one position to each applicant. An applicant or a position may stay unmatched. An applicant is matched to a position if and only if the position is matched to that applicant. 
Formally, a matching is a function $\mu:\app\cup \pos\to \app\cup \pos$ satisfying (1) $\forall x\in \app, \mu(x) \in \pos \cup \{x\}$, (2) $\forall x\in\pos, \mu(x) \in \app \cup \{x\}$, and (3) $\mu^2(x) = x$. The case where $\mu(x) = x$ is interpreted as $x$ remaining unmatched. The \emph{empty matching} has $\mu(x)=x$ for every $x\in \app \cup \pos$.

We assume that an applicant cannot be matched to a position without an interview, and that once an applicant and position meet for an interview, the value they assign to the match is fixed forever (so there is no point for them to meet for an interview again). 
An \emph{interview} specifies an applicant $a_i$ and a position $p_j$ which interviews the applicant. 
The \emph{result of the interview} of applicant $a_i$ and position $p_j$ is a pair of samples, one from $F_{i,j}$ (the value of applicant $a_i$ for the match with $p_j$) and the other from $G_{j,i}$ (the value of position $p_j$ for the match $a_i$). 
Formally, an interview (result) is a tuple $(a_i,p_j, v_{i,j}, u_{j,i})\in \app\times \pos\times  \mathbb{R}_{\geq 0}\times  \mathbb{R}_{\geq 0}$ 
such that $v_{i,j}\sim F_{i,j}$ and $u_{j,i}\sim G_{j,i}$. We use $Z$ to denote a set of interviews: a set in which no applicant-position pair appears more than once. We say that a set of interviews \emph{includes an interview of $a_i$ at $p_j$}, if a tuple of the form 
$(a_i, p_j, v, u)$ is in $Z$. In that case, we say that $u$ and $v$ are the results of the interview (realized values).  With a slight abuse of notation, we write $(a_i, p_j)\in Z$ when the set of interviews $Z$ includes an interview of $a_i$ at $p_j$.

\begin{definition}[Interim Utility]
   Fix an instance and a set of interviews $Z$. If $Z$ includes an interview of $a_i$ at $p_j$, and $u$ and $v$ are the results of the interview, then the \emph{interim utility} $\vutility{i}{j}{}$ of applicant $a_i$ for position $p_j$ is $v$, and the \emph{interim utility} $\uutility{j}{i}{}$ of position $p_j$ for applicant $a_i$ is $u$. Otherwise (there was no interview), the interim utilities are equal to the expected utilities according to the priors:
   $\vutility{i}{j}{}= V_{i,j}$ and $\uutility{j}{i}{}=U_{j,i}$, respectively. 
\end{definition}

\begin{definition}[Interim Preference]

For each applicant $a_i$ and a set of interviews $Z$, let $\succ_{a_i}^{Z}$ be the interim preference relation over the set of positions together with the option of remaining unmatched. We say that an applicant $a_i$ \emph{interim prefers position $p_j$ over $p_{j'}$ with respect to $Z$}-and, in short, \emph{interim prefers $p_j$ over $p_{j'}$}-if $\vutility{i}{j}{} > \vutility{i}{j'}{}$. We denote this by $p_j \succ_{a_i}^{Z} p_{j'} $. {In addition, if $\mu(a_i)=a_i$ (i.e., $a_i$ is unmatched), we write $p_j \succ_{a_i}^{Z} \mu(a_i)$ for every position $p_j$.

Similarly, for each position $p_j$ and set of interviews $Z$, let $\succ_{p_j}^{Z}$ denote the interim preference relation over the set of applicants together with the option of remaining unmatched.} We say that a position $p_j$ \emph{interim prefers applicant $a_i$ over $a_{i'}$ with respect to $Z$} if $\uutility{j}{i}{} > \uutility{j}{i'}{} $.
We denote this by $a_i \succ_{p_j}^{Z} a_{i'} $. {If $\mu(p_j)=p_j$ (i.e., $p_j$ is unmatched), we write $a_i \succ_{p_j}^{Z} \mu(p_j)$ for every applicant $a_i$.}

We define the corresponding \emph{weak} interim preference relations 
$\succcurlyeq_{a_i}^{Z}$ and $\succcurlyeq_{p_j}^{Z}$ analogously.
\end{definition}
A blocking pair is a pair of an applicant and a position that are not matched to each other,  and that interim prefer to deviate and match with each other, over staying with their matched partners. 
\begin{definition}[Blocking Pair]
Fix an instance, a set of interviews $Z$ and a matching $\mu$. 
A pair $(a_i,p_j)$ such that $\mu(a_i) \neq p_j$ is a \emph{blocking pair for $\mu$ given $Z$}, if $p_j\succ_{a_i}^Z\mu(a_i)$ as well as $a_i\succ_{p_j}^Z\mu(p_j)$.
\end{definition}

\begin{definition}[Interim-Stable Matching] 
    Fix an instance, a set of interviews $Z$ and a matching $\mu$. The matching $\mu$ is \emph{interim stable given $Z$} if all matched pairs have interviewed each other, and there are no blocking pairs with respect to the interim utilities given $Z$.
\end{definition}

 \section{Minimizing the Number of Interviews}\label{sec:serial}
In this section, we consider the problem of designing an algorithm to minimize the expected number of interviews. 
We present a sequential adaptive algorithm (\cref{algorithm:serial-adaptive}), an algorithm that schedules interviews sequentially, and adaptively  determines the next interview based on the results of prior interviews. 
We show that it always yields an interim-stable matching, and we analyze the expected number of interviews it performs. 
We present two results: First, for any instance in which $m$ positions view applicants as ex-ante equivalent, we prove that the expected number of interviews is at most $4m$. Second, in the bilaterally ex-ante equivalent setting with $n$ applicants and $m\ge n$ positions, we show that the expected number of interviews is $2\cdot n+O(\log^3n)$.  Thus, the average number of interviews an applicant participates in approaches $2$ as the market becomes large.

We complement our positive results with a lower bound. We show that for bilaterally ex-ante equivalent settings any algorithm that for every $n$ and every $m$ always guarantees interim stability, must sometimes perform more than $2\cdot \min\{m,n\}-2$ interviews (\cref{thm:lower-bound-iid}).

The section is organized as follows. 
In \cref{subsec:serial-alg}, we present \cref{algorithm:serial-adaptive} and introduce the basic notation and terminology that will be used throughout the section.
In \cref{subsec:one-sided}, we analyze the case in which positions view applicants as ex-ante equivalent.
In \cref{sec:bilateral}, we introduce the notions required for the bilaterally ex-ante equivalent setting and carry out the analysis leading to the bound on the expected number of interviews in this setting. Finally, in \cref{sec:lower-bounds}, we establish a lower bound on the number of interviews required for interim stability.

\subsection{The Sequential Adaptive Algorithm}\label{subsec:serial-alg}

{In this section we present the sequential adaptive algorithm (\cref{algorithm:serial-adaptive}), which extends applicant-proposing DA  to preference orders induced by agents’ interim beliefs,
but with two key modifications. The first modification arises when DA dictates that an applicant propose to a position, although she has not yet interviewed with the position. In that case, the pair first meets for an interview (unless the position rejects the applicant without an interview), yet no proposal is made at that step.
Interim preferences are then updated, and the DA proposal step is resumed according to the updated preference order. Second, ties are broken in favor of positions whose current tentative match has the lowest interim value.}

\begin{algorithm}
\caption{\textsc{Sequential\_Adaptive\_Algorithm}}\label{algorithm:serial-adaptive}
\begin{algorithmic}[1]
\Statex \textbf{Input:} 
  an instance $I= (\app,\pos, {\bf F}, {\bf G})$;
  \Statex Optional: an initial matching $\mu$;
  and an initial set of interviews $Z$ (both set to be empty if not provided. That is, if not provided,  set $\mu(x)=x$ for every $x\in \app\cup \pos$, and set $Z=\emptyset$).
  \Statex \textbf{Output:} a matching $\mu$.
\State 
 Initialize iteration $t \leftarrow 1$.
\While {$\exists a\in \app$ s.t $\mu(a)=a$ and $a$ has not been rejected by all positions} 
\Comment{iteration $t$}
    \State Let $i^*$ be the smallest index of an unmatched applicant 
    that has not been rejected by all positions.
    \State Let $\pos^*$ be the set of $a_{i^*}$'s most preferred positions from which she has not yet been rejected.\Statex \hspace{\algorithmicindent}Let $p_{j^*} \in \pos^*$ be a position that gives the smallest interim value to its current match.\footnotemark
    \If{$(a_{i^*}, p_{j^*})\notin Z$
  ($a_{i^*}$ did not interview with $p_{j^*}$) 
  and $a_{i^*}\succ_{p_{j^*}}^Z\mu(p_{j^*})$ ($p_{j^*}$ prefers $a_{i^*}$ over its \Statex \hspace{\algorithmicindent}current match)}
        \State Position $p_{j^*}$ interviews applicant $a_{i^*}$: 
        \State $Z \gets Z \cup {(a_{i^*}, p_{j^*}, v, u)}$, where $v$ and $u$ are drawn (independently) from $ F_{i^*,j^*}$ and $G_{j^*,i^*}$\Statex \hspace{\algorithmicindent}\hspace{\algorithmicindent}respectively.
    \Else
        \If {$\mu(p_{j^*})\succeq_{p_{j^*}}^Za_{i^*}$} 
            \State $p_{j^*}$ rejects applicant $a_{i^*}$
        \Else
            \State $p_{j^*}$ rejects applicant $\mu(p_{j^*})$
            \State $\mu(a_{i^*})\leftarrow p_{j^*}$ and $\mu(p_{j^*})\leftarrow a_{i^*}$
        \EndIf
    \EndIf
    \State $Z^t\leftarrow Z$ \Comment{$Z^t$ is saved to be used in the analysis.}
    \State $t\leftarrow t+1$ \Comment{Go to iteration $t+1$}
\EndWhile
\State\Return $\mu$
\end{algorithmic}
\end{algorithm}

\footnotetext{With ``not being matched” having the minimal score (as values are non-negative). If several positions have this minimal value, ties are broken lexicographically.}

We now describe the operation of \cref{algorithm:serial-adaptive}.
In each iteration, the algorithm selects an unmatched applicant that has not been rejected from all positions. It then chooses her most preferred position (among those from which she has not yet been rejected) that has the lowest interim value for its current match. 
If the applicant and that selected position have not yet met for an interview, the position interviews the applicant if the position interim prefers the applicant over its current match. Otherwise, the position rejects the applicant without an interview. If the applicant and the position have already met for an interview, 
then if the position prefers the applicant over its current match, it is tentatively matched to her (and otherwise it rejects her).

\begin{remark}
    We remark that the execution of these algorithms does not require the mechanism to rely on the underlying value distributions. Instead it relies solely on agents' ex-ante expected values
    and the realized post-interviews values.
\end{remark}

\begin{restatable}{lemma}{Stable}\label{lemma:serial-stable}
     When \cref{algorithm:serial-adaptive} is given as input an instance and an empty initial matching $\mu = \emptyset$, it terminates and outputs an interim-stable matching.
\end{restatable}
We prove interim stability by adapting the classic Gale-Shapley argument to interim preferences. We defer the proof to \cref{appendix:stable}.

We introduce the following notion, which will be used throughout the analysis.

\begin{definition} [Interim Like]
    Let $(a_i,p_j)\in Z$, i.e., an applicant and a position that have already met for an interview. We say that \emph{$a_i$ interim likes $p_j$ (over the prior expected value of $V_{i,j}$)} if $\vutility{i}{j}{}>V_{i,j}$. Similarly, we say that \emph{$p_j$ interim likes $a_i$ (over the prior expected value of $U_{j,i}$)} if it holds that $\uutility{j}{i}{}> U_{j,i}$.\newline
    We say that  $a_i$ and $p_j$ \emph{interim like each other}, if $a_i$ interim likes $p_j$ and $p_j$ interim likes $a_i$.
\end{definition}

\subsection{Positions View Applicants as Ex-Ante Equivalent}\label{subsec:one-sided}

In this section, we begin with a simple analysis, before turning to the more challenging analysis of a stronger bound on the number of interviews, presented
in \cref{sec:bilateral}. We analyze \cref{algorithm:serial-adaptive} under the assumption that all positions view applicants as ex-ante equivalent. We allow applicants to have arbitrary preferences. We show that the algorithm terminates with an interim-stable matching and that each position participates in at most four interviews in expectation. Formally:
\begin{theorem}\label{thm:4m}
    Consider an instance in which $m$ positions view $n$ applicants as ex-ante equivalent.  Then
\cref{algorithm:serial-adaptive} terminates with an interim-stable matching, and its adaptive interviewing process performs at most $4m$ interviews in expectation. 
\end{theorem}
Stability follows from \cref{lemma:serial-stable}, and then the theorem immediately follows from \cref{{lem:position-stops-interview-after++}} and \cref{lem:4m} below (see \cref{appendix:serial-applicants-iid} for the proofs).  
In these two claims we show that once a position finds an applicant such that they interim like each other, it stops interviewing (\cref{{lem:position-stops-interview-after++}}), which leads to a constant bound on the expected number of interviews per position (\cref{lem:4m}). 
\begin{restatable}{lemma}{PositionStopsToInterview}
\label{lem:position-stops-interview-after++}
Fix a position $p_j$ and an applicant $a_i$, and assume that in iteration $t$ they have met for an interview and that $a_i$ and $p_j$ interim like each other. 
Then, in iteration $t+1$, applicant $a_i$ will propose to $p_j$, and $p_j$ will accept the proposal. After position $p_j$ accepts $a_i$, position $p_j$ will refuse to participate in any further interviews.
\end{restatable}
This claim implies that each position stops interviewing once it interviews an applicant such that they interim like each other, which happens after at most 4 interviews in expectation:
\begin{restatable}{lemma}{FourInterviews}
\label{lem:4m}
    Consider any instance where $m$ positions view $n$ applicants as ex-ante equivalent.
    Then the expected number of interviews of \cref{algorithm:serial-adaptive} is at most $4m$.
\end{restatable}

While the bound of $4m$ interviews yields a constant expected number of interviews per position, it can still result in a large number of interviews when $m$ is large relative to $n$. Moreover, simulations for bilaterally ex-ante equivalent instances with $n=m$ show that the expected number of interviews per applicant (and position) is close to two rather than four (\cref{appendix:simulations-iid-sequential}). In the next section, we formally show that in bilaterally ex-ante equivalent setting the expected number of interviews is essentially $2n$ when $m\ge n$. This behavior appears both in unbalanced markets ($m > n$) and in the balanced case ($m = n$).

\subsection{Bilaterally Ex-Ante Equivalent Setting}\label{sec:bilateral}

In this section, we analyze the performance of \cref{algorithm:serial-adaptive} under the bilaterally ex-ante equivalent assumption. We establish an improved upper bound, with approximately two interviews per {applicant} in expectation. Unlike the one-sided ex-ante equivalent setting analyzed in the previous section, proving this improved bound is much more challenging. 
In the previous section, we showed that a position stops interviewing once it finds an applicant with whom it interim likes each other. However, even if an applicant and a position interim like each other and become tentatively matched, this match may later be broken if the position receives a proposal from a more preferred applicant. In that case, the newly unmatched applicant may continue interviewing and may eventually propose to another position, potentially displacing yet another applicant from a match they interim like. Such events can trigger a chain of rejections, forcing applicants to continue interviewing despite having already interviewed with a position they interim like each other.

Ruling out the formation of these cascading chains, and showing that they do not arise even in later stages of the algorithm---when fewer positions remain unmatched, is the main technical challenge we need to address. 

We begin by presenting the main theorem for bilaterally ex-ante equivalent setting, and then turn to the
proof outline, which introduces the key technical definitions and describes the overall proof strategy. The full proof can be found in \cref{appendix:serial-all-iid}.

\TwoInterviewsMainResult*
We next highlight some remarks regarding this result: 
\begin{remark}
    In the case where $n > m$, due to the instance being bilaterally ex-ante equivalent, the roles of applicants and positions can be exchanged, and, by \cref{thm:2n}, an interim-stable matching could be reached with an expected number of interviews of $2\cdot m+O(\log^3m)$. Thus, for any bilaterally ex-ante equivalent instance with $n$ applicants and $m$ positions, an interim-stable matching can be reached with expected number of interviews that is essentially $2\cdot \min\{n,m\}$.
\end{remark}

{\begin{remark}
    Although the bound in \cref{thm:2n} is asymptotic, simulations show that the number of interviews per applicant is close to 2, even for very small markets (see \cref{appendix:simulations-iid-sequential}).
\end{remark}}

{\begin{remark}
    We also note that in the bilaterally ex-ante equivalent setting, the tie-breaking rule in line 4 can be changed to first prefer an unmatched position, and then \emph{any} position that does not interim like its current tentative match---rather than necessarily the position with the lowest interim value for its current tentative match---without affecting the results.
\end{remark}}

\subsubsection{Proof of \cref{thm:2n}}\label{sec:2n-proof}
We first present an outline of the proof for \cref{thm:2n}, breaking the proof into several claims (whose proofs can be located in \cref{appendix:serial-all-iid}). Our proof strategy for \cref{thm:2n} relies on dividing {the iterations in} the execution of \cref{algorithm:serial-adaptive} into two phases {(for the sake of analysis)}, {and we call the iteration separating the two phase \emph{the transition point}}. The first phase consists of a prefix of the interviews in which, with high probability, no applicant interviews at a matched position. The analysis of this phase is relatively easy: we show that each applicant participates in at most two interviews in expectation.

The second phase contains the remaining interviews and its analysis is more subtle. In this phase, under a condition that holds with high probability, we prove that the number of interviews is polylogarithmic. We establish this by bounding both the number of unmatched applicants who continue interviewing and the number of interviews they perform, relying on the fact that an unmatched applicant is guaranteed to stop interviewing once she encounters a position such that both she and the position interim like each other.

Finally, we show that the events for which the above bounds hold in both phases indeed occur with high probability, and that the probability of them being violated  
decays sufficiently fast to ensure that the contribution of the worst-case number of interviews to the expectation is negligible. Combining these bounds yields the desired upper bound on the expected number of interviews.

We next move to formalize the proof.
{We begin by defining a general property of the execution that is useful for our analysis of the initial phase. The following definition captures a condition where applicants do not interview excessively without finding a position they interim like.}

\begin{definition} [$k$-Eventual-Happiness for Applicants]
    Consider an {instance} with $n$ applicants and $m$ positions. Fix $k\in \mathbb{N}$.
    We say that \emph{$k$-eventual-happiness for applicants} holds for the execution of \cref{algorithm:serial-adaptive}, if every applicant that has participated in $k$ interviews, interim likes at least one of her interviewing positions.
\end{definition}

The following definition establishes a specific iteration that serves as our boundary between the two phases.  {We use it, assuming that $k$-eventual-happiness for applicants holds, to show that in the first phase no applicant ever interviews at a matched position.}

\begin{definition} [Transition Point]
    Assume there are $n$ applicants and $m$ positions, and fix some $k\in \mathbb{N}$. {Suppose that at some point in the algorithm there are at most $k-1$ unmatched positions (note that if $k > m$ or $m> n + k - 1$, this point is never reached).} Define $t^k_0$ to be the first iteration in which there are exactly $k-1$ unmatched positions at the beginning of the iteration. We refer to $t^k_0$ as the \emph{transition point}
    (between the two phases of the analysis). This means that in iteration $t_0^k - 1$, a match between the $(m - k + 1)$-th applicant that was matched and some position occurred, and it was the first time that $m - k + 1$ applicants were matched.
    For simplicity, when $k$ is clear from the context, we will refer to this iteration as $t_0$.
\end{definition}

While the previous definition focuses on a one-sided guarantee, our analysis of the second phase requires a stronger condition. We define the following property to capture scenarios where applicants and positions interim like each other after a sufficient number of interviews.

\begin{definition} [$k$-Eventual-Mutual-Happiness]
    Consider an {instance} with $n$ applicants and $m$ positions. Fix $k\in \mathbb{N}$.
    We say that \emph{$k$-eventual-mutual-happiness} holds for the execution of \cref{algorithm:serial-adaptive}, if the following two conditions hold:
    \begin{itemize}
        \item every applicant that has participated in $k$ interviews, the applicant and at least one of her interviewing positions interim like each other.
        \item every position that has participated in $k$ interviews, the position and at least one of its interviewing applicants interim like each other. 
    \end{itemize}
    
\end{definition}

Next, we consider the status of the positions at the transition point $t_0$. The following definition concerns the number of positions that do not interim like their match at this specific iteration.

\begin{definition} [$d$-Positions Unhappy]
    Consider an instance with $n$ applicants and $m$ positions. Fix $d,k\in [m]$.
    We say that \emph{$d$-positions unhappy} holds for the execution of \cref{algorithm:serial-adaptive}, if at iteration $t_0^k$ there are at least $d$  
    positions that do not interim like their match (including unmatched positions).
\end{definition}

Combining the concepts above, we define what constitutes a `bad' run of the algorithm for our analytical purposes. We say that a realization is $k$-bad if either the $k$-eventual-mutual-happiness or the $d$-positions unhappy conditions do not hold.

\begin{definition} [$k$-Bad Realization]
    Consider an {instance} with $n$ applicants and $m$ positions. Fix $k\in \mathbb{N}$. 
    A realization is a \textit{``$k$-bad realization''} if at least one of the following does not hold:
    \begin{itemize}
        \item $k$-eventual-mutual-happiness
        \item $d$-positions unhappy for $d = k^3$
    \end{itemize}
\end{definition} 

To formalize the strategy described above, we establish the following four lemmas that analyze the execution across the two phases separated by $t_0$. First, we establish that the expected number of interviews before $t_0$ is $2n$, assuming that $k$-eventual-happiness for applicants holds.
\begin{restatable}{lemma}{InterviewsBeforeTransition}
\label{lem:number-of-interviews-before-t0}
    When $k$-eventual-happiness for applicants holds for the execution of \cref{algorithm:serial-adaptive}, the expected number of interviews conducted up to iteration $t_0$ is at most $2n$.
\end{restatable}
Before $t_0$, under the $k$-eventual-happiness for applicants assumption, an applicant continues interviewing until she meets a position that she interim likes.  Since this occurs with probability $1/2$ in each interview, the expected number of interviews per applicant in this phase is two.  We prove this formally in \cref{subsec:number-of-interviews}. We note that when $m$ is sufficiently larger than $n$,  we set $k=m-n+1$, {and then the algorithm terminates before $t_0$}. In this scenario, the bound from \cref{lem:number-of-interviews-before-t0} applies with high probability. Moreover, the probability that the $k$-eventual-happiness assumption fails is small enough to ensure that the contribution of the worst-case number of interviews to the expectation is negligible, yielding the desired result directly. Otherwise, the execution extends beyond $t_0$, requiring us to bound the subsequent interviews (of the second phase). Subsequently, we show that even after the transition point, provided the realization is not $k$-bad, the number of remaining interviews is at most $k^2\cdot(k-1)$ (which is asymptotically negligible when $k$ is small relative to $n$, {since} $k$ will be chosen to be $\Theta(\log n)$).

\begin{restatable}{lemma}{InterviewsAfterTransition}\label{lem:interviews-after-t0}
    When the realization is not a {``$k$-bad realization,''} the number of interviews after iteration $t_0$ is at most $k^2\cdot(k-1)$.
\end{restatable}
When the realization is not a $k$-bad realization, after $t_0$, at most $k\cdot(k-1)$ applicants continue interviewing, and each of them performs at most $k$ additional interviews, yielding the bound stated in \cref{lem:interviews-after-t0}.  We prove this claim in \cref{subsec:after-t0}. The bounds in the previous two lemmas apply when the realization is not $k$-bad. We now show that this condition holds with high probability for large $n$.

\begin{restatable}{lemma}{KBadProbability}\label{lem:prob-k-bad}
    Let $p_{bad}$ be the probability of a ``$k$-bad realization''. {For every large enough $n$ and $n+\lceil10\log n\rceil-1>m\ge n$, for $k = \lceil10 \log n\rceil$} it holds that $p_{bad}\le{\frac{3}{n^3}}$.
\end{restatable}
We prove \cref{lem:prob-k-bad} in \cref{subsec:probabilities}. 
{Finally, we combine the bounds on the number of interviews in both phases (\cref{lem:number-of-interviews-before-t0,lem:interviews-after-t0}) with the fact that these bounds hold with high probability (\cref{lem:prob-k-bad}). Moreover, the probability that the a realization is $k$-bad is small enough to ensure that the contribution of the worst-case number of interviews to the expectation is negligible. This implies the following lemma that the total expectation is bounded by $2n$, which will be proven formally in \cref{subsec:number-of-interviews}.}

\begin{restatable}{lemma}{TwoInterviews}\label{tmh:iid}
    Consider any bilaterally ex-ante equivalent instance with  $n$ applicants and $m\ge n$ positions.
    The expected number of interviews when using \cref{algorithm:serial-adaptive} is {$2\cdot n+O(\log^3n)$}. Moreover, with probability $1 - O(n^{-3})$, every agent conducts at most $O(\log n)$ interviews.
\end{restatable}
This lemma will be proven in \cref{subsec:number-of-interviews}. This, together with \cref{lemma:serial-stable}, completes the proof of \cref{thm:2n}. 

\subsection{Lower Bound {for the Bilaterally Ex-Ante Equivalent Setting}}\label{sec:lower-bounds}
We have shown that for any setting, \cref{algorithm:serial-adaptive} \emph{always} outputs a matching that is   
interim stable. In the bilaterally ex-ante equivalent setting, we have shown that when $m\ge n$ its expected number of interviews per applicant is $2+o(1)$  (\cref{thm:2n}). In this section, we establish a lower bound on the number of interviews required to guarantee interim stability for the bilaterally ex-ante equivalent setting. Let $c=\min\{n,m\}$ be the length of the short side of the market.\footnote{The short side refers to the side with fewer agents, $\min\{n,m\}$. Thus, when $n \le m$ the short side is the applicants' side, and otherwise it is the positions' side.} {We show that any algorithm that for every $c$ always guarantees interim stability, must sometimes perform more than $2\cdot c-2$  interviews. 
}

\begin{restatable}{proposition}{LowerBoundIid}\label{thm:lower-bound-iid}
    Consider a bilaterally ex-ante equivalent {instance} with $n$ applicants and $m$ positions. 
    For any algorithm that 
    performs no more than $2\cdot\min\{m,n\}-2$ interviews, the probability that the matching is not interim stable is positive.
\end{restatable}

{The intuition for \cref{thm:lower-bound-iid} is as follows.
 Consider an {algorithm} that always performs at most $2\cdot c - 2$ interviews. With positive probability, in some executions} none of the {matched} applicants interim like the first position they interviewed with. Similarly, with the same probability, none of the {matched} positions interim like the first applicant they interviewed. 
{Since the number of interviews is smaller than $2\cdot c$ {by at least $2$}, there exists an applicant-position pair such that each of them had only one interview (not with each other) and strictly prefers any agent they did not interview to the one they did.} {Since each has exactly one interview, each is matched to the agent they interviewed (otherwise the matching is trivially unstable). Thus, that applicant-position pair forms a blocking pair, implying that the output matching may fail to be interim stable with positive probability.}
The formal proof is provided in \cref{appendix:lower-bouds}.

\section{Minimizing the Number of Interview Rounds}\label{sec:parallel}

In this section, we study the problem of designing an interim-stable algorithm that minimizes the expected number of interview \emph{rounds} (on top of minimizing the expected number of interviews), {in the bilaterally ex-ante equivalent setting with $n$ applicants and $m\ge n$ positions}. 
We present 
a ``hybrid  adaptive algorithm'' (\cref{algorithm:parallel-bilateral-instantiation}) that schedules interviews in parallel and, based on the outcomes of previous interviews, adaptively determines the next batch of interviews to be conducted.
The algorithm is hybrid: it begins by conducting interviews in parallel, meaning that in each interview round multiple applicants interview with different positions simultaneously. This first step results in all but at most few applicants (all but $O(\log n)$ of them) being matched. 
It then switches to sequential interviewing in order to complete the matching. The sequential procedure can be viewed as a special case of the parallel framework in which each interview round contains exactly one interview. We show that the algorithm always terminates with an interim-stable matching, and we analyze its expected number of interview rounds.

Our main result is as follows. In the bilaterally ex-ante equivalent setting with $n$ applicants and $m \ge n$ positions, we prove that the algorithm always terminates with an interim-stable matching and that the expected number of interview rounds is $O(\log^3 n)$ and the expected number of interviews is $2\cdot n + O (\log^3 n)$. Moreover, when $m\ge n +  \lceil 10 \log n \rceil -1 $ the expected number of interview rounds is at most {$4+\log n$}.  

{\cref{algorithm:parallel-bilateral-instantiation} transitions from a parallel phase to a sequential phase. This allows us to analyze the algorithm similarly to the sequential one, and use that similarity to bound the number of rounds. Based on simulations, we hypothesize that for this setting, a fully parallel algorithm {(in which, in every loop iteration, all unmatched applicants interview in parallel, and then enter a proposal stage if they wish, until no unmatched applicants remain)}, only uses $O(\log n)$ interview rounds in expectation. See \cref{appendix:simultation-prallel}.}

The rest of the section is organized as follows.
In \cref{subsec:hybrid-alg} we formally describe the algorithm.
In \cref{subsec:hybrid-proof} we introduce the notation used in the analysis, and then establish an upper bound on the expected number of interview rounds that is sufficient for the algorithm to reach interim stability in the bilaterally ex-ante equivalent setting. 

\subsection{The Hybrid Adaptive Algorithm}\label{subsec:hybrid-alg}
We first specify the hybrid adaptive algorithm (\cref{algorithm:parallel-bilateral-instantiation}) and its subroutines, and then show that it always yields an interim-stable matching (\cref{lem:hybrid-stable}). 
Before presenting the pseudo-code, we include a high level description of the algorithm.
The algorithm starts by constructing a subset $\app'$ of unmatched applicants participating in parallel interviews. In each interview round, the algorithm selects a set of pairs, each consisting of an unmatched applicant from $\app'$ and an unmatched position, to conduct interviews simultaneously (Subroutine~\ref{algorithm:pick-interviews-bilateral-equivalent}). After every round, we run a DA step with the preference lists of applicants truncated at the first position with whom they have not yet interviewed (Subroutine~\ref{algorithm:truncated-DA}).
Consequently, in the DA step, applicants are allowed to propose only to positions that appear before {the} first non-interviewed position.
The algorithm runs parallel interview rounds until all applicants in
$\app'$ are matched, at which point it moves to a final 
stage, with interviews done sequentially. In this stage, the remaining unmatched applicants interview one at a time, following the sequential adaptive algorithm (\cref{algorithm:serial-adaptive}), until all applicants are matched.
If, however, at some point during the parallel interviews, an applicant in $\app'$ has no unmatched position left to interview, then the algorithm switches to a fallback stage and does not reach the sequential stage. In the fallback stage, all remaining interviews are conducted, all interview values are realized, and the standard DA procedure is run on the realized preferences (Subroutine~\ref{algorithm:all-interviews-algorithm}). We later show that for any bilaterally  ex-ante equivalent setting,  the fallback stage occurs with very low probability, so the impact on the expected number of interview rounds is negligible.
\begin{algorithm}
\caption{\textsc{Hybrid\_Adaptive\_Algorithm}}
\label{algorithm:parallel-bilateral-instantiation}
\begin{algorithmic}[1]
\Statex \textbf{Input:} an instance $I= (\app,\pos, {\bf F}, {\bf G})${, where $n=|\app|$ and $m=|\pos|$ and $m\ge n$;}
\Statex \textbf{Output:} a matching $\mu$.
\State {let $k\gets \max\{\lceil 10 \log n \rceil,m-n+1\}$ }
\State let $\app' \gets \{\, a_i \in \app \mid 1 \le i \le \min\{n,m-(k-1)\} \,\}$ \Statex\hspace{\algorithmicindent}\Comment{the set of applicants that interview in parallel}
\State let $Z\gets\emptyset$ and $\mu(x)\gets x$ for every $x\in\app\cup\pos$
\State Initialize iteration $t\gets1$.
\While {{$\exists a\in \app'$ s.t $\mu(a)=a$ and $a$ has not been rejected by all positions} }
\Comment{iteration $t$}
\State {let $\app''$ be the set of unmatched applicants in $\app'$ that have not been rejected by all positions}
    \State $\mathcal{M} \gets \textsc{Pick\_Next\_Interviews}(\app'',\pos,\mu, Z)$
    \Comment{see Subroutine~\ref{algorithm:pick-interviews-bilateral-equivalent}}
    \Statex \hspace{\algorithmicindent}\Comment{$\mathcal{M}$ is a set of pairs $(a_i,p_j)$ to be interviewed in this iteration.}
    \If{$|\mathcal{M}|==|\app''|$} \Comment{each applicant in $\app''$ received an interview.}
        \ForAll{$(a_i,p_j) \in \mathcal{M}$ \textbf{in parallel}}
            \State{Position $p_j$ interviews applicant $a_i$, and $Z$ is updated accordingly.}
        \EndFor
    \Else \Comment{The algorithm falls back to running DA on the realized utilities.} 
        \State Conduct every remaining possible interview. \hspace{\algorithmicindent}\Comment{Formal Algorithm in \cref{appendix:all-interviews-algorithm}}
        \State \Return $\textsc{Applicant\_Proposing\_DA}\left(\app,\pos,\{\succ_{a_i}^{Z}\}_{i \in [n]},\{\succ_{p_j}^Z\}_{j \in [m]}\right)$
        {\Comment{see \cref{algorithm:DA}}}
    \EndIf
    
        \State $Z^t\leftarrow Z$ \Comment{$Z^t$ is saved to be used in the analysis.}
    \State $\mu \gets {\textsc{DA\_on\_Applicants\_Truncated\_Interim\_Preferences}}(\app,\pos,\mu,Z)$
    \Statex\hspace{\algorithmicindent}\Comment{see Subroutine~\ref{algorithm:truncated-DA}}
    \State $t\leftarrow t+1$ \Comment{Go to iteration $t+1$}
\EndWhile
    \State $\mu \gets \textsc{Sequential\_Adaptive\_Algorithm}(\app, \pos,\mu, Z)$
    \Comment{see Algorithm~\ref{algorithm:serial-adaptive}}
\State \Return $\mu$
\end{algorithmic}
\end{algorithm}

\begin{remark}
\cref{algorithm:parallel-bilateral-instantiation} uses subroutines that rely on rejections of applicants by positions that have already happened during the algorithm. To simplify the presentation, we do not explicitly pass the lists of rejections as arguments, but rather assume that a matching object keeps track of all rejections that occur during its formation.
\end{remark}

The hybrid algorithm updates the current matching after new interviews are conducted. We use a modified version of DA that acts on preference lists truncated at the first non-interviewed position, ensuring that applicants only propose to positions with whom they have actually interviewed. This procedure is detailed in Subroutine~\ref{algorithm:truncated-DA} (see \cref{appendix:da-trunc-alg}).

Both \cref{algorithm:parallel-bilateral-instantiation} and Subroutine~\ref{algorithm:truncated-DA} use the applicant-proposing DA subroutine.  
For completeness, we formally describe that subroutine (see \cref{algorithm:DA}) in \cref{appendix:da-algo}. 
It is the standard applicant-proposing DA algorithm, applied to a given preference profile and an initial matching, which is taken to be empty if not explicitly provided.
In Subroutine~\ref{algorithm:truncated-DA}, this subroutine is executed on the applicants-truncated interim preference orders, while in \cref{algorithm:parallel-bilateral-instantiation} it is invoked as a fallback on preference relations induced by fully realized values.

Having defined how the matching is updated in the hybrid algorithm, we now turn to the mechanism for selecting which interviews to perform in each round. To pair unmatched applicants with relevant positions,
we use the notation \(\Gamma = ((V_A, V_P); E)\) for a bipartite graph with two vertex sets, one a subset of applicants $V_A\subseteq \app$ and the other a subset of positions $V_P\subseteq \pos$, and edge set \(E \subseteq V_A \times V_P\). Using the bipartite graph, the following selection subroutine identifies a set of simultaneous interviews.

{\floatname{algorithm}{Subroutine}
\begin{algorithm}[h]
\caption{\textsc{Pick\_Next\_Interviews}}
\label{algorithm:pick-interviews-bilateral-equivalent}
\begin{algorithmic}[1]
\Statex \textbf{Input:} 
  an instance $I= (\app'',\pos, {\bf F}, {\bf G})$;
  a current matching $\mu$;
  and a set $Z$ of interviews.  
  \Statex \textbf{Output:} a set $\mathcal{M}$ of applicant-position pairs to interview (where all applicants belong to $\app''$).
\ForAll{$a_i \in \app''$}
    \State let $\pos_i^*$ be the set of $a_i$'s most preferred positions in $\pos$ from which she has not yet been rejected 
\EndFor
\State construct a bipartite graph $\Gamma = ((\app'', \pos^{\mathrm{free}}); E)$ where $\pos^{\mathrm{free}} = \{\, p \in \pos \mid \mu(p) = p \,\}$ \State and for every $a_i\in \app''$ and $p_j\in \pos^{\mathrm{free}}$
\[
(a_i, p_j) \in E \iff p_j \in \pos_i^* \text{ and } (a_i,p_j)\notin Z
\]

\State let $\mathcal{M}$ be a maximum matching in $\Gamma$.
\State \Return $\mathcal{M}$

\end{algorithmic}
\end{algorithm}
}

We begin the analysis by establishing the interim stability of the algorithm. The proof of \cref{lem:hybrid-stable} is deferred to \cref{appendix:proof-stable-hybrid}.
\begin{restatable}{lemma}{HybridStable}
\label{lem:hybrid-stable}
    For any input instance, \cref{algorithm:parallel-bilateral-instantiation} terminates with an interim-stable matching.
\end{restatable}

\subsection{Proof of \cref{thm:hybrid}}
\label{subsec:hybrid-proof}
In this section, we analyze the behavior of \cref{algorithm:parallel-bilateral-instantiation} in settings satisfying the bilaterally ex-ante equivalent assumption. Having established interim stability in \cref{lem:hybrid-stable}, we focus on bounding the expected number of interview rounds by $O(\log^3 n)$. Moreover, we show that when $m\ge n +  \lceil 10 \log n \rceil $, the expected number of interview rounds is at most $4+\log n$. 
We begin by presenting the main theorem, followed by the terminology used in the analysis, and the proof structure. 

\restatingtheoremtrue
\ParallelMainResult*
\restatingtheoremfalse

\cref{algorithm:parallel-bilateral-instantiation} is designed for bilaterally ex-ante equivalent settings with $m\ge n$. 
The algorithm is structured to correspond to the analysis of \cref{algorithm:serial-adaptive} in these settings (analysis presented in \cref{sec:bilateral}). 
Recall that the analysis of \cref{algorithm:serial-adaptive} separated its execution to two phases (one phase for iterations before the transition point $t_0$, and the other for iterations after it). 
To mimic these two phases, \cref{algorithm:parallel-bilateral-instantiation} has two main phases: Phase~1 consists of the main parallel-interviewing stage (the main while-loop), and Phase~2 consists of a sequential phase that follows.    
For Phase~1 to be analogous to the execution of \cref{algorithm:serial-adaptive} prior to the transition point $t_0$, we select a subset $\app'$ of unmatched applicants containing all but at most $k$ applicants (where $k$ is the parameter used in the analysis of \cref{sec:bilateral}, and $k=\lceil 10 \log n \rceil$ is the more challenging case). That phase ends with the matching of all applicants in $\app'$, unless it enters the ``else'' clause in line 13 (in that case the algorithm transitions to a ``fallback phase''---Phase~3---described immediately below). Phase~2 begins once all applicants in $\app'$ are matched, and corresponds to running the sequential adaptive algorithm (\cref{algorithm:serial-adaptive}) starting at the transition point $t_0$, with the goal of completing the match. 
Finally, Phase~3 is invoked if some applicant in $\app'$ has no unmatched position left to interview, an event that occurs with low probability when the setting is a bilaterally ex-ante equivalent setting. It triggers all interviews, but contributes little to the expectation since it is rarely invoked.

We next present the outline of the proof of \cref{thm:hybrid} (full proof in \cref{appendix:hybrid-bilateral}).
Stability of the final matching follows from \cref{lem:hybrid-stable}.
To complete the proof of \cref{thm:hybrid} we need to bound the expected number of interview rounds.
We first show that with high probability, the number of interview rounds during Phase~1 is bounded by {$k$}. We then show that the lemmas established for the sequential algorithm (\cref{algorithm:serial-adaptive}) continue to hold during Phase~1. Consequently, once Phase~2 is reached and the execution coincides with the sequential algorithm, the results of \cref{sec:bilateral} apply.
{The expected number of interview rounds is obtained by conditioning on the event that the bounds for both phases hold, and observing that the contribution of the maximum possible number of interview rounds outside this event to the expected number of rounds, is negligible.} 

To formalize this argument, we establish the following three lemmas.  We start by characterizing the algorithm behavior during Phase~1 under the assumption of $k$-eventual happiness for applicants.

\begin{restatable}{lemma}{HybridPhaseOne}\label{lemma:parallel-matched-stays-matched}
Assume that {$m\ge n$ and} $k$-eventual-happiness for applicants holds for the execution of 
\cref{algorithm:parallel-bilateral-instantiation}.
Then during the execution of Phase~1
the following hold:
\begin{enumerate}
    \item As long as an applicant in $\app'$ is unmatched, she receives an interview with some unmatched position in every iteration.
    \item No unmatched applicant in $\app'$ proposes to any position that is already matched.
    \item Every applicant in $\app'$ matches to the first position she interim likes and stays matched to that position throughout Phase~1.
\end{enumerate}
\end{restatable}

During Phase~1, the $k$-eventual-happiness assumption guarantees that there are always enough unmatched positions relative to the unmatched applicants in $\app'$. Consequently, every unmatched applicant receives an interview with an unmatched position in each iteration. Once an applicant encounters a position she interim likes she wants to propose to it. That position is still unmatched and therefore accepts her proposal immediately and she remains matched to that position for the rest of Phase~1. We prove this formally in \cref{appendix:proof-phase1}.
Building on the properties established in \cref{lemma:parallel-matched-stays-matched}, we can now derive a bound on the {expected} number of interview rounds in Phase~1. Since applicants remain matched once they find a position they interim like, the number of rounds is bounded by $k$ {(see \cref{appendix:parallel-log-iterations} for the proof).}
\begin{restatable}{lemma}{KRounds}\label{lemma:parallel-log-iterations}
If $k$-eventual-happiness for applicants holds throughout the execution of \cref{algorithm:parallel-bilateral-instantiation}, then Phase~1 terminates after at most $k$ interview rounds.
\end{restatable}

Finally, we combine the bounds for Phase~1, with the analysis of Phase~2 and the bound on the probability that a realization is not a $k$-bad realization, both established in  \cref{sec:bilateral}. This yields a bound of $O(\log^3 n)$ interview rounds when $k=O(\log n)$. Moreover, when the number of positions is large, the algorithm does not reach Phase~2, and we prove a bound of $O(\log n)$.

\begin{restatable}{lemma}{LogRounds}\label{lem:log-rounds}
    Consider a bilaterally ex-ante equivalent instance with $n$ applicants and $m \ge n$ positions. The expected number of interview rounds in \cref{algorithm:parallel-bilateral-instantiation} is $O(\log^3 n)$. Moreover, when $m\ge n +  \lceil 10 \log n \rceil $ the expected number of interview rounds is at most {$4+\log n$}.
\end{restatable}

This lemma will be proven in \cref{appendix:log-rounds}. This, together with \cref{lem:hybrid-stable}, and the following lemma that will be proven in \cref{appendix:2-parallel}, completes the proof of \cref{thm:hybrid}.

\begin{restatable}{lemma}{TwoInterviewsParallel}\label{prop:2-interviews}
    Consider the bilaterally ex-ante equivalent setting with  $n$ applicants and $m\ge n$ positions.
    The expected number of interviews in \cref{algorithm:parallel-bilateral-instantiation} is $2\cdot n+O(\log^3n)$. Moreover, when $m=n$ with probability $1 - O(n^{-3})$, every agent conducts at most $O(\log n)$ interviews.
\end{restatable} \section{Decoupling Interviewing and Matching}\label{sec:interview-system}

The algorithms we have presented so far for finding an interim-stable matching handle the problem by interleaving interviews with a variant of the applicant-proposing DA algorithm.  
One can also consider decoupling the interviews from the matching procedure. That is, the interview-scheduling system schedules the interviews, and only after it is done, a separate system runs the standard applicant-proposing DA algorithm {on} the resulting preferences, to find the final matching.

In this section, we identify a condition that guarantees that any algorithm outputting an interim-stable matching can equivalently be viewed as an interview-scheduling system rather than as a matching mechanism. We further show that, in bilaterally ex-ante equivalent settings, this condition applies to our adaptive algorithms, \cref{algorithm:serial-adaptive,algorithm:parallel-bilateral-instantiation}. With high probability,\footnote{That is, when the realization is not $k$-bad.} the interviews they schedule suffice to yield an interim-stable matching when applicant-proposing DA is run on the realized interview values, and no additional interviews are required. When this condition does not hold, one can conduct all remaining interviews, again resulting with the same desired property.

\begin{restatable}{proposition}{GeneralDecouple}\label{prop:opt}
    Consider an instance with $n$ applicants and $m$ positions. Let $Z$ be the set of interviews, and assume that there exists an interim-stable matching $\mu$ such that every applicant $a_i$ strictly interim prefers $\mu(a_i)$ to any position she did not interview with in $Z$. 
    {Then, the matching generated by the applicant-proposing DA algorithm using the applicants' interim utilities as the effective preferences is interim stable.\footnote{Recall that interim stability does not allow matching without an interview. Thus, this in particular says that no additional interviews are ever required.}} 
\end{restatable}
Whenever \cref{prop:opt} applies, the applicant-proposing DA outputs an interim-stable matching using only the preferences restricted to the interviewed positions.
The proof of \cref{prop:opt} is deferred to \cref{appendix:interview-system-generic}. From this proposition, we obtain the following result (whose proof is deferred to \cref{appendix:interview-system-cor}).

\begin{restatable}{corollary}{OptIsStable}\label{cor:opt-interviews}
    Consider an {instance} with $n$ applicants and $m$ positions where applicants view positions as ex-ante equivalent. Assume that there exists an interim-stable matching $\mu$ such that 
    every applicant $a_i$ interim likes $\mu(a_i)$.  {Then the matching 
    generated by the applicant-proposing DA algorithm using the applicants' preferences induced only by the values realized by interviews is interim stable.} 
\end{restatable}

In \cref{appendix:example-counter}, we also provide a counterexample showing that {the condition that every applicant interim likes her final match} is necessary: if there exists an applicant who does not interim like her matched position, then the applicant-{proposing} deferred-acceptance may match her to a position she never interviewed with, and the resulting matching is not interim stable.

Based on this corollary, we consider the following approach. We first run one of our adaptive interview algorithms {(\cref{algorithm:serial-adaptive,algorithm:parallel-bilateral-instantiation})}. If, upon termination, every applicant interim likes her assigned position, then by \cref{cor:opt-interviews}, running the applicant-proposing DA algorithm on the realized interview values yields an interim-stable matching. Otherwise, if there exists an applicant who does not interim like her match, we conduct all remaining interviews. After all interview values are realized, running applicant-proposing DA guarantees an interim-stable matching. The following proposition shows that, with high probability, no additional interviews are required. The proof is deferred to \cref{appendix:interview-system-proof}.

\begin{restatable}{proposition}{Decouple}\label{lem:no-extra-interviews}
Consider a bilaterally ex-ante equivalent instance with $n$ applicants and $m\ge n$ positions. 
With probability at least $1 - O(n^{-3})$, both \cref{algorithm:serial-adaptive,algorithm:parallel-bilateral-instantiation} terminate with a matching in which every applicant interim likes her matched position. In every such case, {the matching obtained by running the applicant-proposing DA algorithm using the applicants' preferences induced only by the values realized by interviews with positions, is interim stable.}
\end{restatable}

 \section{Conclusion}\label{sec:conclusion}

Our main contribution is the design and analysis of two adaptive algorithms that achieve interim stability under preference uncertainty. 
These algorithms can be viewed as an extension of the applicant-proposing DA algorithm.
{For both the sequential adaptive algorithm and the hybrid adaptive algorithm,}
we show that in bilaterally ex-ante equivalent markets with $n$ applicants and $m\ge n$ positions, interim stability can be reached with an expected number of interviews of $2n+O(\log^3 n)$.   Additionally, we  show that any interim-stable algorithm must sometimes perform more than $2n-2$ interviews. 
Under weaker assumptions, when only applicants are ex-ante equivalent, we show that when there are $m$ positions, interim stability can be achieved with at most $4m$ interviews in expectation.

We further demonstrate that allowing parallelism enables a substantial reduction in the number of interview rounds. Our hybrid algorithm achieves interim stability in an expected polylogarithmic number of interview rounds in bilaterally ex-ante equivalent markets with $n$ applicants and $m\ge n$ positions, and in $4+\log n$ rounds when $m$ is sufficiently larger than $n$. 

Beyond these results, our analysis reveals an additional property of our algorithms in the bilaterally ex-ante equivalent setting. With high probability, all applicants interim like their final match. Consequently, the applicant-optimal stable matching computed with respect to the realized values is interim stable, without requiring any additional interviews. This observation highlights a potential decoupling between interview scheduling and the matching process.

Our results raise several natural open questions regarding the generality of our bounds.
First, while we establish a constant expected number of interviews in instances in which at least one side of the market is ex-ante equivalent, it remains open whether a constant interview bound can be achieved in more general environments. In particular, is it possible to guarantee interim stability with $O(n)$ interviews in expectation for any instance?

Second, our hybrid algorithm achieves a polylogarithmic bound on the expected number of interview rounds under the bilaterally ex-ante equivalent assumption. 
It will be interesting to understand whether polylogarithmic interview rounds are sufficient to achieve stability in more general instances, or even for any instance. 
\bibliography{sample}

@article{gale1962college,
  title={{College Admissions and the Stability of Marriage.}},
  author={David Gale and Lloyd S. Shapley},
  journal={{The American Mathematical Monthly}},
volume  = {69},
  number  = {1},
  pages   = {9--15},
  year={1962},
}

@inproceedings{ashlagi2025stablematchinginterviews,
  title={{Stable Matching with Interviews}},
  author={Ashlagi, Itai and Chen, Jiale and Roghani, Mohammad and Saberi, Amin},
  booktitle={16th Innovations in Theoretical Computer Science Conference (ITCS 2025)},
  pages={12:1-12:19},
  year={2025},
}

@article{doi:10.1137/0402048,
author = {Pittel, Boris},
title = {{The Average Number of Stable Matchings}},
journal = {SIAM Journal on Discrete Mathematics},
volume = {2},
number = {4},
pages = {530-549},
year = {1989}}

@article{BLUM2002429,
title = {{“Timing Is Everything” and Marital Bliss}},
author = {Yosef Blum and Uriel G. Rothblum},
journal = {Journal of Economic Theory},
volume = {103},
number = {2},
pages = {429-443},
year = {2002}}

@unpublished{cole2025distributedinterviewselectionstable,
      title={{Distributed Interview Selection for Stable Matching in Large Random Markets}}, 
      author={Richard Cole and Pranav Jangir},
      year={2025},
      eprint={2506.19345},
      archivePrefix={arXiv},
      primaryClass={cs.GT},
      url={https://arxiv.org/abs/2506.19345}, 
  note = {arXiv preprint}

}

@InProceedings{agarwal2023stablematchingchoosingproposals,
       author =	{Agarwal, Ishan and Cole, Richard},
  title =	{{Stable Matching: Choosing Which Proposals to Make}},
  booktitle =	{50th International Colloquium on Automata, Languages, and Programming (ICALP 2023)},
  pages =	{8:1--8:20},
  series =	{Leibniz International Proceedings in Informatics (LIPIcs)},
  year =	{2023},
  volume =	{261},
}

@InProceedings{allman2025signalinginterviewsrandommatching,
author = {Allman, Maxwell and Ashlagi, Itai and Saberi, Amin and Yu, Sophie H.},
title = {{From Signaling to Interviews in Random Matching Markets}},
year = {2025},
booktitle = {Proceedings of the 57th Annual ACM Symposium on Theory of Computing},
pages = {1556–1567},
numpages = {12},
series = {STOC '25}
}

@article{roth1999redesign,
Author = {Roth, Alvin E. and Peranson, Elliott},
Title = {{The Redesign of the Matching Market for American Physicians: Some Engineering Aspects of Economic Design}},
Journal = {American Economic Review},
Volume = {89},
Number = {4},
Year = {1999},
Month = {September},
Pages = {748–780}}

@article{abdulkadiroglu2005boston,
Author = {Abdulkadiroğlu, Atila and Pathak, Parag A. and Roth, Alvin E. and Sönmez, Tayfun},
Title = {{The Boston Public School Match}},
Journal = {American Economic Review},
Volume = {95},
Number = {2},
Year = {2005},
Month = {May},
Pages = {368–371}}

@article{abdulkadiroglu2005nyc,
Author = {Abdulkadiroğlu, Atila and Pathak, Parag A. and Roth, Alvin E.},
Title = {{The New York City High School Match}},
Journal = {American Economic Review},
Volume = {95},
Number = {2},
Year = {2005},
Month = {May},
Pages = {364–367}}

@article{AshlagiKanoriaLeshno2017,
  title        = {{Unbalanced Random Matching Markets: The Stark Effect of Competition}},
  author       = {Itai Ashlagi and Yash Kanoria and Jacob D. Leshno},
  journal      = {Journal of Political Economy},
  volume       = {125},
  number       = {1},
  pages        = {69--98},
  year         = {2017}}

@inproceedings{gonczarowski2020matchingisraelimechinotgapyear,
author = {Gonczarowski, Yannai A. and Nisan, Noam and Kovalio, Lior and Romm, Assaf},
title = {{Matching for the Israeli ``Mechinot" Gap-Year Programs: Handling Rich Diversity Requirements}},
year = {2019},
booktitle = {Proceedings of the 2019 ACM Conference on Economics and Computation},
pages = {321},
numpages = {1},
series = {EC '19}
}

@article{Wapnir2021,
  title={{Explaining a Potential Interview Match for Graduate Medical Education}},
  author={Wapnir, Irene and Ashlagi, Itai and Roth, Alvin E. and Skancke, Erling and Vohra, Akhil and Lo, Irene and Melcher, Marc L.},
  journal={Journal of Graduate Medical Education},
  volume={13},
  number={6},
  pages={764--767},
  year={2021},
  publisher={The Accreditation Council for Graduate Medical Education}
}

@InProceedings{beyhaghi_et_al:LIPIcs.ITCS.2021.74,
  author =	{Beyhaghi, Hedyeh and Tardos, \'{E}va},
  title =	{{Randomness and Fairness in Two-Sided Matching with Limited Interviews}},
  booktitle =	{12th Innovations in Theoretical Computer Science Conference (ITCS 2021)},
  pages =	{74:1--74:18},
  series =	{Leibniz International Proceedings in Informatics (LIPIcs)},
  year =	{2021},
  volume =	{185},
}

@inproceedings{10.5555/2540128.2540145,
author = {Drummond, Joanna and Boutilier, Craig},
title = {{Elicitation and Approximately Stable Matching with Partial Preferences}},
year = {2013},
publisher = {AAAI Press},
booktitle = {Proceedings of the Twenty-Third International Joint Conference on Artificial Intelligence},
pages = {97–105},
numpages = {9},
series = {IJCAI '13}
}

@article{skancke2021welfare,
  title={{Welfare and Strategic Externalities in Matching Markets with Interviews}},
  author={Skancke, Erling},
  journal={Available at SSRN 3960558},
  year={2021}
}

@inproceedings{10.1145/2492002.2482607,
author = {Rastegari, Baharak and Condon, Anne and Immorlica, Nicole and Leyton-Brown, Kevin},
title = {{Two-Sided Matching with Partial Information}},
year = {2013},
booktitle = {Proceedings of the Fourteenth ACM Conference on Electronic Commerce},
pages = {733–750},
numpages = {18},
series = {EC '13}
}

@article{10.1257/mic.5.2.99,
Author = {Coles, Peter and Kushnir, Alexey and Niederle, Muriel},
Title = {{Preference Signaling in Matching Markets}},
Journal = {American Economic Journal: Microeconomics},
Volume = {5},
Number = {2},
Year = {2013},
Month = {May},
Pages = {99–134},
}

@inproceedings{10.1007/978-3-030-04612-5_16,
author = {Jagadeesan, Meena and Wei, Alexander},
title = {{Varying the Number of Signals in Matching Markets}},
year = {2018},
booktitle = {Web and Internet Economics: 14th International Conference, WINE 2018, Oxford, UK, December 15–17, 2018, Proceedings},
pages = {232–245},
numpages = {14},
}

@article{AreWeInterviewingTooMany,
title = {{The Effort and Outcomes of the Pediatric Surgery Match Process: Are We Interviewing Too Many?}},
journal = {Journal of Pediatric Surgery},
volume = {50},
number = {11},
pages = {1954-1957},
year = {2015},
author = {Samir K. Gadepalli and Cynthia D. Downard and Keith A. Thatch and Saleem Islam and Kenneth S. Azarow and Mike K. Chen and Craig W. Lillehei and Pramod S. Puligandla and Marleta Reynolds and John H. Waldhausen and Keith T. Oldham and Max R. Langham and Thomas F. Tracy and Ronald B. Hirschl},
}

@article{whathavewelearnedfromMD,
    author = {Roth, Alvin E.},
    title = {{What Have We Learned from Market Design?}},
    journal = {The Economic Journal},
    volume = {118},
    number = {527},
    pages = {285-310},
    year = {2008},
    month = {02}
}

@article{top_of_the_batch2022,
Author = {Echenique, Federico and González, Ruy and Wilson, Alistair J. and Yariv, Leeat},
Title = {{Top of the Batch: Interviews and the Match}},
Journal = {American Economic Review: Insights},
Volume = {4},
Number = {2},
Year = {2022},
Month = {June},
Pages = {223–38}}

@inproceedings{Drummond_Boutilier_2014, 
title={{Preference Elicitation and Interview Minimization in Stable Matchings}}, 
volume={28},
number={1}, 
booktitle={Proceedings of the AAAI Conference on Artificial Intelligence}, 
author={Drummond, Joanna and Boutilier, Craig},
year={2014},
month={Jun.} }

@article{InterviewHoarding,
	author = {Manjunath, Vikram and Morrill, Thayer},
	title = {{Interview Hoarding}},
	journal = {Theoretical Economics},
	volume = {18},
	number = {2},
	year = {2023},
	pages = {503-527}
}

@unpublished{kadam2021interviewing,
  title={{Interviewing in Matching Markets with Virtual Interviews}},
  author={Kadam, Sangram V},
  year={2021}
}

@unpublished{srh2025,
    author = {Shorrer, Ran I. and Romm, Assaf and Hassidim, Avinatan},
    title = {{Explaining the Too-Good-To-Be-True Puzzle in Two-Sided Matching}},
    note = {Working paper},
year = {2025}
}

@unpublished{vohra2024matching,
  title={{Matching with Costly Interviews: The Benefits of Asynchronous Offers}},
  author={Vohra, Akhil and Yoder, Nathan},
  year={2024},
note = {Working paper}
}

@article{lee2017interviewing,
  title={{Interviewing in Two-Sided Matching Markets}},
  author={Lee, Robin S. and Schwarz, Michael},
  journal={The RAND Journal of Economics},
  volume={48},
  number={3},
  pages={835--855},
  year={2017},
  publisher={Wiley Online Library}
}

@InProceedings{bampis_et_al:LIPIcs.APPROX/RANDOM.2024.17,
  author =	{Bampis, Evripidis and Dogeas, Konstantinos and Erlebach, Thomas and Megow, Nicole and Schl\"{o}ter, Jens and Trehan, Amitabh},
  title =	{{Competitive Query Minimization for Stable Matching with One-Sided Uncertainty}},
  booktitle =	{Approximation, Randomization, and Combinatorial Optimization. Algorithms and Techniques (APPROX/RANDOM 2024)},
  pages =	{17:1--17:21},
  series =	{Leibniz International Proceedings in Informatics (LIPIcs)},
  year =	{2024},
  volume =	{317},
}

@article{ashlagi2020clearing,
  title={{Clearing Matching Markets Efficiently: Informative Signals and Match Recommendations}},
  author={Ashlagi, Itai and Braverman, Mark and Kanoria, Yash and Shi, Peng},
  journal={Management Science},
  volume={66},
  number={5},
  pages={2163--2193},
  year={2020},
  publisher={INFORMS}
}

@inproceedings{10.1145/2767386.2767424,
author = {Ostrovsky, Rafail and Rosenbaum, Will},
title = {{Fast Distributed Almost Stable Matchings}},
year = {2015},
booktitle = {Proceedings of the 2015 ACM Symposium on Principles of Distributed Computing},
pages = {101–108},
numpages = {8},
series = {PODC '15}
}

@article{balinski1999tale,
  title={{A Tale of Two Mechanisms: Student Placement}},
  author={Michel Balinski and Tayfun Sönmez},
  journal={Journal of Economic Theory},
  volume={84},
  number={1},
  pages={73--94},
  year={1999},
  publisher={Elsevier}
}

@techreport{biro2008student,
    author = {Bir{\'o}, P{\'e}ter},
    title = {{Student Admissions in Hungary as Gale and Shapley Envisaged}},
    institution = {University of Glasgow Technical Report TR-2008-291},
    year = {2008}
}

@article{romero1998implementation,
  title={{Implementation of Stable Solutions in a Restricted Matching Market}},
  author={Romero-Medina, Antonio},
  journal={Review of Economic Design},
  volume={3},
  number={2},
  pages={137--147},
  year={1998},
  publisher={Springer}
}
\newpage
\appendix
\section{Deferred Proofs from Section \ref{sec:serial} (Minimizing the Number of Interviews)}\label{appendix:serial}

This section contains the proofs of the positive results presented in \cref{sec:serial}. The proof of the lower bound (\cref{thm:lower-bound-iid}) is deferred to \cref{appendix:lower-bouds}.
\subsection{Deferred Proofs from \cref{subsec:serial-alg} (The Sequential Adaptive Algorithm)\label{appendix:stable}}
\subsubsection{Proof of \cref{lemma:serial-stable}}
\Stable*
\begin{proof}
    We prove interim stability by adapting the classic Gale-Shapley argument, interpreting preferences as interim preferences.  Suppose, toward a contradiction, that $(a_i,p_j)$ is a blocking pair. Since $a_i$ interim-prefers $p_j$ to her final match, she must have proposed to $p_j$ at some iteration.  At some iteration, $p_j$ rejected $a_i$ (with or without an interview) in favor of an applicant it interim preferred to $a_i$.  {Moreover, every subsequent tentative match of $p_j$, including its final match, is to an applicant that $p_j$ interim prefers to its previous tentative match, and hence, by transitivity, is interim preferred to $a_i$.} Because a position's interim value of an applicant never changes after rejecting them or after becoming matched to them, {the preference comparisons made between $p_j$'s tentative matches and $a_i$ remain valid throughout the process up to the final match}. Thus, $(a_i,p_j)$ cannot be a blocking pair, a contradiction.
\end{proof}

{\subsection{Deferred Proofs from \cref{subsec:one-sided} (Positions View Applicants as Ex-Ante Equivalent)}\label{appendix:serial-applicants-iid}
}

\subsubsection{Proof of \cref{lem:position-stops-interview-after++}}\label{appendix:position-stops-interview-after++-proof}
{To prove \cref{lem:position-stops-interview-after++}, we first prove the following three lemmas, which together imply the lemma. The first lemma (\cref{lem:position-stops-interview}) implies that a position matched with an applicant they interim like {each other,} does not conduct any further interviews.} {After that, in \cref{lem:one-sided-propose-accept} we prove that immediately after a position and an applicant interview and interim like each other, they are matched. In \cref{lem:position-stops-interview2} we show that even if, after a position and an applicant interim like each other and are matched, the position later changes its match (possibly to one that does not interim like the position), it still does not conduct any further interviews.}
\begin{lemma}\label{lem:position-stops-interview}
    Let $a_{i^*}$ be the unmatched applicant chosen by {\cref{algorithm:serial-adaptive}} in iteration $t$, and let $p_{j^*}$ be the position picked by the algorithm at line 4. Assume that $(a_{i^*},p_{j^*})\notin Z^t$ (that is, $a_{i^*}$ has not interviewed with $p_{j^*}$). If $p_{j^*}$ is currently matched to an applicant such that $p_{j^*}$ and its match interim like each other, then $p_{j^*}$ will reject $a_{i^*}$ without conducting an interview. 
\end{lemma}
\begin{proof}
    Since $p_{j^*}$ and its match $\mu(p_{j^*})$ interim like each other, it holds that $\mu(p_{j^*})\succ_{p_{j^*}}^{Z^t}a_{i^*}$, and according to the algorithm, position $p_{j^*}$ rejects applicant $a_{i^*}$ without an interview.
\end{proof}

We now prove that immediately after a position and an applicant interview and interim like each other, they are tentatively matched.

\begin{lemma}\label{lem:one-sided-propose-accept} 
     Fix a position $p_j$ and an applicant $a_i$, and assume that at iteration $t$ they have met for an interview ($(a_i,p_j)\in Z^t$) and that $a_i$ and $p_j$ interim like each other. Then, in iteration $t+1$, applicant $a_i$ will propose to $p_j$ and $p_j$ will accept the proposal.
\end{lemma}
\begin{proof}
     Since $p_j$ agreed to interview in iteration $t$, we must have that $a_i\succ^{t-1}_{p_j}\mu(p_j)$, which implies $a_i\succ^{t}_{p_j}\mu(p_j)$ (since $p_j$ interim likes $a_i$). Therefore, $p_j$ will agree to be matched to $a_i$ in iteration $t+1$ if it receives a proposal. If a position $p_j$ and an applicant $a_i$ interview in iteration $t$, then it must hold that
    $a_i$ weakly prefers $p_j$ over any other position that had not yet rejected $a_i$. Moreover, if $a_i$ interim likes $p_j$ then  $\vutility{i}{j}{t}> V_{i,j}= \vutility{i}{j}{t-1}$, so $a_i$ still prefers $p_j$ in iteration $t+1$. Since \cref{algorithm:serial-adaptive} selects the same applicant until she is matched, in iteration $t+1$ applicant $a_i$ proposes to $p_j$. Since $p_j$ accepts, they are tentatively matched.
\end{proof}

The last lemma shows that even if, after a position and an applicant interim like each other and are matched, the position later changes its match (possibly to one that does not interim like the position), it still does not conduct any further interviews. This is the final lemma needed for the proof of \cref{lem:position-stops-interview-after++}.

\begin{lemma}\label{lem:position-stops-interview2}
 Fix a position $p_j$ and an applicant $a_i$, and assume they interim like each other and that they are tentatively matched at iteration $t$. 
 Then $p_j$ will not interview any other applicant after iteration $t$.
\end{lemma}
\begin{proof}
    From \cref{lem:position-stops-interview}, this claim holds as long as $p_j$ and $a_i$ remain matched. Assume that in some iteration $t$, position $p_j$ receives a proposal from an applicant $a_{i^*}$ whom it prefers over $a_i$, i.e. $a_{i^*}\succ_{p_{j}}^{Z^t}a_i$. Since $\uutility{j}{i}{t}>U_{j,i} = U$, it must hold that $\uutility{j}{i^*}{t}> U$. Assume, towards contradiction, that in a later iteration $t'>t$ some applicant $a_{i'}$ such that $(a_{i'},p_j)\notin Z^{t'}$ is chosen by the algorithm, and $a_{i'}$ and $p_j$ meet for an interview. For this to happen, it must hold that $U=U_{j,i'}=u^{Z^{t'}}_{j,i'}>u^{Z^{t'}}_{j,i^*}=u^{Z^{t}}_{j,i^*}>U$ which is a contradiction. 
\end{proof}

From all the lemmas above, we obtain the correctness of the following result.
\PositionStopsToInterview*

\subsubsection{Proof of \cref{lem:4m}}\label{appendix:4m-proof}
\FourInterviews*
\begin{proof}
    Let $X_j$ be a random variable counting the number of interviews performed by position $p_j$.  Clearly, the expected number of interviews in \cref{algorithm:serial-adaptive} is $\sum_{j=1}^m \mathbb{E}[X_j].$ By \cref{lem:position-stops-interview-after++}, each position continues interviewing applicants until it finds one such that they interim like each other.  
    (A position may stop interviewing earlier, but once this condition is met, it stops permanently.)
    Since each interview is a Bernoulli trial with success probability at least {$\left(\frac{1}{2}\right) ^ 2 = \frac{1}{4}$} (since other outcomes may also be considered successful), the expected number of trials until success is at most $4$. Therefore, $\sum_{j=1}^m \mathbb{E}[X_j] \leq 4m.$
\end{proof}
\subsection{Deferred Proofs from \cref{sec:bilateral} (Bilaterally Ex-Ante Equivalent Setting)}\label{appendix:serial-all-iid}
We note that all lemmas established in \cref{appendix:serial-applicants-iid} continue to hold in this setting as well, since the bilaterally ex-ante equivalent setting is a special case of the positions view applicants as ex-ante equivalent setting. 
\subsubsection{General Lemmas}\label{appendix:general-lemmas}

We begin with a basic lemma that will be used throughout the analysis.
\begin{lemma}\label{lem:applicant-like-propose}
    Fix a position $p_j$ and an applicant $a_i$, and assume that at iteration $t$ they have met for an interview ($(a_i,p_j)\in Z^t$) and that $a_i$ interim likes $p_j$. Then, in iteration $t+1$, applicant $a_i$ will propose to $p_j$. Moreover, if $p_j$ is unmatched, it will accept the proposal.
\end{lemma} 
\begin{proof}
    According to \cref{algorithm:serial-adaptive} if a position $p_j$ and an applicant $a_i$ have met for an interview in iteration $t$, then it must hold that
    $a_i$ weakly prefers $p_j$ over any other position $p_{j'}$ that had not yet rejected $a_i$, so for any such $j'\neq j$ it holds that
    \begin{equation}\label{eq:prefer-j}
        \vutility{i}{j}{t-1}\ge \vutility{i}{j'}{t-1}.
    \end{equation} 
    Assume that $a_i$ interim likes $p_j$. From the definition of interim like, it follows that $\vutility{i}{j}{t}>V_{i,j}=\vutility{i}{j}{t-1}$. 
    Then, using \cref{eq:prefer-j}, we have that for every other position $p_{j'}$ that had not yet rejected $a_i$: $\vutility{i}{j}{t}>V_{i,j}=\vutility{i}{j}{t-1}\ge \vutility{i}{j'}{t-1}=\vutility{i}{j'}{t}$. Hence, $p_j$ will be the position picked by the algorithm at line 4 as $p_{j^*}\in \pos^*$ 
    in iteration $t+1$. If $p_j$ is unmatched ($\mu(p_j)=p_j$), then $p_j$ will agree to be matched to $a_i$, as any match is better than being unmatched (any realization is non-negative). 
\end{proof}
{
We continue with another general lemma that will be used throughout the analysis.
\begin{lemma}\label{lem:minus-breaks+}
Let $a_{i^*}$ be the unmatched applicant selected by \cref{algorithm:serial-adaptive} in iteration $t$, and let $p_{j^*}$ be the position chosen at line~4. Suppose that in iteration $t$ applicant $a_{i^*}$ proposes to $p_{j^*}$. If $p_{j^*}$ is currently tentatively matched to an applicant such that they interim like each other, then $a_{i^*}$ does not interim like $p_{j^*}$.
\end{lemma}

\begin{proof}
Once a position becomes tentatively matched with an applicant they interim like each other, it stops conducting further interviews (\cref{lem:position-stops-interview-after++}). Consequently, any proposal that $p_{j^*}$ may receive 
can only come from an applicant whom $p_{j^*}$ interviewed before this tentative match was formed. By \cref{lem:applicant-like-propose}, if an applicant interim likes a position, it proposes to that position in the immediate next iteration. Therefore, if $p_{j^*}$ interviewed $a_{i^*}$ before forming its current
tentative match, and $a_{i^*}$ proposes to $p_{j^*}$ only now, then
$a_{i^*}$ did not propose immediately after that interview, which implies
that $a_{i^*}$ does not interim like $p_{j^*}$.

\end{proof}}

We divide the remainder of the section into 5 sections. \cref{subsec:before-t0} analyzes the phase of the algorithm before the transition point and \cref{subsec:after-t0} analyzes the phase after it. \cref{subsec:probabilities,sec:k-bad} address bounding probabilities, and \cref{subsec:number-of-interviews} focuses on bounding the expected number of interviews. 

\subsubsection{Lemmas on the Execution of \cref{algorithm:serial-adaptive} Prior to the Transition Point}\label{subsec:before-t0}
We  aim to bound the expected number of interviews in the first phase of the algorithm, during which $k$-eventual happiness for applicants holds for the execution of \cref{algorithm:serial-adaptive}. For this purpose, we will prove the following lemmas.

\begin{lemma}\label{lem:applicant-propose-to-unmatched}
Assume that $k$-eventual-happiness for applicants holds for the execution of \cref{algorithm:serial-adaptive}. Then before time $t_0$, an unmatched applicant will not interview or propose to any matched position.
\end{lemma}
\begin{proof}
According to the algorithm, if an applicant is selected by the algorithm and she has to meet for an interview with a position, she will first interview with an unmatched position.  
Since $k$-eventual-happiness holds, each applicant that the algorithm chooses, if she meets for an interview with up to $k$ unmatched positions, she will eventually meet a position in some iteration which she interim likes. 
By \cref{lem:applicant-like-propose}, after such a meeting occurs, the applicant and the position will be tentatively matched. Since there are always $k$ unmatched positions before iteration $t_0$, any applicant selected by the algorithm that has to meet for an interview with a position will initially interview only with unmatched positions.  As proposals occur only after interviews, she will never propose to a matched position.
\end{proof}

\begin{corollary}\label{lem:applicant-matched-to-liked} 
      When $k$-eventual-happiness for applicants holds for the execution of \cref{algorithm:serial-adaptive},  until iteration $t_0$, every applicant that is selected as $a_{i^*}$ (in line 3) is tentatively matched to the first position that she interim likes, and she stays matched to her until $t_0$. 
\end{corollary}
\begin{proof}
    From \cref{lem:applicant-propose-to-unmatched} and \cref{lem:applicant-like-propose}.
\end{proof}

{\begin{lemma}\label{lem:before-to-no-propose-to-no-like}    
    Assume that $k$-eventual-happiness for applicants holds for the execution of \cref{algorithm:serial-adaptive}. 
    Then before time $t_0$, an unmatched applicant will not propose to any position that she does not interim like.
\end{lemma}
\begin{proof}
    Suppose toward contradiction that at some iteration before $t_0$, an unmatched applicant $a_{i^*}$ proposes to a position she does not interim like. This implies that $a_{i^*}$ has already been rejected by every position that $a_{i^*}$ did not interview with. However, an unmatched position never rejects an applicant without an interview. Since there are at least $k$ unmatched positions available before $t_0$, applicant $a_{i^*}$ must have interviewed with at least $k$ positions. Given the $k$-eventual-happiness for applicants assumption, $a_{i^*}$ must have found at least one position she 
    interim likes 
    among those $k$ interviews. By \cref{lem:applicant-like-propose}, $a_{i^*}$ would have been matched to such a position. For $a_{i^*}$ to be currently proposing to a position she does not interim like, she must have been subsequently unmatched from that position. This directly contradicts \cref{lem:applicant-matched-to-liked}, which establishes that before $t_0$, once an applicant is matched to a position they interim like, they remain matched.
\end{proof}}

\begin{corollary}\label{cor:applicant-didnt-interview}
    Assume that $k$-eventual-happiness for applicants holds for the execution of \cref{algorithm:serial-adaptive}. Then, in iteration $t_0$, all the $n-(m-(k-1))$ unmatched applicants $a_{m-k+2},\dots a_{n}$ did not interview with any position, i.e., $\forall m-k+2 \le i\le n, \forall j\in[m]$, it holds that $(a_i,p_j)\notin Z^{t_0}$.
\end{corollary}
\begin{proof}
    According to the algorithm, the same applicant is chosen in each iteration until she is tentatively matched. By \cref{lem:applicant-matched-to-liked}, no matched applicant becomes unmatched before iteration $t_0$. Therefore, by iteration $t_0$, all applicants $a_1,...,a_{m-(k-1)}$ have been matched and the remaining applicants $a_{m-k+2},...,a_n$ have not yet been selected by the algorithm. 
\end{proof}

\begin{remark}
    If $m \ge n + k - 1$, then \cref{cor:applicant-didnt-interview} holds trivially, as the index range $i \ge m - k + 2$ is empty (because $m - k + 1\ge n $).
\end{remark}

\begin{remark}
\cref{lem:applicant-matched-to-liked} implies that if $m \ge n + k - 1$, then under the 
$k$-eventual-happiness for applicants assumption every applicant becomes matched before $t_0$ (and the algorithm stops).  Therefore, in the following analysis of the number of interviews after the transition point, we assume that $n + k - 1 > m $, since otherwise the algorithm never reaches this phase.
\end{remark}
\subsubsection{Proof of \cref{lem:interviews-after-t0}}\label{subsec:after-t0}
We next aim to upper bound the number of interviews after time $t^k_0$ for every realization that is not a \textit{$k$-bad realization} {(\cref{lem:interviews-after-t0}). To this end, we first prove several intermediate claims. In \cref{lem:interviewed-applicant-number} we bound the number of interviewing applicants, and \cref{lem:new-dont-propose-to++,cor:why-++-dont-break,lem:++dont-break,lem:applicant-do-max-k-interviews}
help us bound the number of interviews they perform. Together, this yields the bound claimed in \cref{lem:interviews-after-t0} on the number of interviews after $t_0$ for every realization that is not a \textit{$k$-bad realization}.} 
Note that if a realization is not a \textit{$k$-bad realization}, then $k$-eventual happiness for applicants must hold. Consequently, all the lemmas proved in \cref{subsec:before-t0} {also} apply whenever the realization is not \textit{$k$-bad}.

Denote the last iteration of the algorithm by $T$. The following lemma bounds the number of applicants that interview after $t_0$.

{
\begin{lemma}\label{lem:interviewed-applicant-number}
Assume that a realization is not a \textit{``$k$-bad realization.''} 
Let $r$ denote the number of applicants that interview from iteration $t_0$ onward. Then it holds that $ r\le k \cdot (k-1)$.
\end{lemma}
\begin{proof}
According to the algorithm, if an applicant interviews in some iteration, it must have been interviewed by a position that is unmatched at that iteration. This is because if the interview is with an unmatched position, this is trivially true. Otherwise, if she interviews with a matched position, she must have been previously interviewed by all unmatched positions in earlier rounds. This follows from the tie-breaking rule of \cref{algorithm:serial-adaptive}: in the bilaterally ex-ante setting, the algorithm resolves ties between non-interviewed positions by prioritizing interviews with unmatched positions. Furthermore, since a matched position remains matched until the end of the algorithm, every interview conducted after $t_0$ involves an applicant who has interviewed {at some point} with a position that was unmatched at iteration $t_0$.
For the transition point $t_0$, consider the $k-1$ positions that are unmatched at that iteration. Consider the interviews conducted by these positions. Assume in contradiction that $ r> k \cdot (k-1)$. By the pigeonhole principle, it implies that there exists at least one position out of the $k-1$ positions that has conducted more than $k$ interviews. However, since the $k$-eventual-mutual-happiness condition holds, any position that has interviewed at least $k$ applicants must have interviewed an applicant such that they interim like each other. By \cref{lem:position-stops-interview-after++}, once such an interview occurs the position stops interviewing. This contradicts the assumption that there is a position that conducts more than $k$ interviews, completing the proof. 
\end{proof}
}

To bound the number of interviews, we also want to show that when the realization is not a $k$-bad realization, each applicant stops interviewing once they meet a position they interim like each other (which takes at most $k$ interviews). For this, we prove the following lemmas
(\cref{lem:new-dont-propose-to++,lem:++dont-break,lem:applicant-do-max-k-interviews})

\begin{lemma}\label{lem:new-dont-propose-to++}
    Assume that a realization is not a \textit{``$k$-bad realization.''} For iteration $t\ge t_0$, denote by $Q(t)$ the number of interviews conducted from $t_0$ (including) until $t$ (including) in which the agents interim like each other. Then the following statements hold for every $t\geq t_0$ for which $Q(t)\leq k(k-1)$:
    \begin{enumerate}
        \item Any unmatched applicant chosen by the algorithm in iteration $t$ does not propose to a position she does not interim like. 
        \item If a proposal occurs in iteration $t$, then the number of positions (matched and unmatched) that do not interim like their current match is at least $k^3-Q(t)$. 
    \end{enumerate}
\end{lemma}
\begin{proof}
We prove the lemma using induction on $t$. \newline
\emph{Induction base ($t=t_0$):} 
By definition,  iteration $t_0$ is an iteration in which an  interview occurs. Consequently, both claims hold trivially as there is no proposal made in iteration $t=t_0$.

\emph{Induction step:} Fix any iteration $t\geq t_0$ for which $Q(t)\leq k\cdot(k-1)$.
Assume the following statements hold for
every iteration $t'<t$: 
    \begin{enumerate}
        \item Any unmatched applicant chosen by the algorithm in iteration $t'$ does not propose to a position she does not interim like. 
        \item If a proposal occurs in iteration $t'$, then the number of positions (matched and unmatched) that do not interim like their current match is at least $k^3-Q(t')$ (which is at least $k$).
    \end{enumerate}
We note that since $Q(t)\leq k\cdot(k-1)$, for any $t'<t$ it holds that $Q(t')\leq k\cdot(k-1)$, as $Q$ is non-decreasing. 

Assuming that the induction hypothesis holds for
every iteration $t'<t$, we now prove that it holds for iteration $t$. If no proposal occurs in iteration $t$, both claims are trivially true. Thus, assume a proposal occurs at $t$, so in that case $Q(t)=Q(t-1)$. Let $a_{i^*}$ be the applicant selected by the algorithm in iteration $t$. Denote by $\tilde{t}$ the most recent iteration after $t_0$ but prior to $t$, in which a proposal occurred (or $t_0$ if no such iteration exists). 

We observe that in the case that $\tilde{t}=t_0$, as the 
realization is not a \textit{``$k$-bad realization,''} the number of positions (matched and unmatched) that do not interim like their current match in $t_0$ is at least $k^3$, {which is at least $k^3-Q(t_0)$ as $Q(t_0)\geq 0$}. If $\tilde{t}>t_0$, then this number is at least $k^3-Q(\tilde{t})$ by the induction hypothesis. So for every $\tilde{t}$ (whether it is $t_0$, or the most recent iteration after $t_0$ but prior to $t$ in which a proposal occurred), this number is at least $k^3-Q(\tilde{t})$.

Between $\tilde{t}$ and $t$, at most one interview could have resulted in the agents interim liking each other. That is because, by \cref{lem:one-sided-propose-accept}, such an outcome immediately triggers a proposal in the subsequent iteration. Since $t$ is the first proposal after $\tilde{t}$, no more than one such outcome can occur in this interval. Moreover, if $a_{i^*}$ proposes to a position that does not interim like $a_{i^*}$, then clearly $Q(t)=Q(\tilde{t})$.
We first observe that $a_{i^*}$ was never previously matched to a position where they interim like each other. Once such a match is formed, the applicant can only become unmatched if another applicant who does not interim like that position proposes to it (\cref{lem:minus-breaks+}). However, no such proposals could have ever occurred. This is true  after time $t_0$ by the inductive assumption. Additionally, it holds up to time $t_0$ by  \cref{lem:before-to-no-propose-to-no-like}.
Consequently, one of two cases must hold for $a_{i^*}$:
\begin{itemize}
    \item In iteration $t-1$, applicant $a_{i^*}$ has interviewed with a position, and they interim like each other. In this case, $Q(t) = Q(t-1) = Q(\tilde{t}) + 1$. 
    By \cref{lem:one-sided-propose-accept}, and noting that the match of any position never changes without a proposal, the number of positions that do not interim like their current match is at least $k^3 - Q(\tilde{t}) - 1 = k^3 - Q(t)$. Thus, the second statement holds for iteration $t$.
    The first statement trivially holds for iteration $t$, because $a_{i^*}$ is proposing to the position she just interviewed and liked.
    \item
    {In iteration $t-1$, applicant $a_{i^*}$ did not interview with a position such that they interim like each other. In this case, we claim that
    $a_{i^*}$ has never interviewed with a position such that they interim like each other}.
    {This holds because, as claimed above, $a_{i^*}$ was never matched with such a position in any previous iteration.  
Any such interview would have triggered an immediate tentative match of that kind in the following iteration {(\cref{lem:one-sided-propose-accept})}. Therefore, the absence of such a match necessarily means that if $a_{i^*}$ had not interviewed with such a position prior to the current iteration, she could not have encountered one.} Since the $k$-eventual-mutual-happiness condition holds, $a_{i^*}$ must have participated in fewer than $k$ interviews up to this point. 
As shown above, there are at least $k^3 - Q(\tilde{t}) \ge k$ positions that do not interim like their current match. Since $a_{i^*}$ has conducted fewer than $k$ interviews, there must exist a position $p_j$ such that 
    \begin{itemize}
        \item Position $p_j$ does not interim like their currently matched applicant, position $p_j$
        has not interviewed $a_{i^*}$, and therefore has not rejected $a_{i^*}$. Therefore it willing to interview $a_{i^*}$.
        \item Therefore, $a_{i^*}$ prefers position $p_j$ 
        to all the positions that she has interviewed with and does not interim like. 
    \end{itemize}
    This ensures that if applicant $a_{i^*}$ proposes at iteration $t$, she must interim like that position, satisfying the first statement. Finally, if the position $a_{i^*}$ proposes to does not interim like $a_{i^*}$, the count of positions that do not interim like their current tentative match 
    remains unchanged (position only improve their tentative match), which is by the inductive hypothesis still $k^3 - Q(t)$, since $Q(t) = Q(\tilde{t})$ in this case.
\end{itemize}
\end{proof}

As a consequence of the previous lemma, we obtain the following {corollaries (\cref{cor:why-++-dont-break,cor:interview-is-unmatched})}, which will be used in the proof of \cref{lem:++dont-break}.
\begin{corollary}\label{cor:why-++-dont-break}
    Assume that a realization is not a \textit{``$k$-bad realization.''} For iteration $t\ge t_0$, denote by $Q(t)$ the number of interviews conducted from $t_0$ (including) until $t$ (including) in which the agents interim like each other. Then the following statement hold for every $t\geq t_0$ for which $Q(t)\leq k\cdot(k-1)$.
    There does not exist an applicant and a position who interim like each other after an interview and become tentatively matched (including before $t_0$), but their match is later broken.
\end{corollary}
\begin{proof}
    By \cref{lem:minus-breaks+}, such an event implies that the applicant proposing to $p_j$ does not interim like $p_j$. Yet, \cref{lem:new-dont-propose-to++} ensures no such proposal occurs while $Q(t) \le k\cdot(k-1)$.
\end{proof}

{
\begin{corollary}\label{cor:interview-is-unmatched}
Assume that a realization is not a \textit{``$k$-bad realization.''} For iteration $t\ge t_0$, denote by $Q(t)$ the number of interviews conducted from $t_0$ (including) until $t$ (including) in which the agents interim like each other. Then the following statement hold for every $t\geq t_0$ for which $Q(t)\leq k\cdot(k-1)$. No unmatched applicant proposes to any position between $t_0$ and $t$ without being interviewed with some position in some previous iteration $t'\ge t_0$ (implying that an applicant is unmatched between $t_0$ and $t$ if and only if she interviews between $t_0$ and $t$).
\end{corollary}
\begin{proof}
    Assume toward contradiction that some applicant proposes (and thus must be unmatched) at some $t \ge t_0$, yet her only interviews occurred before $t_0$. If $t = t_0$, then by \cref{cor:applicant-didnt-interview}, there is no proposal. Thus, for any $t > t_0$, if all her interviews occurred before $t_0$, the interview with the position to which she now proposes did not occur in the previous iteration. By \cref{lem:applicant-like-propose}, this implies that the applicant does not interim like the position she is proposing to now, which contradicts \cref{lem:new-dont-propose-to++}.
\end{proof}}

Building on {\cref{lem:new-dont-propose-to++}}, we show that if the realization is not a $k$-bad realization, then each applicant stops interviewing once they meet a position they interim like each other.

\begin{lemma}\label{lem:++dont-break}
     Assume that a realization is not a \textit{``$k$-bad realization.''} 
     Then, if an applicant $a_i$ and a position $p_j$ meet for an interview at iteration $t\in [T]$ and interim like each other, they become tentatively matched at iteration $t+1$ and remain matched until the end of the execution of the algorithm. 
\end{lemma}
\begin{proof}
 {By \cref{lem:one-sided-propose-accept} if an applicant $a_i$ and a position $p_j$ meet for an interview at iteration $t\in [T]$ 
     and interim like each other, they become tentatively matched at iteration $t+1$.}
    By \cref{lem:applicant-matched-to-liked}, any such pair that interviewed before $t_0$, will stay matched till $t_0$.
    We are left to show that no pair that interim like each other will be broken after iteration $t_0$.
    Let $Q(t)$ be the number of interviews conducted from $t_0$ (including) through $t$ (including) in which the agents interim like each other. Assume, towards contradiction, that there exists an applicant $a_i$ and a position $p_j$ who interim like each other after an interview and become tentatively matched (including before $t_0$), but their match is later broken. For any $t$ where $Q(t) \le k\cdot(k-1)$, this contradicts {\cref{cor:why-++-dont-break}.}
    If at $T$, it holds that $Q(T)\le k\cdot(k-1)$, the claim holds for every $t\in[T]$. 
    To complete the proof we prove that it indeed holds that $Q(T)\le k\cdot(k-1)$.

    {Denote the last iteration in which $Q(t)=k\cdot(k-1)$ as $t^*$.  We claim that at the conclusion of $t^*$, no unmatched applicants remain, implying $T = t^*$. We observe that by {\cref{cor:interview-is-unmatched}}, no unmatched applicant proposes between $t_0$ and $t^*$ without an interview occurring in those iterations (implying that an applicant is unmatched between $t_0$ and $t^*$ if and only if it interviews between $t_0$ and $t^*$).
    In that interval, every applicant matched to a position they interim like each other remains matched to it (\cref{cor:why-++-dont-break}). Since $Q(t^*) = k\cdot(k-1)$, exactly $k\cdot(k-1)$ such interviews occurred in that interval each resulting in a match. Consequently, there are $k\cdot(k-1)$ applicants who interviewed at some point after $t_0$ and are now matched at $t^*$.
    By \cref{lem:interviewed-applicant-number}, as the realization is not a \textit{``$k$-bad realization''}, at most $r = k\cdot(k-1)$ applicants interview after $t_0$. Since we have shown that the $k\cdot(k-1)$ applicants interviewed between $t_0$ and $t^*$ (inclusive) (who are exactly those who were unmatched sometime in that interval) and became matched by iteration $t^*$ (and stayed matched), it follows that all applicants who were unmatched in the interval are matched in $t^*$. Thus, no unmatched applicants remain at $t^*$, which implies $t^* = T$ and completes the proof.
    }
\end{proof}

{The following lemma (\cref{lem:applicant-do-max-k-interviews}) shows that when the realization is not a $k$-bad realization, every applicant will participate in at most $k$ interviews, a claim which will be used in the proof of \cref{lem:interviews-after-t0}.}

\begin{lemma}\label{lem:applicant-do-max-k-interviews}
    Assume that a realization is not a \textit{``$k$-bad realization.''} Then every applicant will participate in at most $k$ interviews.
\end{lemma}
\begin{proof}
    Assume towards contradiction that there exists an applicant $a_i$ who participated in more than $k$ interviews. By the $k$-eventual-mutual-happiness 
     condition, after $k$ interviews, $a_i$ must have interviewed with a position $p_j$ such that $a_i$ and $p_j$ interim like each other. 

    Then, by \cref{lem:++dont-break}, applicant $a_i$ would have been tentatively matched to $p_j$ and remained matched until the end of the algorithm. This contradicts the assumption that she continued to interview after her $k$-th interview. Hence, no applicant can have more than $k$ interviews.
\end{proof}

    \InterviewsAfterTransition*
\begin{proof}
By \cref{lem:interviewed-applicant-number}, there are at most $r$ applicants who interview starting from $t_0$ and by \cref{lem:applicant-do-max-k-interviews}, each such applicant participates in at most $k$ interviews. It follows that the total number of interviews after $t_0$ is bounded by $k \cdot r$.
\end{proof}

\subsubsection{Bounding the Probability for a $k$-Bad Realization and $k$-Eventual-Happiness for Applicants}\label{subsec:probabilities}

We next aim to bound the probability that an execution of \cref{algorithm:serial-adaptive} satisfies $k$-eventual-happiness for applicants, as well as the probability that the realization is a $k$-bad realization. {The first two lemmas (\cref{lem:mutual-happiness-app,lem:mutual-happiness-pos}) bound the probabilities of the two first conditions required for the realization to be a $k$-bad realization.}

\begin{lemma}\label{lem:mutual-happiness-app}
    Let $k = \lceil10 \log n\rceil$. Then the probability that an execution of \cref{algorithm:serial-adaptive} satisfies that every applicant that has participated in $k$ interviews, the applicant and at least one of her interviewing positions interim like each other is at least $\left(1 - \frac{1}{n^4} \right)^n$.
\end{lemma}
\begin{proof}
    Assume that before the execution of the algorithm, a table of quantiles is sampled. Each entry in the table contains a pair $(w_{i,q},z_{i,q})$, where both values are drawn independently from the uniform distribution on $[0,1]$. The table has one row for each applicant and $m$ columns.
    When an applicant $ a_i $ interviews with a position $p_j$, instead of sampling the values from the distributions during runtime, we can use the pre-sampled quantiles and convert them into realized values via the inverse CDFs of their true distributions.
    Specifically, if it is the $q$-th interview for applicant $a_i$, then the interview value for $a_i$ is set to be $F_{i,j}^{-1}(w_{i,q})$ and the interview value for $p_j$ is set to be $G_{j,i}^{-1}(z_{i,q})$.
    
    Note that this method is statistically equivalent to sampling values from the distributions $F_{i,j}$ and $G_{j,i}$ during runtime. Therefore, the probability that an execution of \cref{algorithm:serial-adaptive} satisfies that every applicant who has participated in $k$ interviews, the applicant and at least one of her interviewing positions interim like each other is equal to the probability that, for each applicant, among the first $k$ columns there exists at least one quantile pair $(w_{i,q},z_{i,q})$ such that $w_{i,q} > 0.5$ and $z_{i,q} > 0.5$.
    
    Let $A_{i_l}$ for $l\in[k]$ be the event that in the $l$-th interview of applicant $a_i$, $a_i$  and the position interim like each other.  
    Let $A_i=\bigcup_{l=1}^kA_{i_l}$ be the event that one of the $k$ interviews that applicant $a_i$ performs, ends with the position and the applicant interim like each other. Let $B_1=\cap_{i=1}^nA_i$ be the event that every applicant $a_i$ that participates in $k$ interviews with $k$ positions, there exists a position such that the applicant $a_i$ and the position interim like each other. As the results of the interviews are sampled independently for each applicant and each interview: 
    \[
        \begin{split}
            \Pr(B_1)&=\prod _{i=1}^n\Pr(A_i)=\prod_{i=1}^n\left(1-\Pi_{l=1}^k\Pr(\overline{A_{i_l}})\right)
            \\&=\left(1-\left(\frac{3}{4}\right)^k\right)^n=\left(1-\left(\frac{3}{4}\right)^{\lceil10 \log n\rceil}\right)^n\ge\left(1-\left(\frac{3}{4}\right)^{10 \log n}\right)^n
            \\&=\left(1-n^{10\log\left(\frac{3}{4}\right)}\right)^n\ge\left(1-n^{-4}\right)^n
        \end{split}
    \]
\end{proof}

\begin{lemma}\label{lem:mutual-happiness-pos}
    Let $k = \lceil10 \log n\rceil$ and assume that $n+k-1>m\ge n$. Then the probability that an execution of \cref{algorithm:serial-adaptive} satisfies that every position that has participated in $k$ interviews, the position and at least one of its interviewed applicants interim like each other is at least $\left(1 - \frac{1}{n^4} \right)^{n+k-1}$.
\end{lemma}
\begin{proof}
    The proof follows a symmetric argument to the one used in the proof of \cref{lem:mutual-happiness-app}.  
    As in \cref{lem:mutual-happiness-app}, we assume that, prior to the execution of the algorithm, a table of quantiles is sampled and later converted into realized values during runtime. Each entry in the table consists of a pair of quantiles $(w'_{j,q},z'_{j,q})$, where both values are drawn independently from the uniform distribution on $[0,1]$.
    However, in this case, the table has one row for each position and $n$ columns. Similarly, let $B_2$ be the event that for every position $p_j$ that conducts $k$ interviews with $k$ applicants, there exists at least one applicant such that the applicant and $p_j$ interim like each other. Then \[\Pr(B_2)\ge\left(1-n^{-4}\right)^m >\left(1-n^{-4}\right)^{n+k-1}\]
\end{proof}

Similarly to the previous lemmas, the following claims (\cref{lem:happiness-app-for-general-k,cor:happiness-app-for-general-k,lem:happiness-app,cor:k-happ-better-than-good}) bound the probability that an execution of \cref{algorithm:serial-adaptive} satisfies $k$-eventual-happiness for applicants.
\begin{lemma}\label{lem:happiness-app-for-general-k}
     Let $k\le m$ and let $p_D$ be the probability that an execution of \cref{algorithm:serial-adaptive} satisfies $k$-eventual-happiness for applicants. Then $p_D=\left(1-\left(\frac{1}{2}\right)^k\right)^n$.
\end{lemma}
\begin{proof}
    Let $D$ be the event that $k$-eventual-happiness for applicants holds. Similarly to the proof of \cref{lem:mutual-happiness-app}, assume that before the execution of the algorithm, a table of quantiles is sampled. 
    Each entry in the table contains a value that is sampled independently of the uniform distribution on $[0,1]$. 
    The table has one row for each applicant and $m$ columns.
    When an applicant $a_i$ interviews with a position $p_j$, 
    instead of sampling the value from distribution $F_{i,j}$ during runtime, we can use the pre-sampled quantiles and convert them into realized values.
    
    The probability that an execution of \cref{algorithm:serial-adaptive} satisfies $k$-eventual-happiness for applicants equals the probability that, for each applicant, among the first $k$ columns there exists at least one quantile that is bigger than $\frac{1}{2}$.
    Let $E_{i_l}$ for $l\in[k]$ be the event that the value sampled in the $l$-th column of applicant $a_i$ is such that $a_i$ interim likes the position. Let $D_i=\bigcup_{l=1}^kE_{i_l}$ be the event that one of the $k$ interviews that applicant $a_i$ performs, ends with the applicant interim liking the position. Let $D=\bigcap_{i=1}^nD_i$ be the event that for every applicant $a_i$ that participates in $k$ interviews with $k$ positions, there exists at least one position that $a_i$ interim likes. Since the interview values are sampled independently for each applicant and each interview, and for each $i$, $\Pr(D_i)=1-\Pi_{l=1}^k\Pr(\overline{E_{i_l}})=1-\left(\frac{1}{2}\right)^k$:
    \begin{equation*}
    \Pr(D)=\Pi_{i=1}^n\Pr(D_i)=\Pi_{i=1}^n\left(1-\Pi_{l=1}^k\Pr(\overline{D_{i_l}})\right)
            =\left(1-\left(\frac{1}{2}\right)^k\right)^n
    \end{equation*}
\end{proof}

\begin{corollary}\label{cor:happiness-app-for-general-k}
     Let $p_D$ be the probability that an execution of \cref{algorithm:serial-adaptive} satisfies $k$-eventual-happiness for applicants. Then $1-p_D\le\frac{n}{2^k}$
\end{corollary}
\begin{proof}
    By \cref{lem:happiness-app-for-general-k}, using Bernoulli's inequality: 
    \[1-p_D=1-\left(1-\frac{1}{2^k}\right)^n\le 1-\left(1-\frac{n}{2^k}\right)=\frac{n}{2^k}\]
\end{proof}

\begin{corollary}\label{cor:k-happ-better-than-good}
     Let $p_D$ be the probability that an execution of \cref{algorithm:serial-adaptive} satisfies $k$-eventual-happiness for applicants and let $k \ge \lceil10 \log n\rceil$. Then $p_D\ge1-\frac{1}{n^9}$.
\end{corollary}
\begin{proof}
    By \cref{lem:happiness-app-for-general-k} and Bernoulli's inequality, \[p_D\ge\left(1-\frac{1}{2^k}\right)^n\ge\left(1-\frac{1}{n^{10}}\right)^n\ge1-\frac{1}{n^{9}}.\] 
\end{proof}

\begin{lemma}\label{lem:happiness-app}
    Let $k = \lceil10 \log n\rceil$. Let $p_D$ be the probability that an execution of \cref{algorithm:serial-adaptive} satisfies $k$-eventual-happiness for applicants. Then $p_D\ge\left(1 - \frac{1}{n^{10}} \right)^n$.
\end{lemma}
\begin{proof}
    By \cref{lem:happiness-app-for-general-k}, and since  $k = \lceil10 \log n\rceil$:
    \begin{equation*}
        \begin{split}
            p_D&=\left(1-\left(\frac{1}{2}\right)^k\right)^n=\left(1-\left(\frac{1}{2}\right)^{ \lceil10 \log n\rceil}\right)^n
            \\&\ge\left(1-\left(\frac{1}{2}\right)^{10 \log n}\right)^n=\left(1-n^{10\log\left(\frac{1}{2}\right)}\right)^n=\left(1-n^{-10}\right)^n
        \end{split}
    \end{equation*}
\end{proof}

Before bounding the probability of the last condition required for the realization to be a $k$-bad realization (\cref{lem:prob-d-unhappy}), we first prove {a useful claim (\cref{lem:chernof})}.

\begin{lemma}\label{lem:chernof}
    For positive integer $q$, let $X_{1}, X_{2},\ldots, X_{q}$ be $q$  independent indicator random variables, and let $X=\sum_{i} X_i$. Denote $M \triangleq \mathbb{E}[X]$. Then, for any $s\ge 0$ {such that $s\leq M$},\[\Pr(X > s) \ge 1-\exp\left(\frac{-M+2s}{2}\right).\]
\end{lemma}
\begin{proof}

To lower bound the probability $\Pr(X> s)$, we apply Chernoff bound. For any $0\le\lambda\le1$: \[\Pr\left(X\le (1-\lambda)M\right)\le\exp\left(\frac{-\lambda^2M}{2}\right) \]
$$\Pr\left(X> (1-\lambda)M\right)\ge1-\exp\left(\frac{-\lambda^2M}{2}\right)$$
Set $\lambda = \frac{M - s}{M}$
and note that $\lambda\in [0,1]$ as $0\leq s\leq M$, and observe that $(1-\lambda)M = s$. We obtain
\begin{equation*}
    \begin{split}
        \Pr\left(X> s\right)&\ge1-\exp\left(\frac{-(M-s)^2}{2M}\right)
    \ge1-\exp\left(\frac{-(M-s)^2}{2M}\right)\\&=1-\exp\left(\frac{-M^2+2sM-s^2}{2M}\right)
        \ge1-\exp\left(\frac{-M^2+2sM}{2M}\right)\\&=1-\exp\left(\frac{-M+2s}{2}\right)
    \end{split}
\end{equation*}
\end{proof}

\begin{lemma}\label{lem:prob-d-unhappy} 

For any $n$ that is large enough, for $k = \lceil10 \log n\rceil$ and $d=k^3$,
    the probability that an execution of \cref{algorithm:serial-adaptive} satisfies $d$-positions unhappy is at least $1-\frac{3}{n^{9}}$. 
\end{lemma}

\begin{proof}

W.L.O.G, assume that in $t_0$ the tentatively matched positions are $p_1,\dots,p_{m-k+1}$. Let $X_{1},\dots,X_{m-k+1}$ be indicator random variables, where $X_j=1$ if the tentatively matched position $p_{j}$ does not interim like its matched applicant in iteration $t_0$, and $X_j=0$ otherwise. Let $D$ be the event that $k$-eventual-happiness for applicants holds. Assuming that event $D$ holds, these variables are independent and identically distributed, each being a Bernoulli random variable with success probability $\frac{1}{2}$. 
That is, since up to iteration $t_0$ each position is tentatively matched to the first applicant that interim likes it (by \cref{lem:applicant-matched-to-liked}, regardless of the position's value for that applicant). Let $X=\sum_{j=1}^{m-k+1} X_j$ be a random variable that counts the number of tentatively matched positions that do not interim like their matched applicant in iteration $t_0$. Then, $M\triangleq\mathbb{E}\left[X \mid D\right]=\frac{m-k+1}{2}$. Let $C$ be the event that an execution of \cref{algorithm:serial-adaptive} results in $d$-positions unhappy. 
We bound $\Pr(C)$ using $\{X\ge d\}\subset C$ and the formula of total probability:
\begin{equation}\label{equation4}
    \Pr(C)\ge\Pr(X\ge d)\ge\Pr(X\ge d \mid D)\Pr(D)
\end{equation}

To lower bound the probability $\Pr(C \mid D)$, we apply this to bound the probability that an execution of \cref{algorithm:serial-adaptive} results in $d$-positions unhappy, given that $k$-eventual-happiness for applicants holds, by using \cref{lem:chernof}, and since $M\ge d$ for large enough $m$
\begin{equation*}
    \begin{split}
        \Pr\left(X\ge d \mid D\right)&\ge\Pr\left(X> d \mid D\right)\ge1-\exp\left(\frac{-M+2d}{2}\right)
        \\&=1-\exp\left(\frac{\frac{-m+k-1}{2}+2k^3}{2}\right)=1-\exp\left(\frac{-m+k-1+4k^3}{4}\right)
        \\&\ge1-\exp\left(\frac{-n+k+4k^3}{4}\right)\ge1-\exp\left(-\frac{n}{5}\right)
    \end{split}
\end{equation*}
where the last inequality holds for any $n$ that is large enough, as $k = \lceil10 \log n\rceil$. We conclude:
\begin{equation}\label{equation5}
    \Pr\left(X\ge d \mid D\right)\ge1-\exp\left(-\frac{n}{5}\right)
\end{equation}
Note that by \cref{lem:happiness-app}$$p_D\ge \left(1-n^{-10}\right)^n=\exp\left(n\cdot\ln\left(1-n^{-10}\right)\right)$$ using the Taylor series of $\ln(1-x)=-\sum_{i=1}^\infty\frac{x^i}{i}\ge -2x$ for sufficiently small positive $x\ge0$:
$$p_D \ge\exp\left(n\cdot\left(-\frac{2}{n^{10}}\right)\right)=\exp\left(-\frac{2}{n^{9}}\right)$$
Combining with \cref{equation4} and \cref{equation5}:
\begin{equation*}
\begin{split}
    \Pr(C)&\ge \left( 1-\exp\left(-\frac{n}{5}\right)\right)\cdot\exp\left(-\frac{2}{n^{9}}\right)
    \\&\ge\exp\left(-\frac{2}{n^{9}}\right)-\exp\left(-\frac{n}{5}-\frac{2}{n^{9}}\right) 
 \\& \geq \exp\left(-\frac{2}{n^{9}}\right)-\exp\left(-\frac{n}{5}\right)
\end{split}
\end{equation*}
Using the fact that $e^x\ge1+x$ for every $x$:
$$\Pr(C)\ge  1-\frac{2}{n^{9}}-\exp\left(-\frac{n}{5}\right)\ge1-\frac{3}{n^{9}}$$
Where the last inequality holds for any $n$ that is large enough.

\end{proof}
\subsubsection{Proof of \cref{lem:prob-k-bad}}\label{sec:k-bad}
\KBadProbability*
\begin{proof}
    For a realization to be considered a ``$k$-bad realization'', at least one of the following events must hold:
    \begin{enumerate}
         \item$k$-eventual-mutual-happiness fails. By \cref{lem:mutual-happiness-app} and \cref{lem:mutual-happiness-pos} and the union bound, it occurs with probability at most $2-\left(1-\frac{1}{n^4}\right)^n-\left(1 - \frac{1}{n^4} \right)^{n+k-1}$.
        \item$d$-positions unhappy fails, which, for large enough $n$, occurs with probability at most $\frac{3}{n^{9}}$, by \cref{lem:prob-d-unhappy}.
    \end{enumerate}

Applying the union bound to these two events, the total probability that at least one fails is at most: 
$$p_{bad}\le2-\left(1-\frac{1}{n^4}\right)^n-\left(1 - \frac{1}{n^4} \right)^{n+k-1}+\frac{3}{n^{9}}$$ 
Using Bernoulli's inequality (as $-\frac{1}{n^4}\ge -{1}$ for every $n\in\mathbb{N}$), we get the inequalities $ \left(1 - \frac{1}{n^4}\right)^n \ge 1 - \frac{n}{n^4} $ and $ \left(1 - \frac{1}{n^4}\right)^{n+k-1} \ge 1 - \frac{n+k-1}{n^4} $.  For large enough $n$,
$$p_{bad}\le\frac{1}{n^3}+\frac{n+\lceil10\log n\rceil-1}{n^4} +\frac{3}{n^{9}}\le\frac{3}{n^3}$$
\end{proof}
\subsubsection{Bounding the Expected Number of Interviews}\label{subsec:number-of-interviews}
\paragraph{Proof of \cref{lem:number-of-interviews-before-t0}}
Before we can bound the expected number of interviews in \cref{algorithm:serial-adaptive}, we first bound the expected number of interviews conducted before reaching the transition point.

\InterviewsBeforeTransition*
\begin{proof} 
    For $q\leq k$, let $Z_{i,q}$ be the indicator that in the $q$-th interview of applicant $a_i$ she interim likes the position with which she interviewed.
    Let $D_i$ be the event that if applicant $a_i$ has participated in $k$ interviews, she interim likes at least one of her interviewing positions. Let $D=\bigcap_{i=1}^n D_i$ be the event that $k$-eventual-happiness for applicants holds. We want to compute $\Pr(Z_{i,q} \mid D)$. Since the value for interviews are independent across applicants, $\Pr(Z_{i,q} \mid D)=\Pr(Z_{i,q} \mid D_i)$. By Bayes' theorem, $$\Pr(Z_{i,q}=1\mid D_i)=\frac{\Pr(D_i\mid Z_{i,q}=1)\cdot \Pr(Z_{i,q}=1)}{Pr(D_i)}$$ Note that for $q\leq k$ it holds that $\Pr(D_i \mid Z_{i,q} = 1) = 1$, since if one of the interviews is successful, $D_i$ occurs. Moreover, by \cref{lem:happiness-app-for-general-k} we have that for every $i$, $\Pr(D_i)=1-\left(\frac{1}{2}\right)^k.$ So, $$\Pr(Z_{i,q}=1\mid D_i)=\frac{\Pr(Z_{i,q}=1)}{Pr(D_i)}=\frac{\frac{1}{2}}{\left(1-\left(\frac{1}{2}\right)^k\right)}\ge\frac{1}{2}$$ for every $k\in\mathbb{N}$. Since the $k$-eventual-happiness for applicants condition holds, by \cref{lem:applicant-matched-to-liked} and \cref{cor:applicant-didnt-interview} all the interviews up to iteration $t_0$ are conducted by applicants $a_1,...,a_{m-k+1}$, if $n+k-1>m$, or by all applicants otherwise. Each of these applicants interviews until they find a position they interim like, which takes at most $k$ steps. Since each interview is a Bernoulli trial with success probability at least $\frac{1}{2}$, the expected number of trials until success is at most $2$. Therefore, the expected number of interviews across $\min(m-k+1,n)$ applicants is at most $2n$.
\end{proof}
We now have the necessary tools to prove the following lemma:
\paragraph{Proof of \cref{tmh:iid}}
We prove \cref{tmh:iid} by separating the proof on the expected number of interviews into two cases, one in which  $n+\lceil10\log n\rceil-1>m\ge n$ (which is handled in \cref{lem:2n-for-small-m}), and the other when $m\ge n+\lceil10\log n\rceil-1$ (which is handled in \cref{lem:2n-for-big-m}). Later, in \cref{lem:max-position-index-interviewed,lem:agent-k-interviews} we prove that with high probability, every applicant stops interviewing after $O(\log n)$ interviews.
\begin{lemma}\label{lem:2n-for-small-m}
    Consider any bilaterally ex-ante equivalent {instance} with  $n$ applicants and $n+\lceil10\log n\rceil-1>m\ge n$ positions.
    The expected number of interviews in \cref{algorithm:serial-adaptive} is  $2\cdot n+O(\log^3n)$. 
\end{lemma}
\begin{proof} 
For a given $n$, let $k=\lceil10\log n\rceil$. Let $p_{bad}$ be the probability that a realization is a \textit{``$k$-bad realization''} and let $p_D$ be the probability that $k$-eventual-happiness for applicants holds. Let $Y$ be a random variable counting the number of interviews conducted by \cref{algorithm:serial-adaptive}. By the law of total expectation: 
    \begin{equation}\label{equation3}  
    \begin{split}
        \mathbb{E}\left[Y\right]&=\mathbb{E}\left[Y \mid D\right]\Pr(D)+\mathbb{E}\left[Y \mid\overline{D}\right]\Pr(\overline{D})
        {\le\mathbb{E}\left[Y \mid D\right]+\mathbb{E}\left[Y \mid \overline{D}\right]\Pr(\overline{D})}
    \end{split}
    \end{equation} Let $Y_1$ count the number of interviews conducted up to iteration $t_0$ and let $Y_2$ count the number from iteration $t_0$, so that, $Y=Y_1+Y_2$. 
    Then: 
\begin{equation} \label{equation2}
    \mathbb{E}\left[Y \mid D\right]=\mathbb{E}\left[Y_1 \mid D\right]+\mathbb{E}\left[Y_2 \mid D\right]
\end{equation}
    
    Let $B_1$ be the event that for every applicant that has participated in $k$ interviews, the applicant and at least one of her interviewing positions interim like each other. Let $B_2$ be the event that for every position that has participated in $k$ interviews, the position and at least one of its interviewing applicants interim like each other. Finally, let $C$ be the event that $d$-positions unhappy for $d=k^3$ holds. Let event $E_G=B_1\cap B_2\cap C$ be the event that a realization is not a \textit{``$k$-bad realization''} (G stands for Good). Then $\Pr(E_G)=1-p_{bad}$. 
     
    \begin{equation*}\begin{split}
         \mathbb{E}\left[Y_2 \mid D\right] &=\mathbb{E}\left[Y_2 \mid D\cap E_G\right]\cdot\Pr(E_G \mid D)+\mathbb{E}\left[Y_2 \mid D\cap \overline{E_G}\right]\cdot\Pr(\overline{E_G} \mid D) 
    \\&=\frac{1}{\Pr(D)}\left(\mathbb{E}\left[Y_2 \mid D\cap E_G\right]\cdot\Pr(D\cap E_G)+\mathbb{E}\left[Y_2 \mid D\cap \overline{E_G}\right]\cdot\Pr(D\cap \overline{E_G})\right)
    \end{split}
    \end{equation*}
    Since $k$-eventual-mutual-happiness implies $k$-eventual-happiness for applicants, we have that $E_G\subseteq B_1 \subseteq D$ and therefore:
    \begin{equation*}
        \begin{split}
            \mathbb{E}\left[Y_2 \mid D\right] &=\frac{1}{\Pr(D)}\left(\mathbb{E}\left[Y_2 \mid  E_G\right]\Pr(E_G)+\mathbb{E}\left[Y_2 \mid D\cap \overline{E_G}\right]\cdot (\Pr(D)-\Pr(E_G))\right)
            \\&\le\frac{1}{\Pr(D)}\left({\mathbb{E}\left[Y_2 \mid  E_G\right]+\mathbb{E}\left[Y_2 \mid D\cap \overline{E_G}\right]\cdot (\Pr(D)-\Pr(E_G))}\right)
        \end{split}
    \end{equation*}

    By \cref{lem:interviews-after-t0} and since the number of interviews is always at most $n\cdot m$, we can bound the remaining terms:
    \begin{equation} \label{equation1}
         \mathbb{E}\left[Y_2 \mid D\right] \le \frac{1}{p_D}\left(k^3+n\cdot m\cdot (p_D-1+p_{bad})\right)
    \end{equation}
    
    {As the number of interviews is always at most $n\cdot m$ we have $\mathbb{E}\left[Y \mid \overline{D}\right]\leq n\cdot m$. Thus, by \cref{lem:number-of-interviews-before-t0}, \cref{equation2}, \cref{equation1} and \cref{equation3}:
    \begin{equation*}
        \begin{split}
            \mathbb{E}\left[Y\right] &\leq\mathbb{E}\left[Y_1 \mid D\right]+\mathbb{E}\left[Y_2 \mid D\right] +\mathbb{E}\left[Y \mid \overline{D}\right]\Pr(\overline{D})
            \\&\leq 2n+\frac{1}{p_D}\left(k^3+n\cdot m\cdot (p_D-1+p_{bad})\right)+n\cdot m\cdot(1-p_D)
            \\&\leq 2n+\frac{1}{p_D}\left(k^3+n\cdot m\cdot (p_D-1+p_{bad})+n\cdot m\cdot(1-p_D)\right)
            \\&\leq 2n+\frac{1}{p_D}\left(k^3+n\cdot m\cdot p_{bad}\right)
        \end{split}
    \end{equation*}}
        
    By \cref{lem:prob-k-bad}, for large enough $n$ it holds  that $p_{bad}\le{\frac{3}{n^3}}$. Thus:
    \begin{equation*}
        \begin{split}
        \mathbb{E}\left[Y\right]&\leq 2n+\frac{1}{p_D}\left(k^3+ n\cdot m\cdot{\frac{3}{n^3}}\right)
    \le2n+\frac{1}{p_D}\left(k^3+ n^3\cdot{\frac{3}{n^3}}\right)
    \end{split}
    \end{equation*}
    where the last inequality holds for any $n$ that is large enough, as $k = \lceil10 \log n\rceil$ and $m\le n + k - 1$. 
    By \cref{lem:happiness-app}, $p_D\ge\left(1 - \frac{1}{n^{10}} \right)^n\ge \frac{n^9-1}{n^{9}}$.
    Thus,
     \begin{equation*}
        \begin{split}
        \mathbb{E}\left[Y\right]&\le2n+2\cdot\left(k^3+3\right)=2n+O(\log^3n).
    \end{split}
    \end{equation*}
\end{proof}

\begin{lemma}\label{lem:2n-for-big-m}
    Consider any bilaterally ex-ante equivalent {instance} with  $n$ applicants and $m\ge n+\lceil10\log n\rceil-1$ positions.
    The expected number of interviews in \cref{algorithm:serial-adaptive} is {$2\cdot n+O(\log^3n)$.} 
\end{lemma}
\begin{proof} 
For a given $n$ and $m\ge n+\lceil10\log n\rceil-1$, let $k=m-n+1$. Let $D$ be the event that $k$-eventual-happiness for applicants holds and let $p_D$ be the probability for that event. Let $Y$ be a random variable counting the number of interviews conducted by \cref{algorithm:serial-adaptive}. By the law of total expectation: 
    \begin{equation}\label{eq:condition-on-k-happ}  
    \begin{split}
        \mathbb{E}\left[Y\right]&=\mathbb{E}\left[Y \mid D\right]\Pr(D)+\mathbb{E}\left[Y \mid \overline{D}\right]\Pr(\overline{D})
        {\le\mathbb{E}\left[Y \mid D\right]+\mathbb{E}\left[Y \mid \overline{D}\right]\Pr(\overline{D})}
    \end{split}
    \end{equation}
    
    As the number of interviews is always at most $n\cdot m$ we have $\mathbb{E}\left[Y \mid \overline{D}\right]\leq n\cdot m$. Since $m\ge n+ k -1$ the algorithm stops before $t_0$, by \cref{lem:number-of-interviews-before-t0}, \cref{eq:condition-on-k-happ}:
    \begin{equation}\label{eq:after-numbers}
        \begin{split}
            \mathbb{E}\left[Y\right] &\leq\mathbb{E}\left[Y \mid D\right] +\mathbb{E}\left[Y \mid \overline{D}\right]\Pr(\overline{D})
            \\&\leq 2n+m\cdot n \cdot \Pr(\overline{D})
        \end{split}
    \end{equation}
        
    By \cref{cor:happiness-app-for-general-k}, it holds that $ \Pr(\overline{D})=1-p_D\le\frac{n}{2^k}$.
    Thus, by \cref{eq:after-numbers}:
    \begin{equation*}
        \begin{split}
        \mathbb{E}\left[Y\right]&\le2n+\frac{n^2\cdot m}{2^k}
    \end{split}
    \end{equation*}
    when $m\le n^7$, since $k\ge\lceil10\log n\rceil$. Then,
     \begin{equation*}
        \begin{split}
        \mathbb{E}\left[Y\right]&\le2n+\frac{n^2\cdot m}{2^k}\le2n+\frac{n^2\cdot m}{2^{\lceil10\log n\rceil}} 
        \\&\le2n+\frac{n^2\cdot m}{2^{10\log n}}=2n+\frac{n^2\cdot m}{n^{10}}\\& \le2n+\frac{n^9}{n^{10}} =2n+\frac{1}{n} =2 n+O(\log^3n)
    \end{split}
    \end{equation*}
    Otherwise, when $m>n^7$, then $k> m -\sqrt[7]{m}+1\ge \sqrt{m}$:
    \begin{equation*}
        \begin{split}
        \mathbb{E}\left[Y\right]&\le2n+\frac{n^2\cdot m}{2^k}\le2n+\frac{m^3}{2^{ \sqrt{m}}} =2 n+O(\log^3n)
    \end{split}
    \end{equation*}
\end{proof}

Before proving that with high probability every applicant stops interviewing after $O(\log n)$ interviews, we first introduce a supporting lemma that limits the set of positions we need to analyze when $ m\ge n + \lceil10\log n\rceil - 1$.
\begin{lemma}\label{lem:max-position-index-interviewed}
Assume that $ m\ge n + \lceil10\log n\rceil - 1$ and $k$-eventual happiness for applicants holds for $k = \lceil10\log n\rceil$. Then, applicant $a_i$ will not interview any position $p_j$ with index $j \ge k + i$. Consequently, no position $p_j$ with $j \ge n + k$ is ever interviewed.
\end{lemma}
\begin{proof}
We prove this by induction on the applicant index $i \ge 1$. Recall that by \cref{algorithm:serial-adaptive}, an applicant is selected by the algorithm until she is matched, and she interviews unmatched positions in increasing order of their indices (tie-breaking towards smaller indices).

For the base case $i=1$, applicant $a_1$ interviews positions in strictly increasing order starting from $p_1$. Under the $k$-eventual happiness assumption and when $ m\ge n +k - 1$, she stops interviewing once she finds a position she interim likes, which occurs within her first $k$ interviews. Therefore, the maximum index of a position she interviews is $k$, meaning she will not interview any position $p_j$ with $j \ge k+1$.

Assume the claim holds for all applicants $a_{i'}$ where $i' < i$. This inductive hypothesis implies that all $i-1$ previously matched applicants found their matches at positions with indices strictly less than $k+i-1$. When applicant $a_i$ begins interviewing, she considers the remaining unmatched positions in increasing order of index. Because the preceding $i-1$ applicants occupy $i-1$ positions,  the first $k$ unmatched positions are guaranteed to fall within $[1,k+i-1]$. Under the $k$-eventual happiness assumption and when $m\ge n+k-1$, $a_i$ will find a match within her first $k$ available positions. Thus, she will not interview any position beyond index $k + i - 1$.

Since the final applicant is $a_n$, no applicant will ever interview a position with an index greater than $k + n - 1$. Therefore, under the $k$-eventual happiness assumption, every position $p_j$ such that $j \ge n + k$ never interviews.
\end{proof}

\begin{lemma}\label{lem:agent-k-interviews}
Consider any bilaterally ex-ante equivalent instance with $n$ applicants and $m\ge n$ positions.
In \cref{algorithm:serial-adaptive}, with probability $1 - O(n^{-3})$, every agent conducts at most $O(\log n)$ interviews.
\end{lemma}
\begin{proof}
   Assume that $ n + \lceil10\log n\rceil - 1> m$, and let $k=\lceil10\log n\rceil$ and assume that the realization that is not $k$-bad. In \cref{algorithm:serial-adaptive}, by \cref{lem:position-stops-interview-after++}, a position stops interviewing once it meets an applicant such that they interim like each other. Under the assumption the realization is not $k$-bad, this occurs after at most $k$ interviews. Moreover, when the realization is not $k$-bad, \cref{lem:applicant-do-max-k-interviews} applies. Consequently, every applicant will participate in at most  $k$ interviews. Therefore, when the realization is not $k$-bad, which holds a probability of at least $1-\frac{3}{n^3}$ for large enough $n$ (by \cref{lem:prob-k-bad}), every agent conducts at most $k$ interviews, which is $O(\log n)$.

   Now, consider the case that $ m\ge n + \lceil10\log n\rceil - 1$. Assume that $k$-eventual happiness for applicants holds for $k=\lceil10\log n\rceil$, and denote this event by $D$. When $ m\ge n +k - 1$, by \cref{lem:applicant-matched-to-liked}, an applicant stops interviewing after the first position she interim likes (which can take at most $k$ interviews). 
    Therefore, under event $D$-which occurs with a probability of at least $1-\frac{1}{n^9} \ge 1-\frac{1}{n^3}$ by \cref{cor:k-happ-better-than-good}-every applicant conducts at most $k = O(\log n)$ interviews.
    
   By \cref{lem:max-position-index-interviewed}, no position $p_j$ with index $j \ge n + k$ is ever interviewed. Therefore, we only need to analyze the interview counts for the first $n + k - 1$ positions. By \cref{lem:position-stops-interview-after++}, a position stops 
   interviewing once it meets an applicant such that they interim like each other. 
   Let $E_P$ be the event that every position stops interviewing after at most $k$ interviews. Conditioned on $D$, we 
   established that only the first $n+k-1$ positions may be interviewed. 
   Therefore, by \cref{lem:mutual-happiness-pos}, we have $$\Pr(E_P \mid D) \ge \left(1-n^{-4}\right)^{n+k-1}\ge 1 - \frac{n+k-1}{n^4}$$ From \cref{cor:k-happ-better-than-good}, we know that $\Pr(D) \ge 1-\frac{1}{n^9}$. By the law of total probability
   $$ \Pr( E_P) \ge \Pr(E_P \mid D)\Pr(D) \ge \left( 1 - \frac{n+k-1}{n^4}\right) \cdot \left(1-\frac{1}{n^9}\right)\ge 1-\frac{3}{n^3}$$ for large enough $n$. Therefore, with probability $1 - O(n^{-3})$, every position conducts at most $k = O(\log n)$ interviews.

\end{proof}
From \cref{lem:2n-for-small-m,lem:2n-for-big-m,lem:agent-k-interviews} above, we obtain the correctness of the following result.
\TwoInterviews*

\section{Deferred Proofs and Algorithm from Section \ref{sec:parallel} (Minimizing the Number of Interview Rounds)}\label{appendix:hybrid}
This section contains the proofs for the positive results presented in \cref{sec:parallel} regarding the hybrid algorithm (\cref{algorithm:parallel-bilateral-instantiation}).
\subsection{Deferred Algorithms and Proof from \cref{subsec:hybrid-alg}}
{\subsubsection{DA on Applicants Truncated Interim Preferences}\label{appendix:da-trunc-alg}
We present a modified version of the deferred acceptance algorithm that acts on restricted preference lists (Subroutine~\ref{algorithm:truncated-DA}), the algorithm is used by \cref{algorithm:parallel-bilateral-instantiation}.
\floatname{algorithm}{Subroutine}
\begin{algorithm}[H]
\caption{\textsc{DA\_on\_Applicants\_Truncated\_Interim\_Preferences}}
\label{algorithm:truncated-DA}
\begin{algorithmic}[1]

\Statex \textbf{Input:} 
    an instance $I= (\app,\pos, {\bf F}, {\bf G})$;
    a current matching $\mu$; 
    and a set of interviews $Z$;
    \Statex \textbf{Output:} Updated matching $\mu$.
\ForAll{$a_i \in \app$}
    \State let $\succ_{a_i}^{Z}$ be $a_i$'s interim preference order over positions in $\pos$.
    \State let $p_j$ (if any) be the most preferred position in $\succ^{Z}_{a_i}$ such that $(a_i,p_j)\notin Z$.
    \State let $\succ_{a_i}^{trunc}$ be the prefix of $\succ^{Z}_{a_i}$ preferred to $p_j$\;
    \Statex \hspace{\algorithmicindent}\Comment{after truncation, $\succ_{a_i}^{trunc}$ contains only positions that have interviewed $a_i$}
\EndFor
\State$ \mu\gets \textsc{Applicant-Proposing\_DA}\left(\app,\pos,\{\succ_{a_i}^{trunc}\}_{i \in [n]},\{\succ_{p_j}^Z\}_{j \in [m]},\mu\right)${\Comment{see \cref{algorithm:DA}}}  \
\State\Return$\mu$
\end{algorithmic}
\end{algorithm}
}
\subsubsection{Applicant Proposing DA Algorithm}\label{appendix:da-algo}
For completeness, we formally describe the DA subroutine (\cref{algorithm:DA}) that we use in Subroutine~\ref{algorithm:truncated-DA}.
\begin{algorithm}[h]
\caption{\textsc{Applicant\_Proposing\_DA}}
\label{algorithm:DA}
\begin{algorithmic}[1]
\Statex \textbf{Input:} 
    a set $\app$ of $n$ applicants; a set $\pos$ of $m$ positions;
    preference relations $\{\succ_{a_i}\}_{i \in [n]}$ and $\{\succ_{p_j}\}_{j \in [m]}$;
    and a current matching $\mu$ (set to be empty if not provided. That is, if not provided,  set $\mu(x)=x$ for every $x\in \app\cup \pos$); 
    \Statex \textbf{Output:} Updated matching $\mu$.
\While{$\exists a_i \in \app$ such that $\mu(a_i)= a_i$} and $\exists p_j$ such that $p_j$ is listed in $\succ_{a_i}$ and such that $p_j$ has not yet rejected $a_i$
    \State let $p_j$ be $a_i$'s most preferred position in $\succ_{a_i}$ that has not yet rejected $a_i$.
    \If{$a_i \succ_{p_j} \mu(p_j)$}
        \State $p_{j}$ rejects applicant $\mu(p_{j})$
        \State $\mu(a_i) \gets p_j$ and $\mu(p_j) \gets a_i$.
    \Else
        \State $p_j$ rejects $a_i$.
    \EndIf
\EndWhile
\State\Return$\mu$
\end{algorithmic}
\end{algorithm}
\subsubsection{All Interviews Algorithm}\label{appendix:all-interviews-algorithm}
We present the \textsc{All\_Interviews\_Algorithm} for the case in which the number of positions is at least the number of applicants, that is, $m \ge n$. 
{
\floatname{algorithm}{Subroutine}
\begin{algorithm}[H]
\caption{\textsc{All\_Interviews\_Algorithm}}
\label{algorithm:all-interviews-algorithm}
\begin{algorithmic}[1]

\State \textbf{Input:}
 a set $\app$ of $n$ applicants; a set $\pos$ of $m$ positions; and a set of interviews $Z$; 
\State \textbf{Output:} Updated set of interviews $Z$.
\For{$\ell = 0$ \textbf{to} $m-1$} 
    \ForAll{$a_i \in \app$ \textbf{in parallel}}
        \State let $p_j \gets p_{1+(i+\ell-1)}\bmod m$
        \If{$(a_i,p_j) \notin Z$}
            \State $p_j$ interviews $a_i$
        \EndIf
    \EndFor
\EndFor
\State \Return $Z$

\end{algorithmic}
\end{algorithm}}
\subsubsection{Proof of \cref{lem:hybrid-stable}}\label{appendix:proof-stable-hybrid}
\HybridStable*
\begin{proof}
If the algorithm reaches Phase~3 of \cref{algorithm:parallel-bilateral-instantiation} (i.e., the fallback stage in line 13), then all values have been realized, and \cref{algorithm:DA} execute the standard DA algorithm over fully realized random preferences, which is known to produce a stable matching. Otherwise, if the algorithm proceeds only through Phase~1 and Phase~2 (i.e., does not reach the fallback stage in line 13), the proof is exactly as in \cref{lemma:serial-stable}.
\end{proof}
\subsection{Deferred Proofs from \cref{subsec:hybrid-proof}}\label{appendix:hybrid-bilateral}
\subsubsection{Proof of \cref{lemma:parallel-matched-stays-matched}}\label{appendix:proof-phase1}
\HybridPhaseOne*
\begin{proof}
We prove all three statements by induction on the iteration~$t$ (i.e., a single execution of the while-loop of the hybrid  \cref{algorithm:parallel-bilateral-instantiation}).

\noindent\emph{Base case ($t=1$).}
\begin{enumerate}
    \item In iteration $t=1$, under Subroutine~\ref{algorithm:pick-interviews-bilateral-equivalent} when the instance is in the bilateral ex-ante equivalent setting, each applicant $a_i$ has $\pos_i^{*} = \pos$. Since $m\ge n$, the algorithm assigns each applicant an interview with some unmatched position. 
    \item All positions are initially unmatched, so no unmatched applicant in $\app'$ can propose to a matched position.
    \item Let $p_j$ be the position that applicant $a_i$ interviews in iteration $t=1$. If $a_i$ interim likes $p_j$, then by definition after the interview  \[
    \vutility{i}{j}{1} > V_{i,j} = V =  V_{i,j'} = \vutility{i}{j'}{1} 
    \qquad\text{for every } j' \neq j,
    \] Thus $p_j$ is $a_i$’s top interviewed position, and she will propose to $p_j$ in DA. For any other applicant $a_{i'} \ne a_i$, since $p_j \notin \succ^{trunc}_{a_{i'}}$, the only proposal that $p_j$ receives is from $a_i$.  Because $p_j$ is unmatched, it accepts the proposal.  
    Hence every applicant who interim likes the first position she interviewed with becomes tentatively matched to it.
\end{enumerate}

\noindent\emph{Induction step.}
Assume that for every iteration $t'<t$:
\begin{enumerate}
    \item Every unmatched applicant in $\app'$ in iteration $t'$ interviews with an unmatched position in iteration $t'$.
    \item No unmatched applicant in $\app'$  proposes to a matched position.
    \item Any applicant who interim liked the position she interviewed in iteration $t'$ became tentatively matched to that position in $t'$ and remained matched until iteration $t-1$.
\end{enumerate}
We prove the statements for iteration $t$.

\begin{itemize}
    \item Since $m \ge n$, at the beginning of iteration $t$ there are at least $|\app''| + (k-1)$ unmatched positions. By the induction hypothesis, the applicants in $\app''$ in iteration~$t$ are precisely those who were in $\app'$ at $t=1$ and have not interim liked any position so far. Since $k$-eventual-happiness for applicants holds, each applicant meets a position that she interim likes within her first $k$ interviews. Thus, at iteration $t$, every applicant in $\app''$ has interviewed fewer than $k$ positions.

    Therefore each applicant in $\app''$ still has at least $|\app''|$ unmatched positions that she has not yet interviewed.
    Therefore, for any subset $S \subseteq \app''$, the set of neighbors has size at least $|\app''|$, and Hall’s condition holds for the bipartite graph whose one side is $\app''$ and the other side consists of unmatched positions.  Hence Subroutine~\ref{algorithm:pick-interviews-bilateral-equivalent} assigns interviews to all applicants in $\app'$ using only unmatched positions.
    \item  Fix some $a_i \in \app''$ and let $p_j$ be the unique unmatched position with which she interviews in iteration $t$.
    \begin{itemize}
        \item If $a_i$ interim likes $p_j$, then
          \[
         \vutility{i}{j}{t} > V_{i,j} = V = V_{i,j'}= \vutility{i}{j'}{t}
          \quad\text{for every } j' \neq j \text{ with } p_{j'} \in \pos_i^*,
          \]
          so $a_i$ will want to propose to $p_j$.  
        \item Otherwise, $\vutility{i}{j}{t} < V_{i,j} = V = V_{i,j'} = \vutility{i}{j'}{t}$ for all $j' \neq j$ with $p_{j'} \in \pos_i^*$. Since she has interviewed fewer than $k$ positions so far, there exists a position 
        $p_{j^*} \in \pos_i^{*}$ that she prefers and that has not rejected her, so she will not propose to $p_j$.
    \end{itemize}
    Therefore, an applicant $a_i \in \app''$ proposes in iteration~$t$ only if she interim likes the position she interviews in that iteration, and in that case she proposes to that position and to no other. Since each applicant interviews exactly one position, each position receives at most one proposal in this iteration. Because all these positions are unmatched, every proposal is accepted.  

    In particular:
    \begin{itemize}
        \item no unmatched applicant proposes to a matched position. Consequently, no existing tentative match is broken, and every applicant who was matched prior to iteration $t$ remains matched afterward. and
        \item every applicant who interim liked the position she interviewed in iteration $t$, became tentatively matched to it in that iteration.
    \end{itemize}
\end{itemize}

This completes the induction and the proof.
\end{proof}
\subsubsection{Proof of \cref{lemma:parallel-log-iterations}}\label{appendix:parallel-log-iterations}
\KRounds*
\begin{proof}
    By \cref{lemma:parallel-matched-stays-matched}, an applicant stops interviewing at Phase~1 once she is interviewed with a position that she interim likes. Until then, in each iteration, every unmatched applicant conducts an interview.  
    Since $k$-eventual-happiness for applicants holds, there can be at most $k$ interview rounds.
\end{proof}

\subsubsection{Proof of \cref{lem:log-rounds}}\label{appendix:log-rounds}
We begin with basic corollaries of \cref{lemma:parallel-log-iterations,lemma:parallel-matched-stays-matched} that will be used throughout the proof of \cref{lem:log-rounds}.

\begin{corollary}\label{cor:k-rounds}
    When the realization is not a $k$-bad realization, Phase~1 terminates after most $k$ iterations.
\end{corollary}
\begin{proof}
    For any realization that is not $k$-bad, the property of $k$-eventual-happiness for applicants holds by definition. Therefore, by \cref{lemma:parallel-log-iterations}, Phase~1 terminates after at most $k$ iterations.
\end{proof}

\begin{corollary}\label{cor:happy-k-rounds}
    Assume that $k$-eventual-happiness for applicants holds for the execution of \cref{algorithm:parallel-bilateral-instantiation}, then \cref{algorithm:parallel-bilateral-instantiation} does not reach Phase~3.
\end{corollary}
\begin{proof}
    By \cref{lemma:parallel-matched-stays-matched}, as long as an applicant in $\app'$ remains unmatched, she receives an interview in every iteration. Consequently, Subroutine~\ref{algorithm:pick-interviews-bilateral-equivalent} never outputs a set of interviews that is smaller than $\app''$, and thus \cref{algorithm:parallel-bilateral-instantiation} continues executing the DA phase (Subroutine~\ref{algorithm:truncated-DA}) until all applicants in $\app'$ become matched.
\end{proof}

\begin{corollary}\label{cor:bad-k-rounds}
    If the realization is not a $k$-bad realization, then \cref{algorithm:parallel-bilateral-instantiation} does not reach Phase~3.
\end{corollary}
\begin{proof}
    For any realization that is not $k$-bad, the property of $k$-eventual-happiness for applicants holds by definition. Therefore, by \cref{cor:happy-k-rounds}, \cref{algorithm:parallel-bilateral-instantiation} does not reach Phase~3.
\end{proof}

\paragraph{{Claims about the Sequential Algorithm that also Hold for the Hybrid Algorithm under the Non-$k$-Bad Realization Assumption}
}
To prove \cref{lem:log-rounds}, we next show that the claims stated in the lemmas for the sequential adaptive algorithm in \cref{appendix:serial-all-iid}, also hold for \cref{algorithm:parallel-bilateral-instantiation} whenever the realization is not $k$-bad. This allows us to bound the number of interviews performed during Phase~2. Throughout the remainder of this section, we verify that each lemma from the sequential analysis remains valid in the hybrid setting by establishing it for Phase~1. Since Phase~2 coincides with the sequential adaptive algorithm, these lemmas hold unchanged during Phase~2. Moreover, by \cref{cor:bad-k-rounds}, when the realization is not $k$-bad the algorithm never reaches Phase~3. Consequently, it suffices to prove the lemmas for Phase~1.
\begin{corollary}\label{lem:hybrid-position-stops-interview}
    Suppose that the realization is not a $k$-bad realization.
    Let $a_{i^*}$ be the unmatched applicant chosen to interview, and let $p_{j^*}$ be the position selected for the interview.  If $p_{j^*}$ is currently matched to an applicant with whom it interim likes each other, then $p_{j^*}$ rejects $a_{i^*}$ without conducting an interview.  
    This statement coincides with \cref{lem:position-stops-interview}, {but for \cref{algorithm:parallel-bilateral-instantiation} rather than for \cref{algorithm:serial-adaptive}}.
\end{corollary}

\begin{proof}
    In Phase~1, the claim holds trivially, since applicants only interview with unmatched positions. 
\end{proof}

\begin{corollary} \label{cor:hybrid-one-sided-propose-accept}
    Suppose that the realization is not a $k$-bad realization.
    Fix a position $p_j$ and an applicant $a_i$, and assume that at iteration $t$ they have met for an interview ($(a_i,p_j)\in Z^t$) and that $a_i$ and $p_j$ interim like each other. Then applicant $a_i$ will propose to $p_j$ and $p_j$ will accept the proposal.
    This statement coincides with \cref{lem:one-sided-propose-accept}, {but for \cref{algorithm:parallel-bilateral-instantiation} rather than for \cref{algorithm:serial-adaptive}}.
\end{corollary}

\begin{proof}
    In Phase~1, since not being a $k$-bad realization implies that $k$-eventual-happiness for applicants holds, the claim follows from \cref{lemma:parallel-matched-stays-matched}.   
\end{proof}

\begin{corollary}\label{cor:hybrid-applicant-like-propose}
    Suppose that the realization is not a $k$-bad realization. Fix a position $p_j$ and an applicant $a_i$, and assume that at some iteration $t$ they have met for an interview ($(a_i,p_j)\in Z^t$) and that $a_i$ interim likes $p_j$. Then applicant $a_i$ will propose to $p_j$. Moreover, if $p_j$ is unmatched, it will accept the proposal. This statement coincides with \cref{lem:applicant-like-propose}, {but for \cref{algorithm:parallel-bilateral-instantiation} rather than for \cref{algorithm:serial-adaptive}}.
\end{corollary}

\begin{proof}
    In Phase~1, since not being a $k$-bad realization implies that $k$-eventual-happiness for applicants holds, the claim follows from \cref{lemma:parallel-matched-stays-matched}. 
\end{proof}

\begin{corollary}\label{lem:hib-minus-breaks+}
    Suppose that the realization is not a $k$-bad realization. Let $a_{i^*}$ be the unmatched applicant
    and suppose that in iteration $t$ applicant $a_{i^*}$ proposes to $p_{j^*}$. If $p_{j^*}$ is currently tentatively matched to an applicant such that they interim like each other, then $a_{i^*}$ does not interim like $p_{j^*}$. This statement coincides with \cref{lem:minus-breaks+}, {but for \cref{algorithm:parallel-bilateral-instantiation} rather than for \cref{algorithm:serial-adaptive}}.
\end{corollary}
\begin{proof}
     In Phase~1, since not being a $k$-bad realization implies that $k$-eventual-happiness for applicants holds, the claim follows trivially from \cref{lemma:parallel-matched-stays-matched}, since applicants don't propose to matched positions.
\end{proof}

\begin{corollary}\label{cor:hybrid-parallel-matched-stays-matched}
Suppose that the realization is not a $k$-bad realization. Fix a position $p_j$ and an applicant $a_i$, and assume that in iteration $t$ they have met for an interview and that $a_i$ and $p_j$ interim like each other. 
Then, applicant $a_i$ will propose to $p_j$, and $p_j$ will accept the proposal. After accepting, $p_j$ will reject any further interviews. This statement coincides with \cref{lem:position-stops-interview-after++}, {but for \cref{algorithm:parallel-bilateral-instantiation} rather than for \cref{algorithm:serial-adaptive}}.
\end{corollary}

\begin{proof}
    In Phase~1, the claim holds from \cref{lemma:parallel-matched-stays-matched}, and since only unmatched positions interview.  
\end{proof}

\begin{corollary}\label{lem:hib:before-to-no-propose-to-no-like}
    Suppose that the realization is not a $k$-bad realization.
    Then before time $t_0$, an unmatched applicant will not propose to any position that she does not interim like. This statement coincides with \cref{lem:before-to-no-propose-to-no-like}, {but for \cref{algorithm:parallel-bilateral-instantiation} rather than for \cref{algorithm:serial-adaptive}}.
\end{corollary}
\begin{proof}
    In Phase~1, the claim holds from \cref{lemma:parallel-matched-stays-matched}.
\end{proof}

\begin{observation}\label{obs:serial-before-transition}
    \cref{lemma:parallel-matched-stays-matched} coincides with \cref{lem:applicant-propose-to-unmatched}, \cref{lem:applicant-matched-to-liked} and \cref{cor:applicant-didnt-interview}, {but for \cref{algorithm:parallel-bilateral-instantiation} rather than for \cref{algorithm:serial-adaptive}}.
\end{observation}

\begin{lemma}\label{lemma:kr}
    Suppose that the realization is not a $k$-bad realization. The number of interviews (and consequently, rounds) in Phase~2 is at most {$k^2\cdot(k-1)$}.
\end{lemma}

\begin{proof}
    Since the realization is not $k$-bad, by \cref{cor:bad-k-rounds}, \cref{algorithm:parallel-bilateral-instantiation} never reaches Phase~3. By \cref{lem:hybrid-position-stops-interview,cor:hybrid-applicant-like-propose,lem:hib-minus-breaks+,lem:hib:before-to-no-propose-to-no-like,cor:hybrid-one-sided-propose-accept,cor:hybrid-parallel-matched-stays-matched}, all general lemmas established for the sequential \cref{algorithm:serial-adaptive} apply in this setting as well.
Furthermore, by \cref{obs:serial-before-transition}, every {claim stated} in the lemmas that applies before the transition point in the sequential algorithm {also holds} during Phase~1 of \cref{algorithm:parallel-bilateral-instantiation}. Since Phase~2 of \cref{algorithm:parallel-bilateral-instantiation} is exactly the execution of the sequential \cref{algorithm:serial-adaptive} after $t_0$, all lemmas that apply after the transition point in the sequential algorithm apply during Phase~2 as well.

In \cref{algorithm:parallel-bilateral-instantiation}, each interview in Phase~2 is interpreted as a single parallel interview round.
Therefore, by \cref{lem:interviews-after-t0}, the number of interview rounds in Phase~2 is at most $k^2\cdot(k-1)$.
\end{proof}
\paragraph{Proof of the interview rounds bound}
We are now ready to prove the bound on the expected number of interview rounds. We first bound the number of interview rounds in the case where Phase~2 is reached {(\cref{lem:rounds-log3n})}. Subsequently, we no longer rely on the lemmas from the sequential analysis, as we derive a bound on the number of interview rounds in executions where Phase~2 is never reached {(\cref{lem:phase1-termination,lem:rounds-k-log-is-log})}.
\begin{lemma}\label{lem:rounds-log3n}
Consider a bilaterally ex-ante equivalent {instance} with $n$ applicants and $ n +  \lceil 10 \log n \rceil -1>m \ge n$ positions. The expected number of interview rounds in \cref{algorithm:parallel-bilateral-instantiation} is $O(\log^3 n)$.
\end{lemma}
\begin{proof}
    For a given $n$, let $k=\lceil 10\log n\rceil$.  
    Let $E_G$ denote the event that a realization is not a $k$-bad realization, and let $\overline{E_G}$ denote the complementary event.  
    Let $Y$ be the random variable counting the number of interview rounds conducted by \cref{algorithm:parallel-bilateral-instantiation}.  
    By the law of total expectation,
    \begin{equation}\label{eq:1}
        \begin{split}
            \mathbb{E}[Y] 
            &= \mathbb{E}[Y \mid E_G]\Pr(E_G) + \mathbb{E}[Y \mid \overline{E_G}]\Pr(\overline{E_G}) \\
            &\le \mathbb{E}[Y \mid E_G] + \mathbb{E}[Y \mid \overline{E_G}]\Pr(\overline{E_G}).
        \end{split}
    \end{equation}

    Let $Y_1$ be the number of interview rounds before the transition point $t_0$ (Phase~1), and let $Y_2$ be the number of rounds after $t_0$ (Phase~2).  
    Since when a realization is not a $k$-bad realization, $Y = Y_1 + Y_2$ (\cref{cor:bad-k-rounds}), we have from \cref{lemma:kr} and \cref{cor:k-rounds}:
    \begin{equation}\label{eq:2}
        \mathbb{E}[Y \mid E_G] 
        = \mathbb{E}[Y_1 \mid E_G] + \mathbb{E}[Y_2 \mid E_G] 
        \le k + k^3.
    \end{equation}

    Since the number of interview rounds is always at most $n\cdot m$, it follows that $\mathbb{E}[Y \mid \overline{E_G}] \le n \cdot m$.
    
    Moreover, by \cref{lem:prob-k-bad}, $\Pr(\overline{E_G}) \le \frac{3}{n^3}$.  
    Using \cref{eq:1,eq:2}, we obtain:
    \begin{align*}
        \mathbb{E}[Y]
        &\le (k + k^3) + n\cdot m \cdot \frac{3}{n^3}= k + k^3 + \frac{3m}{n^2} \\
        &< k + k^3 + \frac{3(n +  \lceil 10 \log n \rceil -1)}{n^2} \\
        &= O(\log^3 n).
    \end{align*}
\end{proof}

\begin{lemma}
\label{lem:phase1-termination}
If the condition of $k$-eventual happiness for applicants holds throughout the execution of \cref{algorithm:parallel-bilateral-instantiation}, and assuming $m\ge n +k -1$, then Phase~1 terminates after at most $3+\log n$ interview rounds in expectation. 
\end{lemma}
\begin{proof}
By \cref{lemma:parallel-matched-stays-matched}, \cref{algorithm:parallel-bilateral-instantiation} terminates once every applicant has interviewed a position they interim like. Let $D$ denote the event that the $k$-eventual happiness condition holds for all applicants. Let $T_q$ denote the number of applicants interviewing in round $q$. For Phase~1 to continue to round $q$, there must be applicants who have not yet found a position they interim like. We initialize $T_1 = n$ (since $m \ge n+k-1$, all applicants participate initially). For $q\leq k$, let $X_{i,q}$ be the indicator variable that applicant $a_i$ does not interim like the position interviewed in round $q$. From the proof of \cref{lem:number-of-interviews-before-t0}, we know that $\Pr(Z_{i,q} \mid D) \ge \frac{1}{2}$, where $Z_{i,q} = 1 - X_{i,q}$ is the indicator that the applicant does interim like the position. Consequently, $\Pr(X_{i,q} \mid D) < \frac{1}{2}$. Note that if some applicant $i$ interviewed in iteration $q$, then $X_{i,q-1} =1.$ Therefore, the expected number of applicants interviewing in round $q$ is:
$$\mathbb{E}[T_q\mid D] = \sum_{i } \Pr(X_{i,q-1} \mid D) < \frac{1}{2} \mathbb{E}[T_{q-1}\mid D]$$
Since $T_1 = n$,
$\mathbb{E}[T_q\mid D] < \frac{n}{2^{q-1}}$ or every $q>1$.
Let $R$ be the random variable representing the total number of rounds in Phase~1. The phase continues to round $q$ only if there is at least one applicant interviewing in that round. Thus, the event $R \ge q$ implies $T_q \ge 1$. By the monotonicity of probability and Markov's inequality, we have:$$\Pr(R \ge q \mid D) \le \Pr(T_q \ge 1 \mid D) \le \frac{\mathbb{E}[T_q \mid D]}{1} = \mathbb{E}[T_q \mid D]$$Using the tail sum formula for expectation, we split the summation at the index $1+\lfloor \log n \rfloor$. For the first part of the sum, we bound the probability by $1$: 
$$\begin{aligned}
\mathbb{E}[R \mid D] &= \sum_{q=1}^\infty \Pr(R \ge q \mid D) \\
&= \sum_{q=1}^{1+\lfloor\log n\rfloor} \Pr(R \ge q \mid D) + \sum_{q=2+\lfloor\log n\rfloor}^\infty \Pr(R \ge q \mid D) \\
&\le \sum_{q=1}^{1+\lfloor\log n\rfloor} 1 + \sum_{q=2+\lfloor\log n\rfloor}^\infty \mathbb{E}[T_q \mid D]
\end{aligned}$$

Substituting the bound $\mathbb{E}[T_q \mid D] < \frac{n}{2^{q-1}}$, the first term sums to $1+\lfloor\log n\rfloor$, and the second term becomes a geometric series:$$\begin{aligned}
\mathbb{E}[R \mid D] &\le (1+\lfloor\log n\rfloor) + \sum_{q=2+\lfloor\log n\rfloor}^\infty \frac{n}{2^{q-1}} = (1+\lfloor\log n\rfloor) + n \cdot \sum_{j=1+\lfloor\log n\rfloor}^\infty \frac{1}{2^j}
\\&\le 1+\log n +  \frac{n}{2^{\lfloor\log n\rfloor}}
\le  1+\log n +2 = 3+ \log n
\end{aligned}$$

\end{proof}

\begin{lemma}\label{lem:rounds-k-log-is-log}
Consider a bilaterally ex-ante equivalent {instance} with $n$ applicants and $ m \ge n +  \lceil 10 \log n \rceil -1$ positions. The expected number of interview rounds in \cref{algorithm:parallel-bilateral-instantiation} is at most {$4+\log n$}.
\end{lemma}
\begin{proof}
    For a given $n$ and $m\ge n+\lceil10\log n\rceil-1$, let $k=m-n+1$.  
    Let $D$ be the probability that $k$-eventual-happiness for applicants holds and let $p_D$ be the probability for that event. Let $Y$ be a random variable counting the number of interview rounds conducted by \cref{algorithm:parallel-bilateral-instantiation}. By the law of total expectation: 
    \begin{equation}\label{eq:condition-on-k-happ-rounds}  
    \begin{split}
        \mathbb{E}\left[Y\right]&=\mathbb{E}\left[Y \mid D\right]\Pr(D)+\mathbb{E}\left[Y \mid \overline{D}\right]\Pr(\overline{D})
        {\le\mathbb{E}\left[Y \mid D\right]+\mathbb{E}\left[Y \mid \overline{D}\right]\Pr(\overline{D})}
    \end{split}
    \end{equation}
    Since $m\ge n+ k -1$, the algorithm stops before reaching Phase~2. Consequently, by \cref{lem:phase1-termination}, we have $\mathbb{E}\left[Y \mid D\right]\le 3+\log n$. Furthermore, if $k$-eventual-happiness does not hold, the number of interview rounds is at most $ n\cdot m$.
    Therefore, by \cref{eq:condition-on-k-happ-rounds}:
    \begin{equation}\label{eq:after-numbers-rounds}
        \begin{split}
            \mathbb{E}\left[Y\right] &\leq\mathbb{E}\left[Y \mid D\right] +\mathbb{E}\left[Y \mid \overline{D}\right]\Pr(\overline{D})
            \\&\leq 3+\log n+n\cdot m \cdot \Pr(\overline{D})
        \end{split}
    \end{equation}
    By \cref{cor:happiness-app-for-general-k}, it holds that $ \Pr(\overline{D})=1-p_D\le\frac{n}{2^k}$.
    Thus, by \cref{eq:after-numbers-rounds}, 
    \begin{equation}\label{eq:before-two-cases}
           \mathbb{E}\left[Y\right]\le 3+\log n+\frac{n\cdot n\cdot m}{2^k}.
    \end{equation}
    When $m\le n^7$, since $k\ge\lceil10\log n\rceil$,
    \begin{equation*}
        \begin{split}
        \mathbb{E}\left[Y\right]&\le 3+\log n+\frac{ n^9}{2^{\lceil 10\log n\rceil}}\le 3+\log n+\frac{n^9}{2^{ 10\log n}} \\
        &= 3+\log n+\frac{n^9}{n^{10}}= 3+\log n+\frac{1}{n} \\
        &\le 4+\log n
    \end{split}
    \end{equation*}
Else, since $n\ge2$ and $m>n^7$, it holds that $m>2^7$ and that $k=m-n+1> m-m^{\left(\frac{1}{7}\right)}$. Therefore, $ m^3\le 2^{m-m^{1/7}}$ and $\frac{n^2\cdot m}{2^k}<{\frac{ m^3}{2^{m-m^{1/7}}}<1}$. In total, by \cref{eq:before-two-cases}, $\mathbb{E}\left[Y\right] \leq 4 + \log n$.
\end{proof}

From \cref{lem:rounds-log3n,lem:rounds-k-log-is-log}, and since {$O(\log ^3 n)$ also upper bounds the expression $4+\log n$}, we get the following lemma immediately. 
\LogRounds*
\subsubsection{Proof of \cref{prop:2-interviews} (Bounding the Expected Number of Interviews)}\label{appendix:2-parallel}
To prove \cref{prop:2-interviews}, we next show that the claims stated in the lemmas in \cref{appendix:serial-all-iid} on the expected number of interviews of the sequential adaptive algorithm \cref{algorithm:serial-adaptive}, also hold for \cref{algorithm:parallel-bilateral-instantiation}. \begin{lemma}
    When $k$-eventual-happiness for applicants holds for the execution of \cref{algorithm:parallel-bilateral-instantiation}, the expected number of interviews conducted in Phase~1 is $2n$. This statement coincides with \cref{lem:number-of-interviews-before-t0}, { but for \cref{algorithm:parallel-bilateral-instantiation} rather than for \cref{algorithm:serial-adaptive}.}
\end{lemma}
\begin{proof}
    The proof follows immediately from \cref{lem:number-of-interviews-before-t0}, as the sequential lemmas correspond to lemmas of Phase~1.
\end{proof}
\begin{lemma}
    Consider any bilaterally ex-ante equivalent {instance} with  $n$ applicants and $n+\lceil10\log n\rceil-1>m\ge n$ positions.
    The expected number of interviews in \cref{algorithm:parallel-bilateral-instantiation} is  $2\cdot n+O(\log^3n)$. 
    This statement coincides with \cref{lem:2n-for-small-m}, but for \cref{algorithm:parallel-bilateral-instantiation} rather than for \cref{algorithm:serial-adaptive}.
\end{lemma}
\begin{proof}
    The proof follows immediately from \cref{lem:number-of-interviews-before-t0}, as the sequential lemmas correspond to lemmas of Phase~1.
\end{proof}
\begin{lemma}
    Consider any bilaterally ex-ante equivalent {instance} with  $n$ applicants and $m\ge n+\lceil10\log n\rceil-1$ positions.
    The expected number of interviews in \cref{algorithm:parallel-bilateral-instantiation} is {$2\cdot n+O(\log^3n)$}. This statement coincides with \cref{lem:2n-for-big-m}, {but for \cref{algorithm:parallel-bilateral-instantiation} rather than for \cref{algorithm:serial-adaptive}.} 
\end{lemma}
\begin{proof}
    The proof follows immediately from \cref{lem:2n-for-big-m}, as the sequential lemmas correspond to lemmas of Phase~1.
\end{proof}
{\begin{lemma}\label{lem:agent-k-interviews-parallel}
Consider any bilaterally ex-ante equivalent instance with $n$ applicants and $m= n$ positions.
In \cref{algorithm:parallel-bilateral-instantiation}, with probability $1 - O(n^{-3})$, every agent conducts at most $O(\log n)$ interviews. This statement coincides with \cref{lem:agent-k-interviews}, but for \cref{algorithm:parallel-bilateral-instantiation} rather than \cref{algorithm:serial-adaptive}.
\end{lemma}
\begin{proof}
   When $m=n$ the algorithm sets $k=\lceil10\log n\rceil$. Assume that the realization is not $k$-bad. In \cref{algorithm:parallel-bilateral-instantiation}, by \cref{cor:hybrid-parallel-matched-stays-matched}, a position stops interviewing once it meets an applicant they interim like each other. Under the not $k$-bad assumption, this occurs after at most $k$ interviews. 
  Moreover, by \cref{lem:hybrid-position-stops-interview,cor:hybrid-applicant-like-propose,lem:hib-minus-breaks+,lem:hib:before-to-no-propose-to-no-like,cor:hybrid-one-sided-propose-accept,cor:hybrid-parallel-matched-stays-matched}
  , when the realization is not $k$-bad, \cref{lem:applicant-do-max-k-interviews} applies for \cref{algorithm:parallel-bilateral-instantiation}. Consequently, every applicant will participate in at most  $k$ interviews. Therefore, when the realization is not $k$-bad, which holds a probability of at least $1-\frac{3}{n^3}$ for large enough $n$ (by \cref{lem:prob-k-bad}), every agent conducts at most $k$ interviews, which is $ O(\log n)$.
\end{proof}}
From all the lemmas above, we obtain the correctness of the following result.
\TwoInterviewsParallel*

\section{Proof of the Lower Bound from Section~3.4 (Minimizing the Number of Interviews)}\label{appendix:lower-bouds}

In this section we present the proof of \cref{thm:lower-bound-iid}. 

\LowerBoundIid*
\begin{proof}
{Let $c=\min\{m,n\}$, and let $\varepsilon=\frac{2}{c}>\frac{1}{c}$, and note that as $\varepsilon\cdot c=2$, the algorithm performs no more than $(2-\varepsilon)\cdot c$ interviews.} A sufficient condition for the matching to be non-interim stable is the existence of an applicant-position pair $(a_i, p_j)$ such that:
\begin{enumerate}
\item $a_i$ and $p_j$ have not interviewed each other;
\item $a_i$ has interviewed only once and did not interim like the position she interviewed with; and
\item $p_j$ has interviewed only once and did not interim like the applicant it interviewed.
\end{enumerate}
If these conditions hold, both agents strictly interim prefer each other to their final match (that is, $a_i \succ_{p_j}^Z \mu(p_j)$ and $p_j \succ_{a_i}^Z \mu(a_i)$). Consequently, $(a_i, p_j)$ forms a blocking pair.
Let $B$ denote the event that the matching is not interim stable, and let $C$ be the event that the sufficient condition described above occurs. Thus, $C \subseteq B$.
We assume that every matched agent participates 
in at least one interview (otherwise the matching is trivially interim unstable). We focus our analysis on the set of interviews performed by the $c$ matched agents on each side. Let $J$ be the random variable counting the number of applicants who, after their first interview, do not interim like the interviewed position. We observe that $\mathbb{E}[J] \geq \frac{c}{2}$. Consider the event where $J > (1-\varepsilon)\cdot c$. We claim that if this occurs, at least one of these applicants must have interviewed only once. 
To see this, suppose for the sake of contradiction that all the $J$ applicants interviewed at least twice, while the remaining $c-J$ applicants interviewed at least once. The total number of interviews would be at least $  2J + (c-J) = c + J  > c + (1-\varepsilon)\cdot c = (2-\varepsilon)\cdot c .$ This contradicts the lemma's assumption that the total number of interviews is at most $(2-\varepsilon)\cdot c$. Therefore, if $J > (1-\varepsilon)\cdot c$, there exists at least one applicant $a_i$ that interviewed only once and who does not interim like the position she interviewed with.

By symmetry, let $L$ be the random variable counting the number of positions who, after their first interview, do not interim like the interviewed applicant. Following the same logic, if $L > (1-\varepsilon)\cdot c+1$, there exists at least two positions that interviewed only once and who do not interim like the applicant they interviewed. Thus, at least one of them, denoted as $p_j$, did not interview with $a_i$, and together they form a blocking pair.
If both $\{J > (1-\varepsilon)\cdot c\}$ and $\{L > (1-\varepsilon)\cdot c + 1\}$ occur, we can guarantee a blocking pair. Notice that $J > (1-\varepsilon)\cdot c + 1$ implies $J > (1-\varepsilon)\cdot c$. {Moreover, since $\varepsilon>\frac{1}{c}$, it holds that $c>(1-\varepsilon)\cdot c+1$, and $J=c$ implies $J > (1-\varepsilon)\cdot c + 1$ and $L=c$ implies $L > (1-\varepsilon)\cdot c + 1$.} Thus

{$$ \{L =c\} \cap \{J =c\} \subseteq C \subseteq B $$}

Thus, the probability that the matching fails (is not interim stable) is at least the probability of the event $ \{L =c\} \cap \{J =c\}$.
The events $ \{L =c\}$ and $\{J =c\}$ are independent, since every matched agent necessarily had a first interview, and the result of that first interview is independent of the result of other agents' first interviews (as depending on any past history, every agent interim like the interviewed agent with probability of $\frac{1}{2}$). and thus:
\[
\begin{split}
    \Pr(B) &\geq \Pr(C) \\&\geq {\Pr\left(\{L =c\} \cap \{J =c\}\right)}\\&\ge\frac{1}{2^c}\cdot \frac{1}{2^c}=\frac{1}{2^{2c}}>0.
\end{split}\] 
\end{proof}

 \section{Deferred Proof and Example for Section \ref{sec:interview-system}}\label{appendix:interview-system}
\subsection{An Example Showing Necessity of the Condition in \cref{cor:opt-interviews}}
\label{appendix:example-counter}
\begin{example}
    Consider a bilaterally ex-ante equivalent {instance} with applicants $\mathcal{A} = \{a_1,\dots,a_5\}$ and positions $\mathcal{P} = \{p_1,\dots,p_5\}$. We assume that the prior distributions are uniform on $[0,1]$. Hence, the prior expected value for all pairs is $0.500$. The utilities from the applicants' and positions' perspectives following the execution of \cref{algorithm:serial-adaptive} are presented in \cref{tab:app-utilities,tab:pos-utilities}.

\begin{table}[h]
\caption{Applicants' utilities over positions {(prior expected values are highlighted)}.}
\label{tab:app-utilities}
\centering
\begin{tabular}{lccccc}
\toprule
 & $p_1$ & $p_2$ & $p_3$ & $p_4$ & $p_5$ \\
 \midrule
$a_1$ & 0.602 &  \textcolor{highlightcolor} {0.500} & \textcolor{highlightcolor} {0.500} & \textcolor{highlightcolor} {0.500} & 0.781 \\
$a_2$ & 0.791 & 0.121 & 0.040 & 0.371 & 0.303 \\
$a_3$ & 0.566 & 0.590 & \textcolor{highlightcolor} {0.500} & \textcolor{highlightcolor} {0.500} & 0.418 \\
$a_4$ & \textcolor{highlightcolor} {0.500} & \textcolor{highlightcolor} {0.500} & 0.894 & \textcolor{highlightcolor} {0.500} & \textcolor{highlightcolor} {0.500} \\
$a_5$ & 0.349 & 0.489 & \textcolor{highlightcolor} {0.500} & \textcolor{highlightcolor} {0.500} & 0.216 \\
\bottomrule
\end{tabular}
\end{table}

\begin{table}[h]
\caption{Positions' utilities over applicants {(prior expected values are highlighted)}.}
\label{tab:pos-utilities}
\centering
\begin{tabular}{lccccc}
\toprule
 & $a_1$ & $a_2$ & $a_3$ & $a_4$ & $a_5$ \\
 \midrule
$p_1$ & 0.409 & 0.142 & 0.900 & \textcolor{highlightcolor} {0.500} & 0.870 \\
$p_2$ & \textcolor{highlightcolor} {0.500} & 0.572 & 0.049 & \textcolor{highlightcolor} {0.500} & 0.122 \\
$p_3$ & \textcolor{highlightcolor} {0.500} & 0.153 & \textcolor{highlightcolor} {0.500} & 0.889 & \textcolor{highlightcolor} {0.500} \\
$p_4$ & \textcolor{highlightcolor} {0.500} & 0.991 & \textcolor{highlightcolor} {0.500} & \textcolor{highlightcolor} {0.500} & \textcolor{highlightcolor} {0.500} \\
$p_5$ & 0.475 & 0.486 & 0.573 & \textcolor{highlightcolor} {0.500} & 0.819 \\
\bottomrule
\end{tabular}
\end{table}

The interview algorithm outputs the matching:
    \[
    \mu = \{(a_1,p_5), (a_2,p_4), (a_3,p_1), (a_4,p_3), (a_5,p_2)\}.
    \]
    Note that applicant $a_5$ is matched to $p_2$ with a realized utility of $0.489$. Since $0.489 < 0.500$, $a_5$ does not interim like her match.

    When we run the applicant-proposing DA on these final utilities, the resulting applicant-optimal matching is:
    \[
     \mu^{\app} =\{(a_1,p_5), (a_2,p_1), (a_3,p_2), (a_4,p_3), (a_5,p_4)\}.
    \]
    This matching is not interim stable. Specifically, the pair $(a_5, p_4)$ is matched in $\mu^{\app}$, but $a_5$ never interviewed with $p_4$. The sequence of interviews is shown below. Here, $(a, p)$ denotes an interview between applicant $a$ and position $p$. The right-hand side indicates the tentative match formed after that sequence of interviews:
    \[
    \begin{aligned}
    & (a_1, p_1) &&  a_1 = \mu(p_1) \\
    & (a_2, p_2), (a_2, p_3), (a_2, p_4),(a_2, p_5), (a_2, p_1) &&  a_2 = \mu(p_4) \\
    & (a_3, p_2) &&  a_3 = \mu(p_2) \\
    & (a_4, p_3) &&  a_4 = \mu(p_3) \\
    & (a_5, p_5), (a_5, p_2), (a_5, p_1) && a_5 = \mu(p_2) \\
    & (a_3, p_5), (a_3, p_1) &&  a_3 = \mu(p_1) \\
    & (a_1, p_5) &&  a_1 = \mu(p_5)
    \end{aligned}
    \]
\end{example}
\subsection{Proof of \cref{prop:opt}}\label{appendix:interview-system-generic}
\GeneralDecouple*
\begin{proof}
    Let $\mu^{\app}$ denote the output of the applicant-proposing DA. A matching is defined as interim stable if (1) there are no blocking pairs with respect to the {final interim utility}, and (2) every matched pair has met for an interview. The first condition is satisfied since the DA algorithm guarantees stability with respect to the input preferences. We prove the second condition: for every $a_i$, the pair $(a_i, \mu^{\app}(a_i))$ has interviewed.
    First, observe that since $\mu$ is interim stable, it is stable with respect to the final interim utilities. Consequently, because $\mu^{\app}$ is the applicant-optimal stable matching, every applicant $a_i$ weakly prefers $\mu^{\app}(a_i)$ to $\mu(a_i)$.
    Now, suppose, towards a contradiction, that there exists an applicant $a_i$ such that the pair $(a_i,\mu^{\app}(a_i))$ did not interview. By assumption, $a_i$ strictly interim-prefers her match in $\mu$ to any position she did not interview, implying $\mu(a_i) \succ_{a_i}^Z \mu^{\app}(a_i)$, which contradicts the optimality of $\mu^{\app}(a_i)$. Therefore, $\mu^{\app}$ must consist only of interviewed pairs and is interim stable.
    Moreover, by transitivity, every applicant strictly prefers $\mu^{\app}(a_i)$ to any position she did not interview. Since the applicant-proposing DA proceeds in decreasing order of preference and terminates at $\mu^{\app}$, no applicant ever reaches the point of proposing to an uninterviewed position. Therefore, it suffices to execute the algorithm using preference lists restricted only to interviewed positions.
\end{proof}
\subsection{Proof of \cref{cor:opt-interviews}}\label{appendix:interview-system-cor}
\OptIsStable*
\begin{proof}
    By the definition of interim liking  when applicants view positions as ex-ante equivalent, every applicant $a_i$ interim strictly prefers $\mu(a_i)$ to any position she did not interview. The claim therefore follows directly from \cref{prop:opt}.
\end{proof}

\subsection{Proof of \cref{lem:no-extra-interviews}}\label{appendix:interview-system-proof}
\Decouple*
\begin{proof}
Assume that $ n + \lceil10\log n\rceil - 1> m$, and let $k=\lceil10\log n\rceil$ and assume that the realization is not $k$-bad. First, in \cref{algorithm:serial-adaptive}, by \cref{lem:applicant-matched-to-liked}, before time $t_0$ an applicant never proposes to a position she does not interim like. Moreover, when the realization is not $k$-bad, \cref{lem:new-dont-propose-to++} applies. Consequently, after $t_0$ no applicant proposes to a position she does not interim like.

We next consider \cref{algorithm:parallel-bilateral-instantiation}.
When the realization is not $k$-bad, \cref{cor:bad-k-rounds} implies that
the algorithm consists only of Phase~1 and Phase~2.
By \cref{lemma:parallel-matched-stays-matched}, during Phase~1 no applicant
is matched to a position she does not interim like.
Since all the {claims stated in the lemmas about the sequential algorithm also hold} for the hybrid algorithm
under a non-$k$-bad realization, the same conclusion applies to Phase~2 as
well: no applicant is matched to a position she does not interim like.

Therefore, when the realization is not $k$-bad (which holds a probability of at least $1-\frac{3}{n^3}$ for large enough $n$ by \cref{lem:prob-k-bad}) an
applicant never proposes to, and hence is never matched with, a position
she does not interim like.

Finally, consider the case $ m\ge n + \lceil10\log n\rceil - 1$. Assume that $k$-eventual happiness for applicants holds for $k=m-n+1\ge\lceil10\log n\rceil$. {By \cref{cor:k-happ-better-than-good} }this event occurs with a probability of {at least $1-\frac{1}{n^{9}}>1-\frac{3}{n^3}$.} Conditioned on this event, neither algorithm proceeds beyond
time $t_0$ (In \cref{algorithm:serial-adaptive} it follows from \cref{lem:applicant-matched-to-liked}, and in \cref{algorithm:parallel-bilateral-instantiation} since $\app'=\app$, it follows from \cref{cor:happy-k-rounds}).
By \cref{lem:applicant-matched-to-liked,lemma:parallel-matched-stays-matched},
it follows that no applicant is matched to a position she does not interim
like in this case either.

Overall, with a probability of at least $1-\frac{3}{n^3}$, both
\cref{algorithm:serial-adaptive,algorithm:parallel-bilateral-instantiation}
terminate with a matching in which every applicant interim likes her
assigned position. Together with \cref{cor:opt-interviews}, this completes the proof.
\end{proof}

\section{Simulations}\label{sec:simulations}
\subsection{Simulations of the Sequential Algorithm (\cref{algorithm:serial-adaptive})}\label{sec:serial-simulations}
\subsubsection{Simulation of the Bilaterally Ex-Ante Equivalent Setting}\label{appendix:simulations-iid-sequential}
{While the theoretical bound established in \cref{thm:2n} is asymptotic and applies only for sufficiently large $n$, our simulations show that the expected number of interviews per applicant is close to 2, even for very small markets (\cref{fig:plot1}).}
\begin{figure}[h]
    \centering
    \includegraphics[width=1\linewidth]{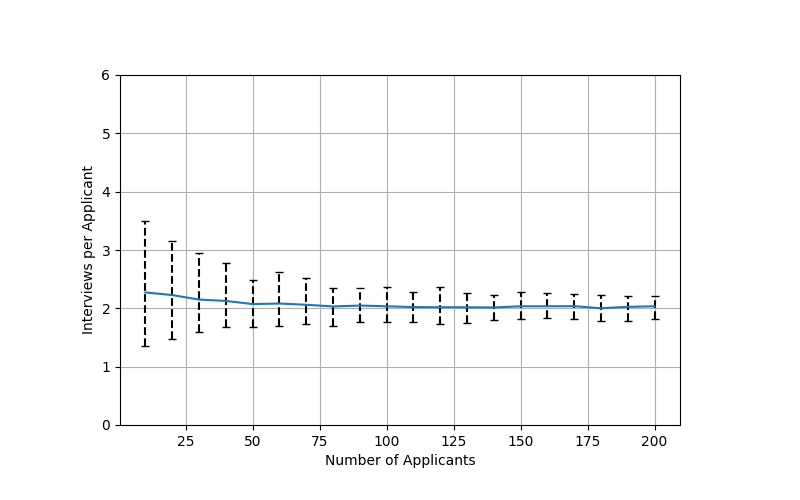}
    \caption{Simulation of \cref{algorithm:serial-adaptive}: the average number of interviews per applicant in a bilaterally ex-ante equivalent instance with $n$ applicants and $n$ positions. For all $i,j \in [n]$, the match values are drawn independently from $F_{i,j}=G_{j,i}=\mathcal{U}[0,1]$. Results are averaged over 100 trials{, with error bars indicating 95\% confidence intervals}. We note that the simulation outcomes are very similar when the distributions are replaced by distributions supported on two points. 
    }

    \label{fig:plot1}
\end{figure}
{\subsubsection{Simulation of a Two-Sided Ordered Setting}\label{appendix:simulations-exp-sequential}

Our simulations indicate that the expected number of interviews per applicant obtained by \cref{algorithm:serial-adaptive} remains constant even in a setting more general than the bilaterally ex-ante equivalent setting. Specifically, when applicants ex-ante agree on the order of positions, and positions ex-ante agree on the order of applicants. 
To illustrate this, \cref{fig:simulation-exp} presents the average number of interviews per applicant for a simulation of the setting of \cref{ex:exp-dist}, as a function of $n$. We observe that the average number of interviews per applicant is not much more than 4.
We present the construction used in the simulation presented in \cref{fig:simulation-exp} and the illustration of it in \cref{fig:dist_G}.\begin{restatable}{example}{ExpDist}\label{ex:exp-dist}
Given $n$ and $m$ we construct the value distributions as follows.  
For every applicant $i\in [n]$ and position $j\in [m]$, define

\begin{equation}\label{eq:f_ij}
    F_{i,j} =
\begin{cases}
2^{\,m-(j-1)} - \dfrac{j}{m+1}, & \text{with probability } \tfrac{1}{2}, \\[1.2ex]
\dfrac{j}{m+1}, & \text{with probability } \tfrac{1}{2}.
\end{cases}
\end{equation}

The expected value of $F_{i,j}$ is
$V_{i,j} = 2^{\,m-j}$, which is strictly decreasing in $j$. Hence, applicants (ex-ante) agree on the order of positions.

Symmetrically, for every position $j\in [m]$ and applicant $i\in [n]$ we define

\begin{equation}\label{eq:g_ij}
    G_{j,i} =
\begin{cases}
2^{\,n-(i-1)} - \dfrac{i}{n+1}, & \text{with probability } \tfrac{1}{2}, \\[1.2ex]
\dfrac{i}{n+1}, & \text{with probability } \tfrac{1}{2}.
\end{cases}
\end{equation}

The expected value of $G_{j,i}$ is
$
U_{j,i} = 2^{\,n-i} ,
$
which is strictly decreasing in $i$. Hence, positions (ex-ante) agree on the order of applicants.  \footnote{We remark that as all positions have the same prior over applicants, ex-ante agreement on an order necessarily holds (and similarly for applicants over positions). 
}
\end{restatable}

{\begin{remark} 
    The construction in \cref{ex:exp-dist} uses two-point distributions for clarity. The example can be extended to continuous distributions by replacing each point mass by a uniform distribution on a sufficiently small interval around that value.
\end{remark}}

\begin{figure}[h]
    \centering
    \includegraphics[width=0.7\textwidth]{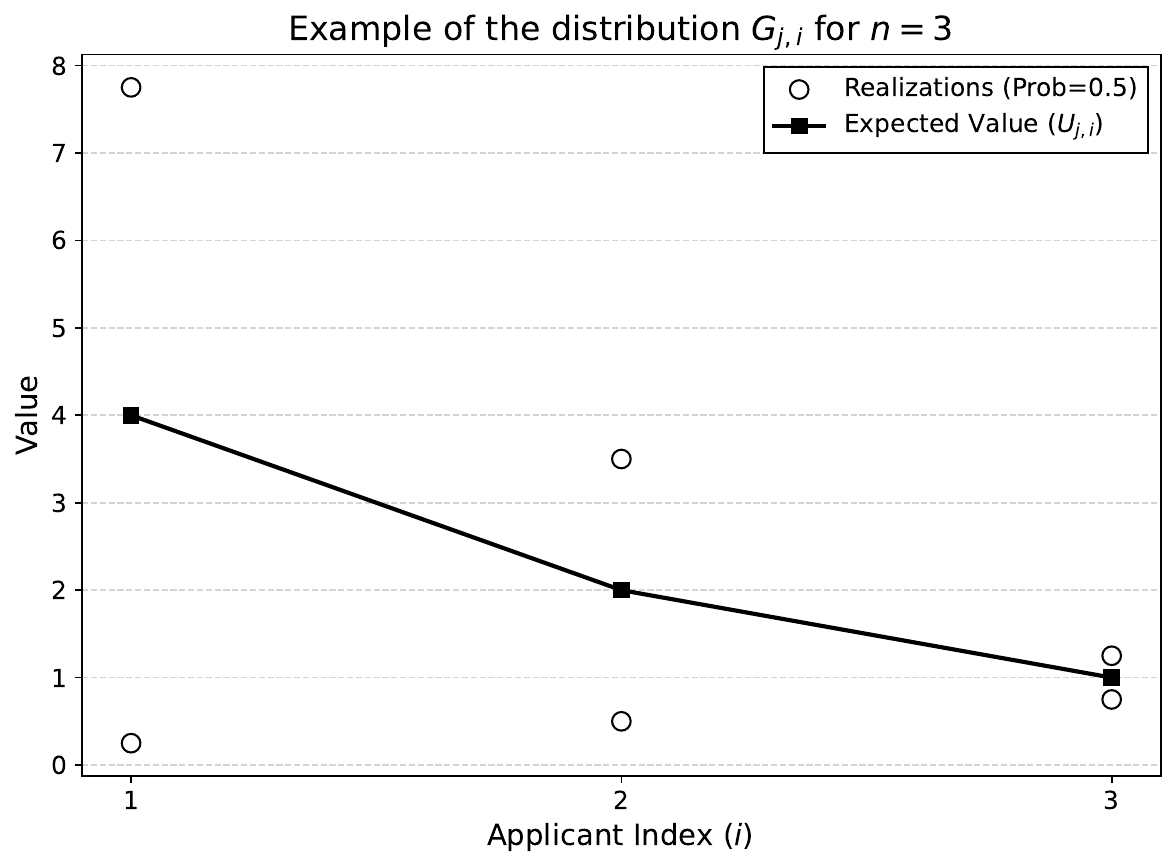}
    \caption{Illustration of the value distributions $G_{j,i}$ for $n=3$ described in \cref{ex:exp-dist}. Notice that the lower realizations increase linearly while the upper realizations decrease exponentially.
    Two nice properties are visible: 
    (1) For every $i' > i$, the entire support of $G_{j,i'}$ lies strictly below the expected value $U_{j,i}$. 
    (2) For every $i'$, the lower realization $\frac{i{'}}{n+1}$ is strictly below the expected value $U_{j,i'}$.}
    \label{fig:dist_G}
\end{figure}
\begin{figure}[h]
    \centering
    \includegraphics[width=1\linewidth]{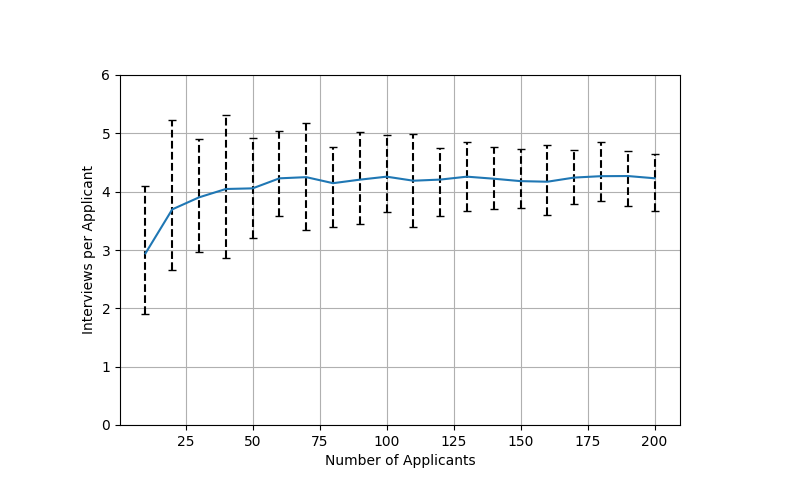}
    \caption{Simulation results for \cref{algorithm:serial-adaptive} with $n$ applicants and $n$ positions. Here, agents on both sides ex-ante agree on the ranking of the other side. Match values for all $i,j \in [n]$ are drawn independently from the distribution described in \cref{ex:exp-dist}. The plot shows the average number of interviews per applicant. {Results are averaged over 100 trials, with error bars indicating 95\% confidence intervals}.}
    \label{fig:simulation-exp}
\end{figure}

{
\subsection{Simulation of a Parallel Algorithm in the Bilaterally Ex-Ante Equivalent Setting}\label{appendix:simultation-prallel} 
Our hybrid algorithm was constructed to correspond to the analysis of \cref{algorithm:serial-adaptive} in the bilaterally ex-ante setting. However, we conjecture that a fully parallel algorithm---which attempts in every round to assign an interview to any unmatched applicant that needs to interview---has expected number of interview rounds which is $O(\log n)$. In \cref{fig:parallel}, we present simulation results for such an algorithm, which prioritizes interviews with unmatched positions and positions that are tentatively matched to applicants they do not interim like. It indicates that the average number of interview rounds grows at the order of $O(\log n)$. 
\begin{figure}[h]
    \centering
    \includegraphics[width=1\linewidth]{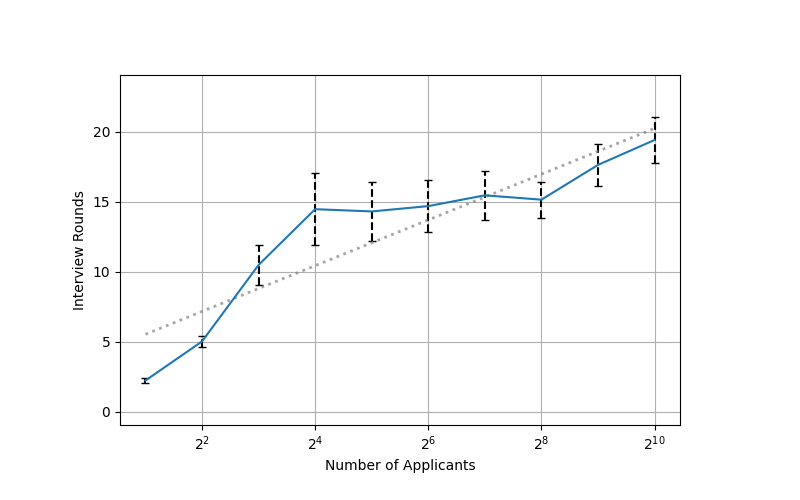}
    \caption{Simulation of a fully parallel algorithm showing the number of interview rounds. The instance is bilaterally ex-ante equivalent with $n$ applicants and $n$ positions. For all $i,j \in [n]$, the match values are drawn independently from $F_{i,j}=G_{j,i}=\mathcal{U}[0,1]$. Results are averaged over 100 trials{, with error bars indicating 95\% confidence intervals}. The dashed line denotes the least-squares linear fit. The horizontal axis is plotted on a {logarithmic} scale, consequently, linear growth in the plot corresponds to logarithmic growth in $n$.}
    \label{fig:parallel}
\end{figure} \section{Extensions }\label{sec:extensions}
\subsection{Ex-Ante Almost Equivalence}\label{appendix:extensions-almost}

{We proved our main results under the assumption that agents are ex-ante equivalent. In this subsection, we show that these results extend to settings in which agents are only ex-ante \emph{almost} equivalent.} This, in particular, includes settings that all agents are ex-ante equivalent (same expected value) yet each agent breaks ties arbitrarily (possibly in a different way).

Fix two small parameters $\varphi_U>0$ and $\varphi_L>0$. 
We say that \emph{applicants view positions as ex-ante almost equivalent} 
if there exist two values $V_H$ and $V_L$ such that for all $i,j\in[n]$, 
\[
V_H > V_{i,j} > V_L,
\qquad
\Pr_{v \sim F_{i,j}}[\,v > V_H\,] = \tfrac12 - \varphi_U,
\qquad
\Pr_{v \sim F_{i,j}}[\,v < V_L\,] = \tfrac12 - \varphi_L.
\]
Similarly, we say that \emph{positions view applicants as ex-ante almost equivalent} 
if there exist values $U_H$ and $U_L$ such that for all {$i \in [n], j \in [m]$},
\[
U_H > U_{j,i} > U_L,
\qquad
\Pr_{u \sim G_{j,i}}[\,u > U_H\,] = \tfrac12 - \varphi_U,
\qquad
\Pr_{u \sim G_{j,i}}[\,u < U_L\,] = \tfrac12 - \varphi_L.
\]
We say that a setting is \emph{bilaterally ex-ante almost equivalent}   if applicants view positions as ex-ante almost equivalent and positions view applicants as ex-ante almost equivalent.
For the ex-ante almost equivalent setting (where the upper thresholds $V_H$ and $U_H$ are defined),
we use the following notion.
\begin{definition} [interim strong like]
    Fix an applicant $a_i$ and a position $p_j$, and assume that $(a_i,p_j)\in Z$, i.e., they have already met for an interview. We say that \emph{$a_i$ interim strong likes $p_j$} if $\vutility{i}{j}{}>V_H$. Similarly, we say that \emph{$p_j$ interim strong likes $a_i$} if it holds that $\uutility{j}{i}{}> U_H$. We say that  $a_i$ and $p_j$ \emph{interim strong like each other}, if \emph{$a_i$ interim strong likes $p_j$} and \emph{$p_j$ interim strong likes $a_i$}.
\end{definition}

First, we obtain the following result, which holds for \cref{algorithm:serial-adaptive} as originally stated.
\begin{theorem}
    Consider an instance in which $m$ positions view $n$ applicants as ex-ante almost equivalent, for sufficiently small $\varphi_U>0$ and $\varphi_L>0$.  Then
\cref{algorithm:serial-adaptive} 
terminates with an interim-stable matching, and its adaptive interviewing process performs at most $4m\cdot \frac{1}{1-2\varphi_U}$ interviews in expectation. 
\end{theorem}

When the setting is bilaterally ex-ante almost equivalent, we may modify line 4 in \cref{algorithm:serial-adaptive} as follows:
``Let $\pos^*$ denote the set of $a_{i^*}$’s most preferred positions from which she has not yet been rejected. All positions that $a_{i^*}$ has not yet interviewed with are treated as indifferent." Specifically, for every such $j$ with $V_{i^*,j}$, we temporarily assign value $V_H$ for the purpose of this step, however, ties are broken according to the original ex-ante expectations rather than lexicographically.
With the corresponding change in Subroutine~\ref{algorithm:pick-interviews-bilateral-equivalent}, the following equivalent theorems hold.

\begin{theorem}
    There exists an algorithm (\cref{algorithm:serial-adaptive}) that for any instance terminates with a matching that is interim stable, and in a bilaterally ex-ante almost equivalent setting, for sufficiently small $\varphi_U>0$ and $\varphi_L>0$, with $n$ applicants and $m\ge n$ positions, the expected number of interviews is
    $2n\cdot \frac{1}{1-2\varphi_U}+O(\log^3n)$.  
\end{theorem}

\begin{theorem}
       There exists an algorithm that conducts interviews in parallel 
     (\cref{algorithm:parallel-bilateral-instantiation}) and always terminates with an interim-stable matching, and in a bilaterally ex-ante almost equivalent setting, for sufficiently small $\varphi_U>0$ and $\varphi_L>0$, with $n$ applicants and $m\ge n$ positions, has the following guarantee:
     its expected number of interview rounds is $O(\log^3 n)$  and its expected number of interviews is $2n\cdot \frac{1}{1-2\varphi_U}+O(\log^3n)$. Additionally, when $m\ge n +  \lceil 10 \log n \rceil $ its  expected number of interview rounds is $O(\log n)$.
\end{theorem}  \end{document}